\documentclass[aps,nobibnotes,twocolumn,superscriptaddress,showkeys,thelongbibliography]{revtex4-2}

\usepackage{dcolumn}

\usepackage{etoolbox}

\usepackage{graphicx}
\usepackage{enumerate}
\usepackage{amsmath}
\usepackage{amssymb}  
\usepackage{mathtools}
\usepackage{bm}  
\usepackage{pdfsync}  
\usepackage{hyperref}
\usepackage{physics}
\usepackage{natbib}
\usepackage{comment}

\usepackage{amsthm}
\newtheorem{theorem}{Theorem}[section]
\newtheorem{lemma}[theorem]{Lemma}

\newtheorem{corollary}[theorem]{Corollary}

\begin{document}

\title{Notions of Adiabatic Drift in the Quantized Harper model}
\author{Alice C. Quillen}
\email{aquillen@ur.rochester.edu}
\affiliation{Department of Physics and Astronomy, University of Rochester,  Rochester, NY, 14627, USA}
\author{Nathan Skerrett} 
\affiliation{Department of Mathematics, University of Rochester,  Rochester, NY, 14627, USA}
\author{Damian R. Sowinski} 
\affiliation{Department of Physics and Astronomy, University of Rochester,  Rochester, NY, 14627, USA}
\author{Abobakar Sediq Miakhel} 
\affiliation{Department of Physics and Astronomy, University of Rochester,  Rochester, NY, 14627, USA}
\date{\today}

\begin{abstract}

We study a quantized, discrete and drifting version of the Harper Hamiltonian, also called the finite almost Mathieu operator, which resembles the pendulum Hamiltonian but in phase space is confined to a torus.  Spacing between pairs of eigenvalues of the operator spans many orders of magnitude,  with nearly degenerate pairs of states at energies that are associated with circulating orbits in the associated classical system.  When parameters of the system slowly vary, both adiabatic and  diabatic transitions can take place at drift rates that span many orders of magnitude.  Only under an extremely negligible drift rate would all transitions into superposition states be suppressed.  The wide range of energy level spacings could be a common property of quantum systems with non-local potentials that are related to resonant classical dynamical systems.  Notions for adiabatic drift are discussed for quantum systems that are associated with classical ones with divided phase space.  
 
\end{abstract}

\maketitle
\tableofcontents

\section{Introduction} \label{sec:intro}


In this study we contrast and explore notions of adiabatic behavior for a drifting classical and associated quantum Hamiltonian system on a torus.  
An advantage of studying a quantized system on the torus is that
it is finite dimensional, facilitating numerical calculations and 
potential applications in quantum computing.   
Finite dimensional quantum spaces are relevant for evolution of spin systems (e.g., \citealt{Bossion_2022})
that are studied in the context of chemistry and nuclear physics. 
The ability to estimate whether a system 
behaves adiabatically is relevant for implementing effective adiabatic quantum computation algorithms \citep{Albash_2018}  which can be used to solve satisfiability and other combinatorial search problems and some optimization problems  \citep{Santoro_2006}.   The system we study, the quantized Harper model,  
resembles the pendulum dynamical system and that model describes superconducting transmons \citep{Koch_2007,Blais_2021,Cohen_2023}, so 
our study is relevant for control of quantum computers that leverage these devices \citep{Champion_2024}. 
 
Complex classical dynamical systems, including multiple planet systems \citep{Quillen_2011} and 
particles in a plasma \citep{Chirikov_1979}, can exhibit resonant behavior that can 
be described with Hamiltonian models that resemble the pendulum.   
The classical pendulum Hamiltonian 
\begin{align}
H_{\rm pend}(p,\phi) = \frac{p^2}{2} - \epsilon \cos \phi  \label{eqn:pend}
\end{align}
where momentum $p$ and angle $\phi$ are canonical coordinates that are functions of time $t$ and 
constant parameter $\epsilon$ describing the resonance strength sets the oscillation frequency at the bottom of the cosine potential well.   
This system  
exhibits two types of dynamical behavior: {\it libration}, where the angle $\phi(t)$ oscillates   
about zero, and {\it circulation}, where $\phi$ continuously increases or continuously decreases in time. In phase space, the two types of behavior are separated by an orbit, known as the separatrix, that has an infinite period
and contains the hyperbolic fixed point located at $(p,\phi)=(0,\pi)$.   

A more general version of 
the pendulum model of equation \ref{eqn:pend}  depends upon 
three parameters $a, b, \epsilon$; 
\begin{align}
H_{\rm pend}(p,\phi)' = a\frac{(p-b)^2}{2} - \epsilon \cos \phi  \label{eqn:pend2}. 
\end{align}
The parameter $b$ shifts the locations of the fixed points. 
We consider the situation where the parameters
$a,b,\epsilon$ are slowly varying in time.   An example setting is the dynamics of 
a migrating moon or planet \citep{Borderies_1984}.  
In celestial mechanics, the unperturbed system is expanded to second order in 
an action variable (giving the quadratic kinetic energy term).  A gravitational perturbations from a planet 
is expanded in a Fourier series (giving $\phi$ dependent terms). 
A motivation for studying a system resembling the pendulum is that 
the Hamiltonian of a transmon in a resonator circuit QED setup takes the form of the 
Hamiltonian in equation \ref{eqn:pend2}  
with momentum $p = n$ equal to the charge operator,  the angle $\phi$ equal to the phase operator, 
the coefficient $a= 4E_C$ with $E_C$ equal to the charging energy, 
the coefficient $\epsilon = E_J$, the Josephson energy and $b= n_g$  an offset charge
\citep{Koch_2007,Blais_2021}.   

%
In classical Hamiltonian systems, Liouville's theorem implies that 
adiabatic drift is associated with near conservation of 
an action variable which depends upon a contour or surface in phase space that  
 encloses a constant phase space volume.  
However,  a particle that nears a separatrix orbit as the system varies must enter a non-adiabatic 
dynamical regime because the period of the separatrix is infinite due to the presence of a hyperbolic fixed point. 
Overcoming this difficulty, 
the Kruskal-Neishtadt-Henrard (KNH) theorem  \citep{Neishtadt_1975,Henrard_1982} 
relates the probability of transition between different phase space regions 
 to the rates that the enclosed phase space volumes vary. 
In celestial mechanics this process is called {\it resonance capture} (e.g., \cite{Borderies_1984,Yoder_1979}). 
The capture probability computed via the Kruskal-Neishtadt-Henrard theorem is accurate only 
 if the drift rate is slower than a dimensionally derived function of the resonance libration frequency 
 which also describes the rate that orbits diverge from unstable fixed points \citep{Quillen_2006}.   
Even though dynamics near the separatrix is not actually adiabatic, resonance 
capture is often described as an adiabatic process as it only occurs at a slow drift rate.  Classical 
Hamiltonian dynamics is deterministic, rather than probabilistic.  However the sensitivity of 
the dynamics near a hyperbolic fixed point contained in the separatrix orbit means that the asymptotic 
behavior of a distribution of orbits that crosses a separatrix is well described with a probability. 

Notions of adiabatic drift differ between classical and quantum systems. 
In a quantum system, adiabatic variation is often used to describe a slowly varying Hamiltonian operator
that drifts sufficiently slowly that a system initially in an eigenstate of the operator remains in an 
eigenstate of the Hamiltonian operator \citep{Born_1928}.    This principle is known as the {\it adiabatic theorem}. 
When the system drifts fast enough that a transition from one
eigenstate to another takes place, the transition is called {\it diabatic}. 
The Landau-Zener model \citep{Zener_1932,Wittig_2005,Watanabe_2021} for a drifting 2-state quantum system gives an expression for the probability of a diabatic transition as a function of the drift rate and the minimum energy difference between the two states at their closest approach.    

Using Bohr-Sommerfeld quantization and the WKBJ semi-classical limit, \citet{Stabel_2022}
showed that a quantum version of the Kruskal-Neishtadt-Henrard theorem holds for a drifting quantized double well potential. 
In this system, as the Hamiltonian operator slowly varies, there is a lattice of 
 of avoided energy level crossings in the vicinity of the separatrix of the associated classical system.  
Each energy level crossing can 
be approximated via the two-level Landau-Zener model, giving a connection between the probabilities of transition 
between quantum states and the growth rates of phase space areas \citep{Stabel_2022}.
 
We build upon the work by \citet{Stabel_2022}, however,  instead of a varying double well potential 
and using a semi-classical limit,
we examine a quantized system on the torus that resembles the pendulum, known as the Harper model,  
which is equivalent to a discrete version of the Mathieu operator (arising from
Schr\"odinger's equations for the quantum pendulum) called the finite almost Mathieu operator 
\citep{Strohmer_2021}.   An advantage of using a discrete or finite dimensional operator is that
we can explicitly calculate the energy differences between eigenstates.  
The classical Hamiltonian of the Harper model \citep{Harper_1955}, describing the motion
of electrons in a 2-dimensional lattice in the presence of a magnetic field,  is  
\begin{align}
H_{\rm Harper}(p,\phi) = a\cos p + \epsilon\cos \phi, 
\end{align}
with canonical coordinates $p, \phi \in [-\pi,\pi]$ in a doubly periodic domain and real coefficients $a, \epsilon$. 
For small momentum $p$,  a constant minus the kinetic term is $1-\cos p \approx p^2/2$ which 
resembles  kinetic energy, so 
at small $p$, the Harper model exhibits dynamics similar to the pendulum. 
An advantage of studying the Harper Hamiltonian is that it can be quantized in a finite dimensional 
complex vector space.   For this model, because it is finite dimensional and can be simply written 
in terms of clock and shift operators (see appendix \ref{ap:defs}), it is straight-forward to compute its eigenvalues to high precision and calculate the distance between eigenvalues 
during avoided crossings.  For the Harper model, quantization via Bohr-Sommerfeld quantization, 
by using two mutually unbiased bases related via Fourier transform to describe the operators $\hat p, \hat \phi$,  or with a point operator constructed from discrete coherent state analogs to carry out
Wigner-Weyl quantization is discussed by \citet{Quillen_2025}. 

As we did in equation \ref{eqn:pend2} for the pendulum, 
we modify the Harper model with an additional parameter $b$ so that the libration regions 
 can be shifted, giving 
\begin{align}
H_{\rm Harper}(p,\phi)' = a\cos (p-b) + \epsilon \cos \phi. \label{eqn:Harp}
\end{align} 
Figure \ref{fig:Harp} illustrates level curves of this Hamiltonian. The two panels
in this figure show how 
the center of the librating region is shifted by the parameter $b$.  The parameter $\epsilon$ controls the width of the librating regions.  If $|\epsilon| <|a|$ and $a,\epsilon \ne 0$, phase space is divided into three regions, 
two with orbits that librate about $\phi = 0$ or $\pi$, and a circulating region where 
orbits cover all possible values of $\phi$.   The energies of the separatrix orbits can be found by 
identifying the hyperbolic fixed points.  
The separatrix orbits have  
\begin{align}
E_{sep} = \pm (|a| - |\epsilon|)  \label{eqn:Esep}
\end{align}
and there are two of them if $|a| \ne |\epsilon|$ and $a, \epsilon \ne 0$.   
If $|a| = |\epsilon| \ne 0 $, the circulating region vanishes and there is only one separatrix energy which separates two librating regions. 

\begin{figure}[htbp]\centering
\includegraphics[width=1.67 truein,trim = {0 0 0 0}, clip]{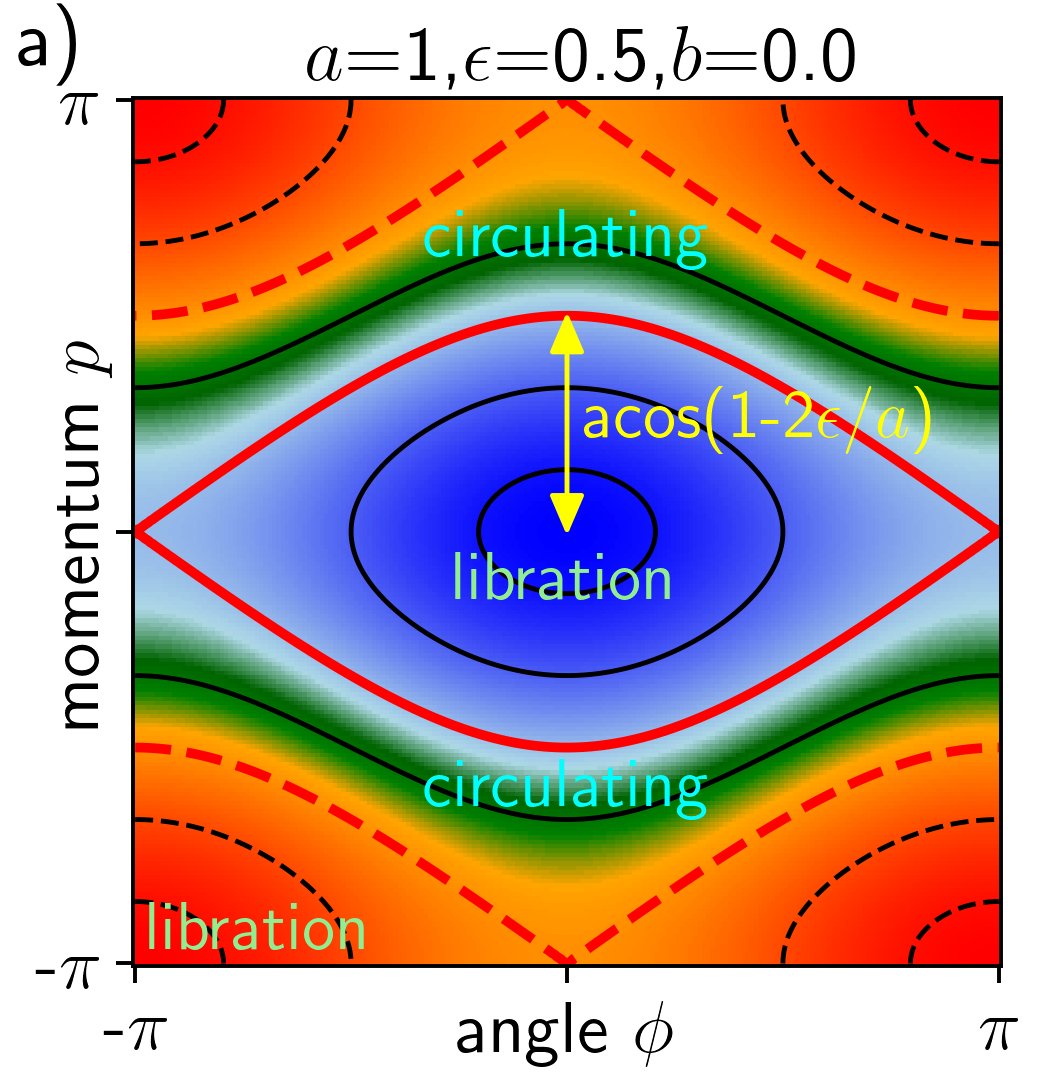}
\includegraphics[width=1.67 truein,trim = {0 0 0 0}, clip]{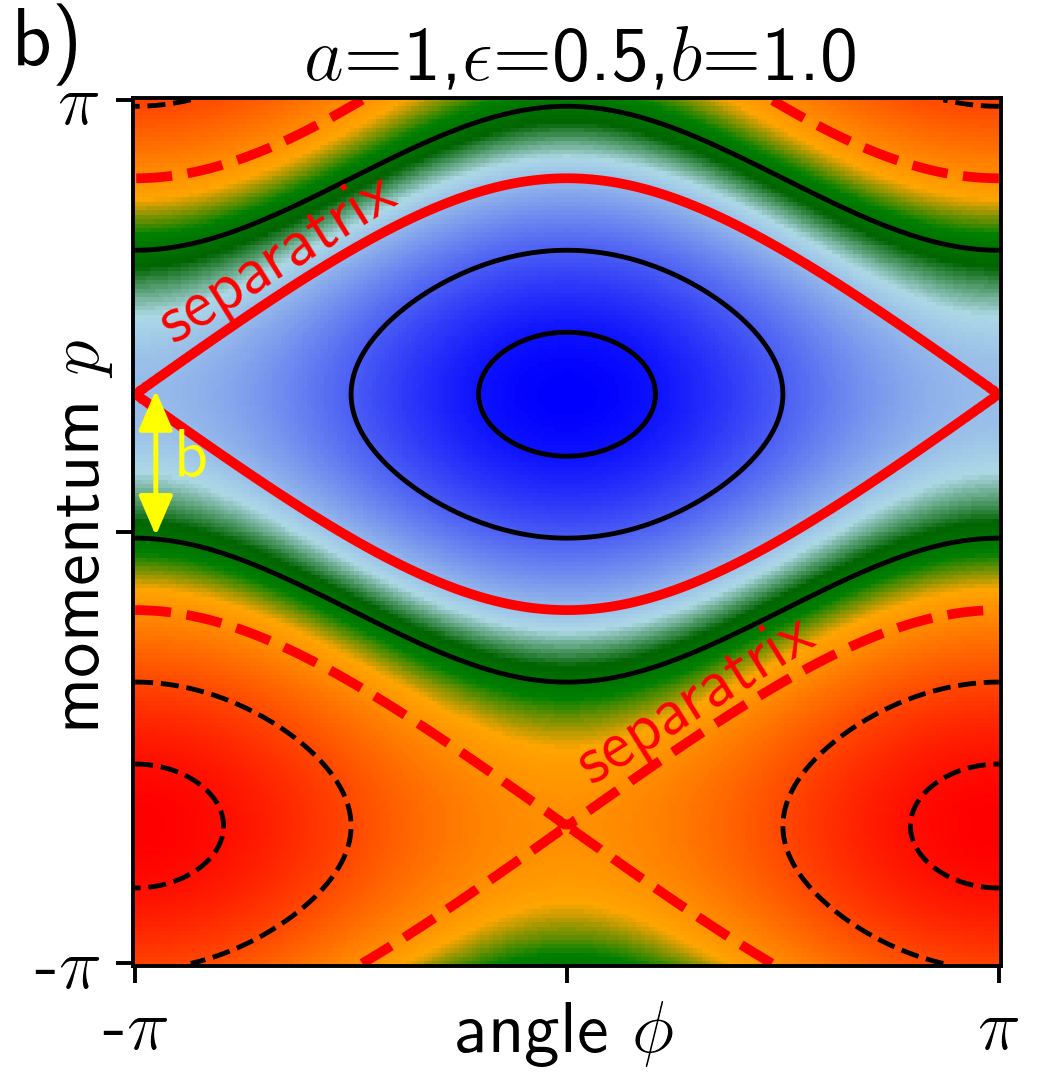}
\caption{Level curves in phase space of the classical Harper Hamiltonian of  
 equation \ref{eqn:Harp}. Thick red level curves show the two separatrix orbits.  
Negative energies are shown with dotted contours.  
A comparison between left and right panels illustrates how the parameter $b$ shifts the center of the librating regions. 
Parameters for the model are shown on the top of each panel. Librating and circulating regions are annotated on the left panel.  The parameter $\epsilon$ sets the width of the librating region. 
\label{fig:Harp}}
\end{figure}

For the quantized version of the Harper model, we work in an $N$ dimensional complex vector space that has an inner product. 
The quantized version of equation \ref{eqn:Harp} is the Hermitian operator 
\begin{align}
\hat h = a \cos (\hat p -b) + \epsilon \cos \hat \phi  \label{eqn:h0}
\end{align}
with operators $\hat p$ and $\hat \phi$ defined in section \ref{ap:defs} (also see \cite{Quillen_2025}). 

\section{The eigenvalues of the Harper operator}

In this section we discuss how the set of eigenvalues or spectrum of the operator 
of equation \ref{eqn:h0}  depends on the parameters $b,\epsilon$.  
For $a=1$, $b=0$,  and $\epsilon=0$,  the Hamiltonian operator is equal to $ \cos \hat p$ 
so its eigenvalues are $\cos \left(\frac{2 \pi j}{N}\right) $ with index $j \in \{0, 1, .... , N-1 \} $. 
For most values of index $j$, eigenvalues have multiplicity 2. 
With dimension $N$ odd there is a single non-degenerate eigenvalue but for $N$ even there are two. 
In Figure \ref{fig:cc1} for different values of $\epsilon$ we compute and plot the spectrum of $\hat h$. 
For $\epsilon>0$, the eigenvalue degeneracy is broken, as can be seen on the left side of Figure \ref{fig:cc1} 
where pairs of eigenvalues split as $\epsilon$ increases. 

\begin{figure}[htbp] \centering 
\includegraphics[width=3truein,trim = 0 0 30 25, clip]{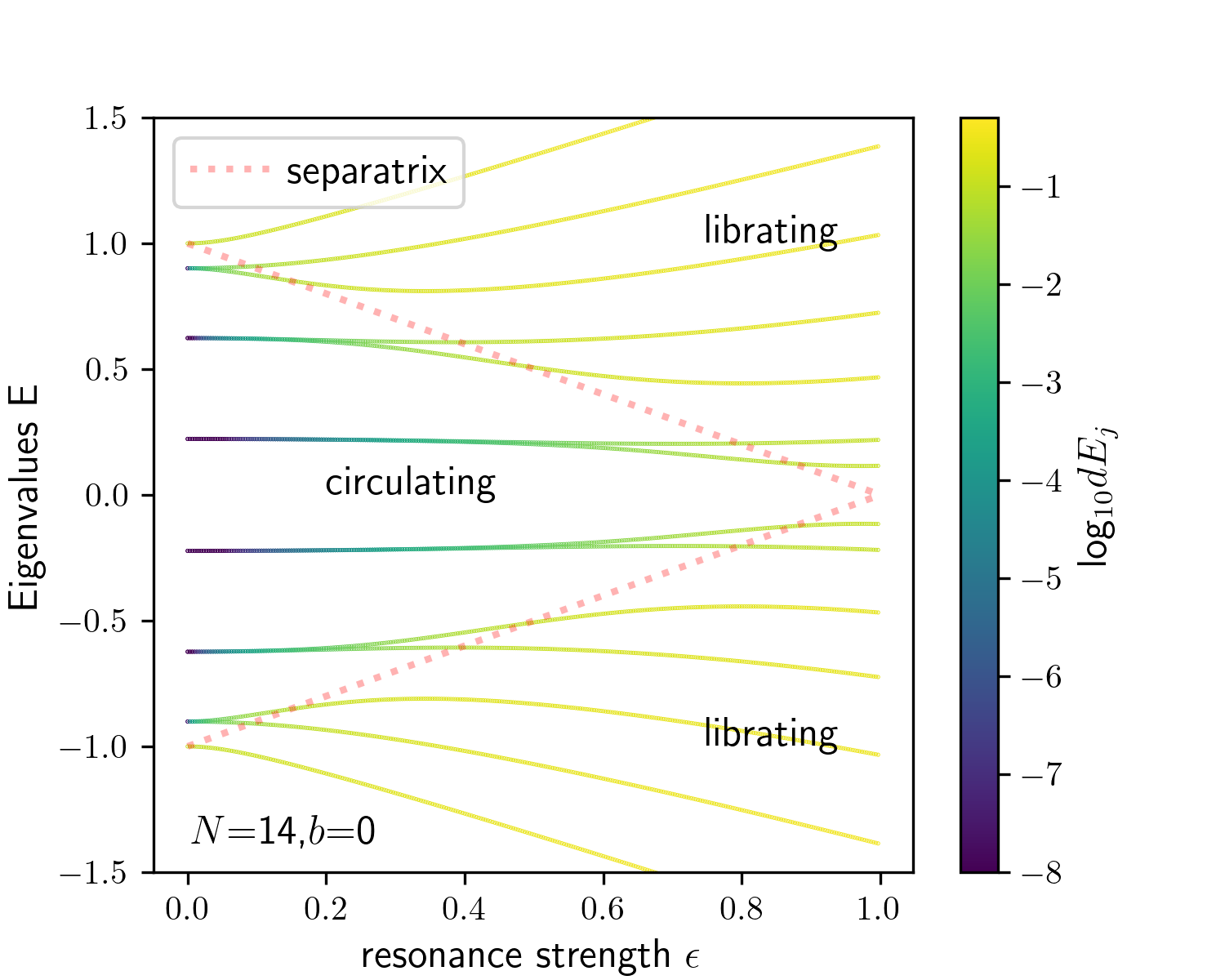}
\caption{Eigenvalues of the Hamiltonian operator in equation \ref{eqn:h0} 
for different values of resonance strength $\epsilon$.  
Parameters $b=0$, $a=1$ remain fixed and the dimension $N=14$.  
For $\epsilon>0$, 
eigenvalues are in nearly degenerate pairs within the circulating region but well 
separated in the region associated with classically librating orbits.
The red dotted line shows the separatrix energy values for the associated classical model 
(equation \ref{eqn:Esep}).
As the resonance strength $\epsilon$ increases, the area within the separatrix orbit increases
and more eigenstates are in the region associated with librating orbits.
The colorbar gives an estimate for the log of the distance to the nearest eigenvalue for each 
eigenvalue and at each value 
of $\epsilon$, however the colorbar is truncated at -8 so the spacings  
can be smaller than shown with navy blue points.
\label{fig:cc1}
}
\end{figure}

At the bottom of the cosine potential well, and at low energy, 
the spectrum of the Harper operator resembles that
of a harmonic oscillator which has evenly spaced eigenvalues.   
For $N$ even there is a reflective symmetry in the spectrum of $\hat h$ about an energy of 0, (see appendix  \ref{ap:even}). 
 In Figure \ref{fig:cc1} we have drawn the energies of the separatrix 
for the associated classical system (given in equation \ref{eqn:Esep}) with dotted red lines.  
The spectrum consists of nearly degenerate eigenvalue pairs in the circulating region and the eigenvalues are well separated in the librating regions.   

The lower half of Figure \ref{fig:cc1} resembles Figure 1 by  \citet{Doncheski_2003} showing the spectrum of 
the quantum pendulum (equivalently, the eigenvalues of the Mathieu function) also as a function of strength parameter (also see \cite{NIST:DLMF}, \url{https://dlmf.nist.gov/28.2}). 
The spectrum of the quantum pendulum also exhibits nearly degenerate pairs of states above 
its separatrix and distinct and well separated energy states within its potential well \citep{Doncheski_2003}.
The dichotomy of energy level spacing in the Harper operator is not caused by confinement to a torus as it is also present in the quantized pendulum with momentum $p \in [-\infty, \infty]$.   Similarity between the 
spectrum of the Harper operator and the quantized pendulum is discussed in more detail in appendix 
\ref{ap:mathieu}. 

The spectrum of the operator $\hat h$ in dimension $N$ at particular values of the parameters $a,b,\epsilon$
is the set of eigenvalues $\{ \lambda_j \}$ indexed by $ j \in \{0, ... , N-1\}$. 
For each spectrum computed with a specific value of parameters $a,b,\epsilon$ and for each
eigenvalue $\lambda_j$ we define the distance in energy to the nearest eigenvalue;  
\begin{align}
dE_j &= {\rm min}_{k\ne j} |\lambda_j - \lambda_k| 
. 
\label{eqn:dEj}
\end{align}
The log$_{10}(dE_j)$ of the eigenvalue (or energy) spacings of the Hamiltonian operator in equation \ref{eqn:h0}   
are shown in color in Figure \ref{fig:cc1} and with numerical values corresponding to the colorbar on the right.   The colorbar does not extend below $dE_j = 10^{-8}$, but the spacings decrease to zero on the left side as $\epsilon \to 0$ where pairs of eigenvalues become degenerate. 
Figure \ref{fig:cc1} shows that the Harper  operator exhibits an extremely wide range of spacings between its eigenvalues.  

\begin{figure}[htbp] \centering 
\includegraphics[width=3truein,trim = 0 0 30 0, clip]{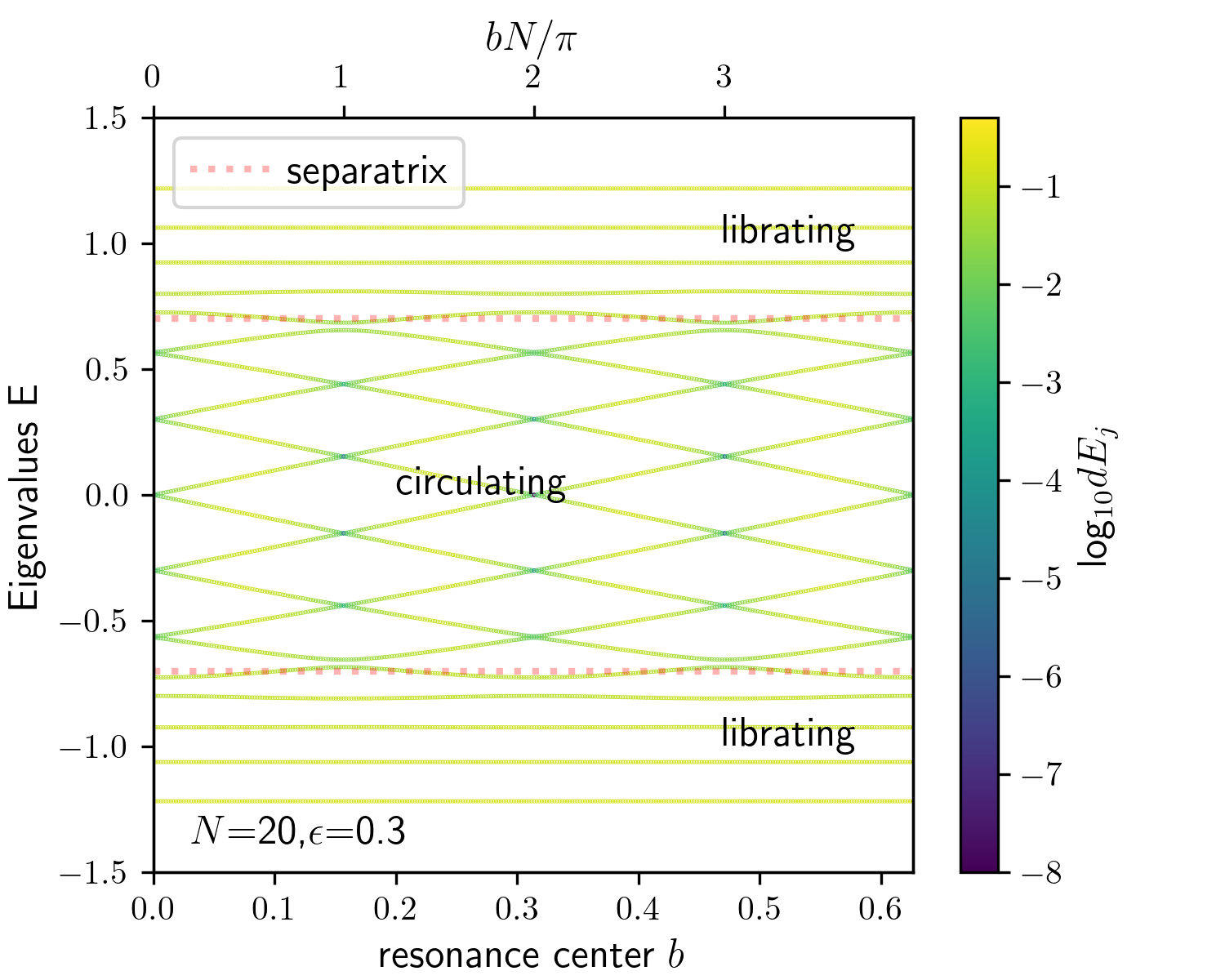}
\caption{
Eigenvalues of the Hamiltonian operator in equation \ref{eqn:h0}   
for different values of parameter $b$.  This figure is similar to Figure \ref{fig:cc1} except
$b$ varies instead of $\epsilon$,  and dimension $N=20$.  
The parameter $\epsilon =0.3$ is fixed.   Energy levels in the librating regions do not vary quickly as $b$ varies.
However, energy levels in the circulating regions exhibit a lattice of avoided crossings. 
The axis on the top shows $b$ in units of $\pi/N$.  Avoided crossings in the circulating region occur
where $b$ is a multiple of $\pi/N$.   A heuristic explanation for the dichotomy is given in section 
\ref{sec:heu}.
\label{fig:cc2}
}
\end{figure}

Figure \ref{fig:cc2} also shows the spectrum of the Hamiltonian operator in equation \ref{eqn:h0}, 
however here we vary the $b$ parameter while the resonance strength $\epsilon$ remains fixed. 
Figure \ref{fig:cc2}  shows that the energy levels in the librating regions remain separated and 
do not vary very much as $b$ is increased.  However, energy levels in the circulating region show a lattice of crossings, 
that are avoided, as we discuss below.  By avoided we mean that the two energy levels don't actually cross, 
rather they approach each other and then diverge from each other.   The avoided crossings have
extrema  
at $b$ equal to a multiple of $\pi/N$, and for these values, the energy levels in the circulating region are nearly degenerate. 
The lattice of avoided crossings was also seen near a separatrix in the 
 time-dependent double well potential model studied by \citet{Stabel_2022}. 
 
The lattice of avoided crossings seen in the circulating region in Figure \ref{fig:cc2} implies that the operator $\hat h(a, b, \epsilon)$ (equation \ref{eqn:h0}) obeys symmetries that are described in more detail in appendix \ref{ap:spectrum} and which we now summarize.  With $a, \epsilon$ fixed,  
the spectrum is identical if $b$ is shifted by $2\pi/N$ (theorem \ref{th:h0r}). 
The spectrum has reflective symmetry about $b$ for $b$ a multiple of $\pi/N$ 
 (theorems \ref{th:bmb} and \ref{th:hhalf}).   
For $N$ even, as shown in Figures \ref{fig:cc1} and \ref{fig:cc2}, 
 the spectrum itself is symmetrical about an energy of zero (theorem \ref{th:even}). 
Degeneracy of eigenvalues is explored in appendix \ref{ap:degen}. 
The eigenvalues are distinct if $b$ is not a multiple of $\pi/N$ (theorem \ref{th:distinct}).  
Numerical calculations show that the spacing between nearly degenerate pairs of eigenvalues is smallest for $b$ a multiple of $\pi/N$ and that eigenvalues are distinction $a,\epsilon \ne 0$ as long as 
$N$ is not a multiple of 4. There are two zero eigenvalues in the special case 
of $N$ a multiple of 4 (lemma \ref{th:zeros}).  
The near approaches of pairs of eigenvalues at $b$ multiples of $\pi/N$ could be related to two symmetries that
are present for $b$ a multiple of $\pi/N$ (see commutators in equations \ref{eqn:Pz2k} 
 and \ref{eqn:Pz1k} of theorems  \ref{th:sym1}, \ref{th:sym2}). 
The fact that spacings are a minimum for specific values of $b$ aids in computing 
 the minimum distance between eigenvalues for drifting systems and is another advantage of 
exploring a simple system such as the Harper operator.   
Many of the symmetries associated with the lattice of avoided crossings are obeyed 
 by operators that have potentials that are more general than the $\cos \hat \phi$ potential 
 in the Harper operator, as discussed in appendix \ref{ap:other}. 
 
With both $b$ and $\epsilon$ increasing, there is both a lattice of avoided crossings in the circulating
region and an increase in the number of eigenstates in the libration region, as shown in Figure \ref{fig:cc3}. 
In the associated classical system, the area in phase space of the librating region grows while that in the circulating region shrinks.  At an appropriate drift rate, 
a system begun in an circulating eigenstate would experience diabatic crossings within the
 circulating region, and then when approaching the separatrix could adiabatically remain in an eigenstate 
 that starts to librate within the potential well, as described for the time dependent double well potential system by \citet{Stabel_2022}.   
This process is the quantum equivalent of the classical process known as {\it resonance capture}. 

\begin{figure}[htbp] \centering 
\includegraphics[width=3truein,trim = 0 0 30 0, clip]{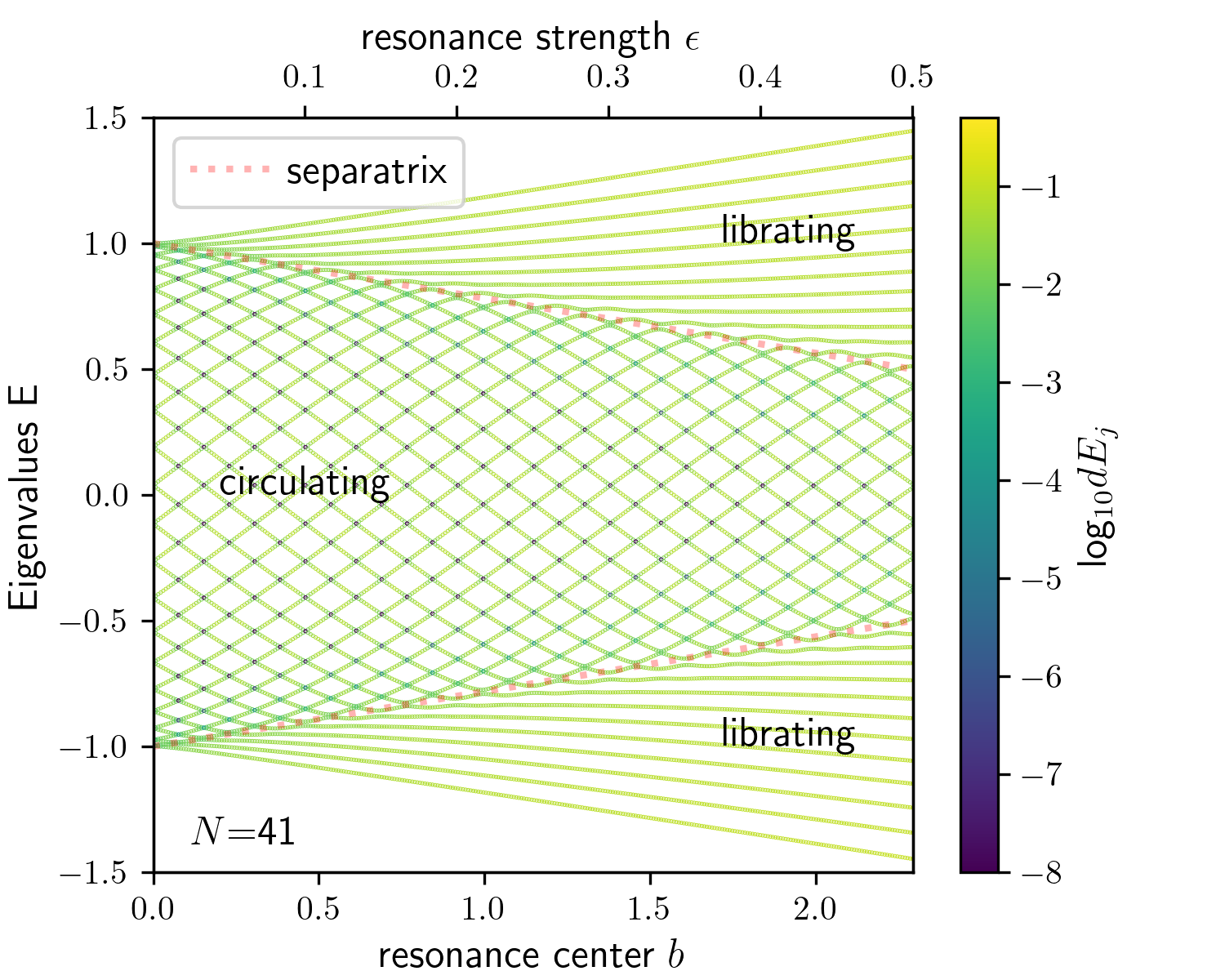}
\caption{Eigenvalues of the Hamiltonian operator in equation \ref{eqn:h0}  
for different values of parameters $\epsilon, b$ 
where both $\epsilon, b$ together increase linearly with values shown on the top and bottom axes. 
This Figure is similar to Figure \ref{fig:cc1} except the parameters $b, \epsilon$ both vary  
and dimension $N= 41$. 
As the resonance strength (set by $\epsilon$) increases, 
more eigenstates are associated with librating orbits.
 \label{fig:cc3}
}
\end{figure}

\subsection{Spacing between eigenvalues}  \label{sec:spacing}

The Landau-Zener two-level model shows that during an avoided crossing the 
probability of a diabatic transition is a function of the drift rate and the minimum energy difference between the two states at their closest approach.     When the parameters describing the Harper operator are 
time dependent, then the probability of diabatic transitions depends on the distances between 
energy levels.  In this section we examine in more detail the near degeneracies seen in 
Figures \ref{fig:cc1} -- \ref{fig:cc3}. 

To compute the eigenvalues in Figures \ref{fig:cc1} -- \ref{fig:cc3} we used 
routines available in the  \texttt{numpy} package within \texttt{Python} which uses the LAPACK linear algebra library. 
The distances between pairs of eigenvalues 
were so small that we were concerned about the accuracy 
of the calculation with the LAPACK linear algebra library.   
To check the accuracy of the calculation, 
we computed eigenvalues values to a higher level of precision using the 
Python library for real and complex floating-point arithmetic with arbitrary precision \texttt{mpmath} \cite{mpmath}. 
In Figure \ref{fig:eee} we show the smallest distance between eigenvalues for $\hat h$ 
(equation \ref{eqn:h0}) with $a=1$ for  fixed values of $\epsilon=0.5$, $b=0$ but 
as a function of $N$ and computed to 50 digits of precision using the \texttt{mpmath} library. 
This is about 20 digits more accurate than computed with the conventional LAPACK library. 

Figure \ref{fig:eee} confirms that the distances between pairs of eigenvalues are non-zero and small within 
the circulating region and smallest in the center of the circulating region where the eigenstate energy 
is near zero.  With the higher level of
precision we find that 
the eigenvalues are not degenerate except in the case of $N$ a multiple of 4, in which 
case there is a pair of zero eigenvalues.   That there is a pair of zero eigenvalues for $N$ 
a multiple of 4 is shown using the determinant of the operator computed in appendix \ref{ap:det} and with lemma \ref{th:zeros}).  
Within the circulating region, the distance between the pairs of eigenvalues decreases over 
orders of magnitude.   The range of spacings between eigenvalues is remarkably vast
and increases with increasing dimension $N$. 
%
In the high energy limit of the quantized pendulum which resembles the quantum rotor, and with increasing distance from the separatrix,  the distance in energy between pairs of eigenvalues also approaches zero
(e.g.,  \cite{Doncheski_2003}, also see appendix \ref{ap:mathieu}). 
For our system, because the circulating region 
is bounded on both sides by potential wells, the minimum spacing between eigenvalues is
found in the center of the circulating region.   

\begin{figure}[htbp]\centering   
\includegraphics[width=3truein,trim = 0 0 0 0, clip]{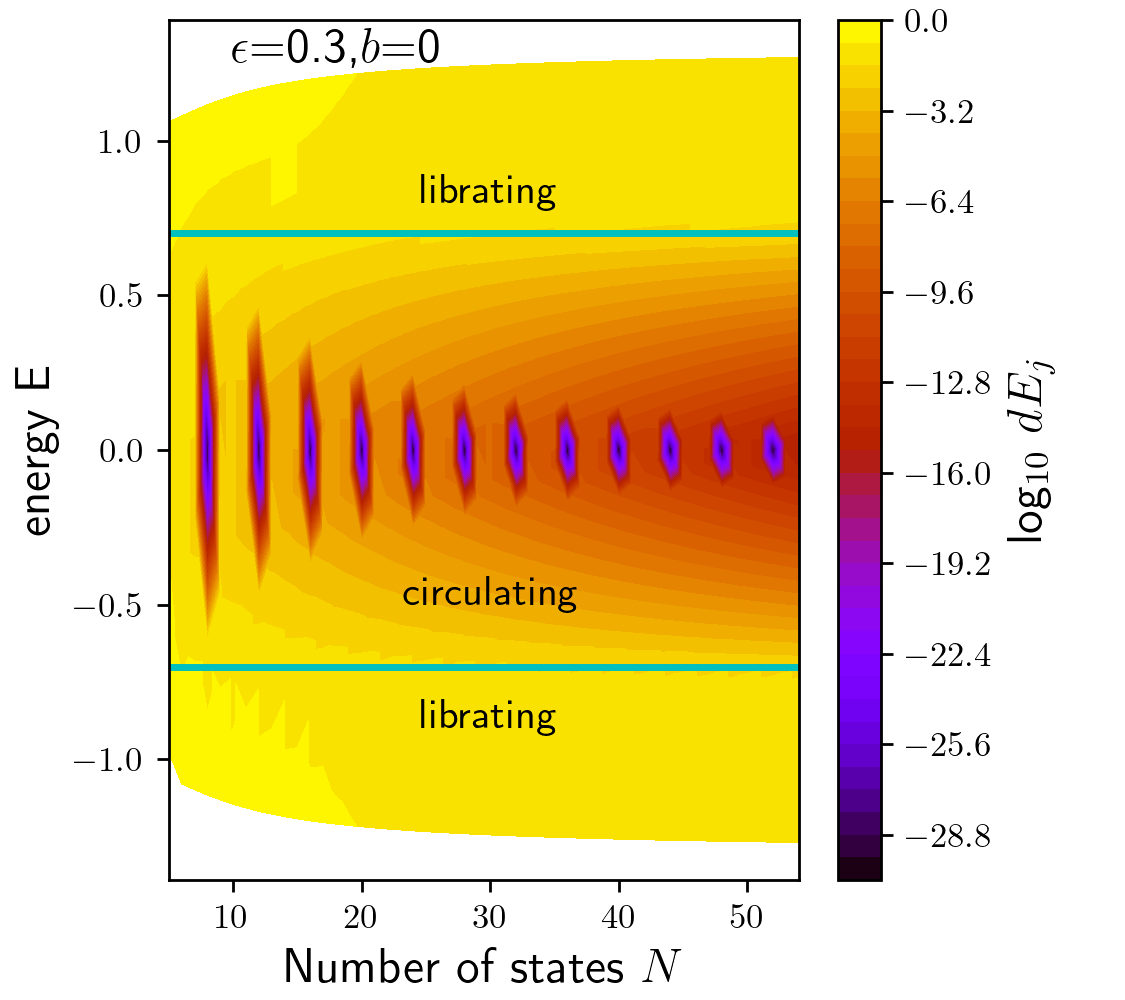}
\caption{We show the smallest distance, in color and with colorbar shown on the right, 
between neighboring eigenvalues as a function of eigenvalue (on the $y$-axis) and as a function of the number of states $N$ in the Hilbert space  on the $x$-axis 
for the Hamiltonian operator of equation \ref{eqn:h0} with $a=1,b=0$.  
The parameter $\epsilon$ is fixed and printed on the plot. 
The cyan thick lines show the energies of the separatrix classical orbits.  
For energies between the separatrix orbits 
there are near degeneracies between pairs of eigenstates.  
The pairs of eigenvalues are not the same except for $N$ a multiple of 4 
and in that case only for a pair of zero eigenvalues. 
\label{fig:eee}
}
\end{figure}

The spacing between eigenvalues at different values of $\epsilon$ for two different values of $N$, 
are shown in Figure \ref{fig:fff}. This figure also illustrates that the minimum distance between pairs of 
eigenvalues is found near an energy of 0 in the circulating region.  The colorbar is in a log 
scale.  Level contours within the circulating region are nearly linear suggesting that the energy level
spacing in the circulating region is approximately a power of $\epsilon$. 

\begin{figure}[htbp]\centering   
\includegraphics[height=1.596truein,trim = 3 0 18 0, clip]{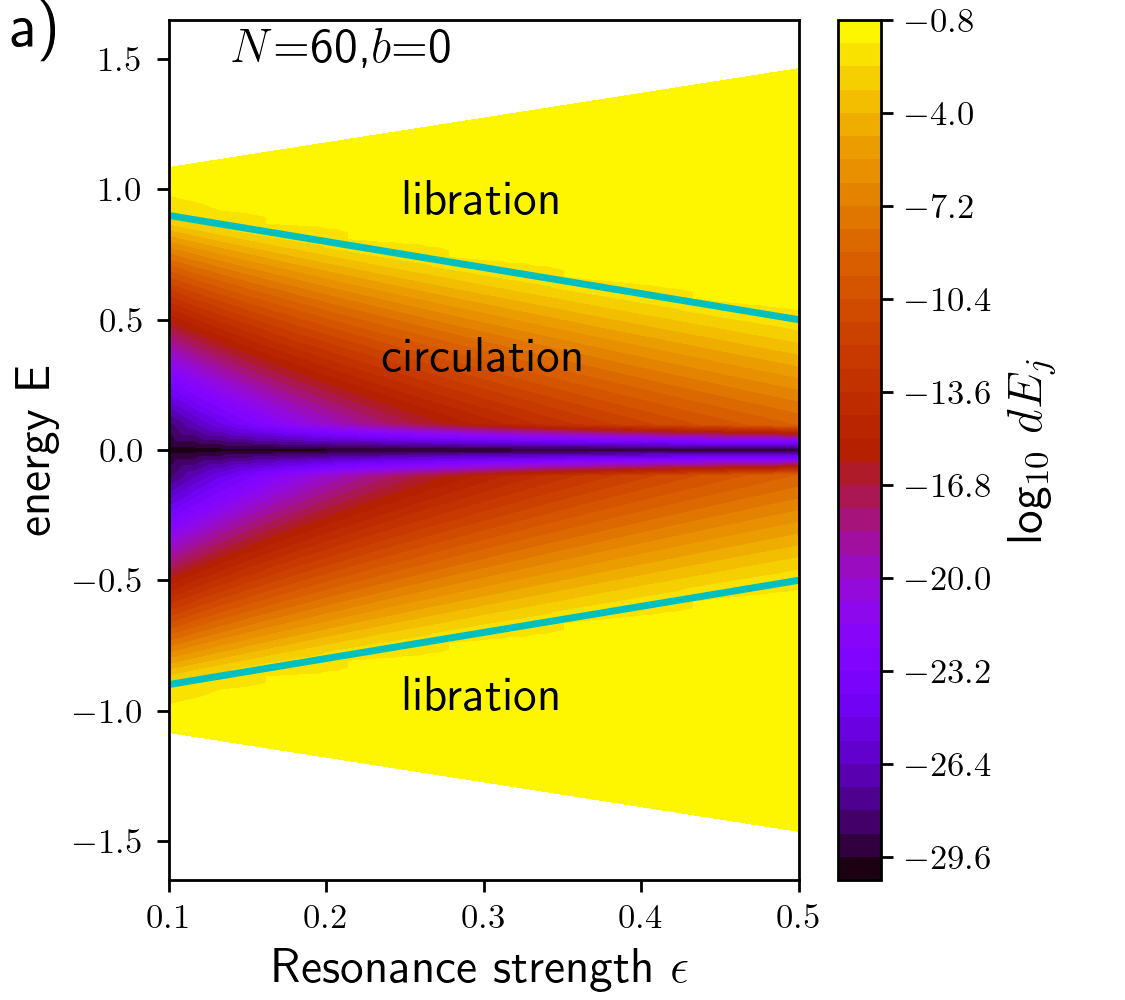}
\includegraphics[height=1.596truein,trim = 3 0 18 0, clip]{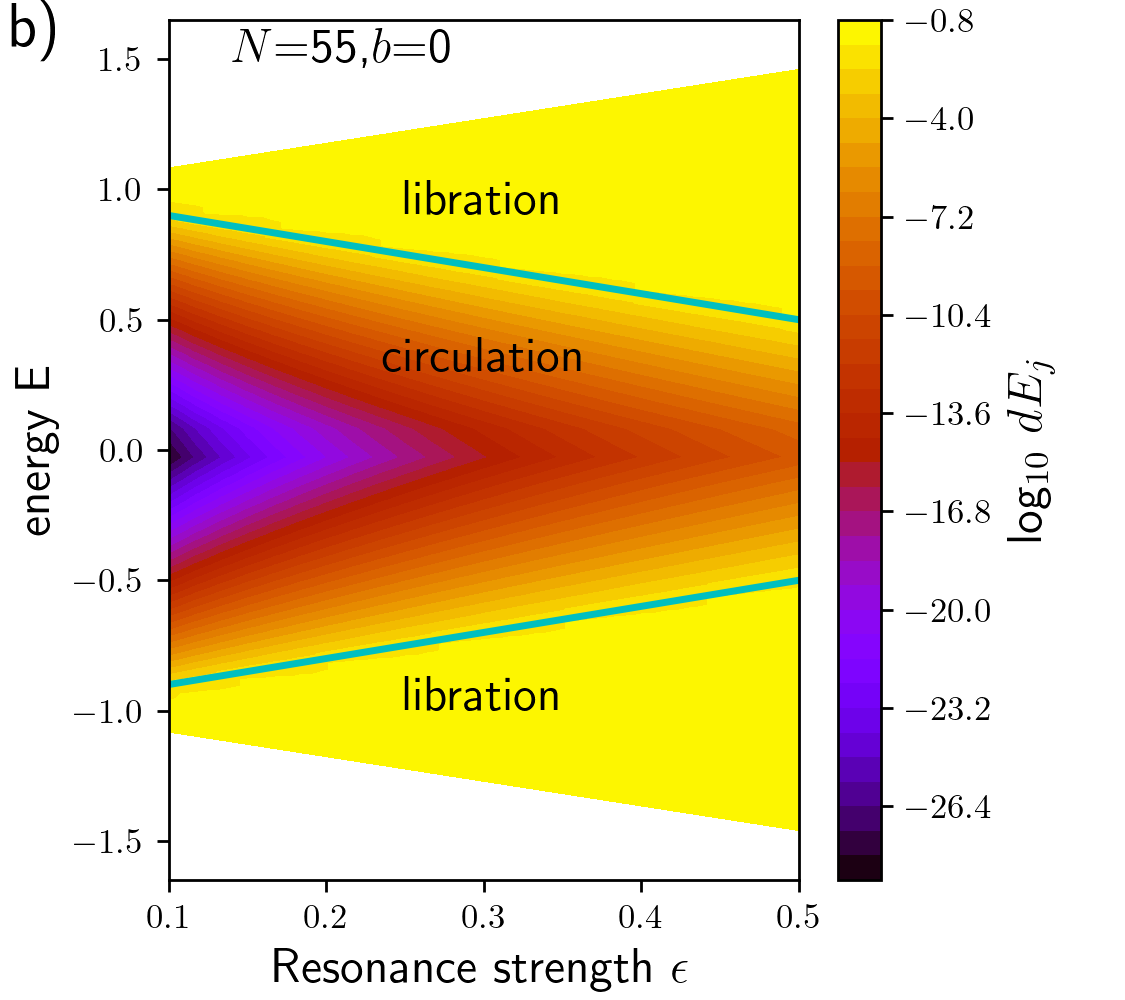}
\caption{
a) We show the spacing,  in color and with colorbar shown on the right, 
between neighboring eigenvalues as a function of eigenvalue (on the $y$ axis) and as a function of 
the parameter $\epsilon$ 
(on the $x$-axis) for the Hamiltonian of equation \ref{eqn:h0} with $a=1$, $b=0$.   
The number of states is shown on the top left.  For $N$ a multiple of 4, there are two zero energy
eigenvalues. 
b) Similar to a) but for an odd dimension $N$. 
The spacing between eigenvalues increases as $\epsilon$ increases 
and for each $\epsilon$ is smallest  in the center of the circulating region. 
The distance between neighboring eigenvalues spans many orders of magnitude. 
\label{fig:fff}
 }
\end{figure}

How the distance between nearby energy levels depends on eigenvalue index is shown in 
Figure \ref{fig:slice} where we plot the eigenvalue spacing $dE_j$  (defined in equation \ref{eqn:dEj}), 
  as a function of eigenvalue index, with index $j$ in order 
of increasing eigenvalue.  Figure \ref{fig:slice} illustrates that the distance between the pairs
of nearly degenerate eigenvalues drops rapidly in the circulating region.  
The drop is nearly linear on a log plot implying 
 that the spacing depends on a power of the index, with exponent 
dependent upon dimension $N$ and $\epsilon$. 
In Figure \ref{fig:slice} we show eigenvalue spacing 
$N=50$ even but not a multiple of 4, and $N=52$ which is a multiple of
4.  The spacings are similar except at $j/N \sim 0.5$ where in the $N=52$ (divisible by 4) case there 
is a zero eigenvalue of multiplicity 2 giving $dE_j = 0$. 

\begin{figure}[htbp]\centering 
\includegraphics[width=1.66truein,trim = 14 0 1 0, clip]{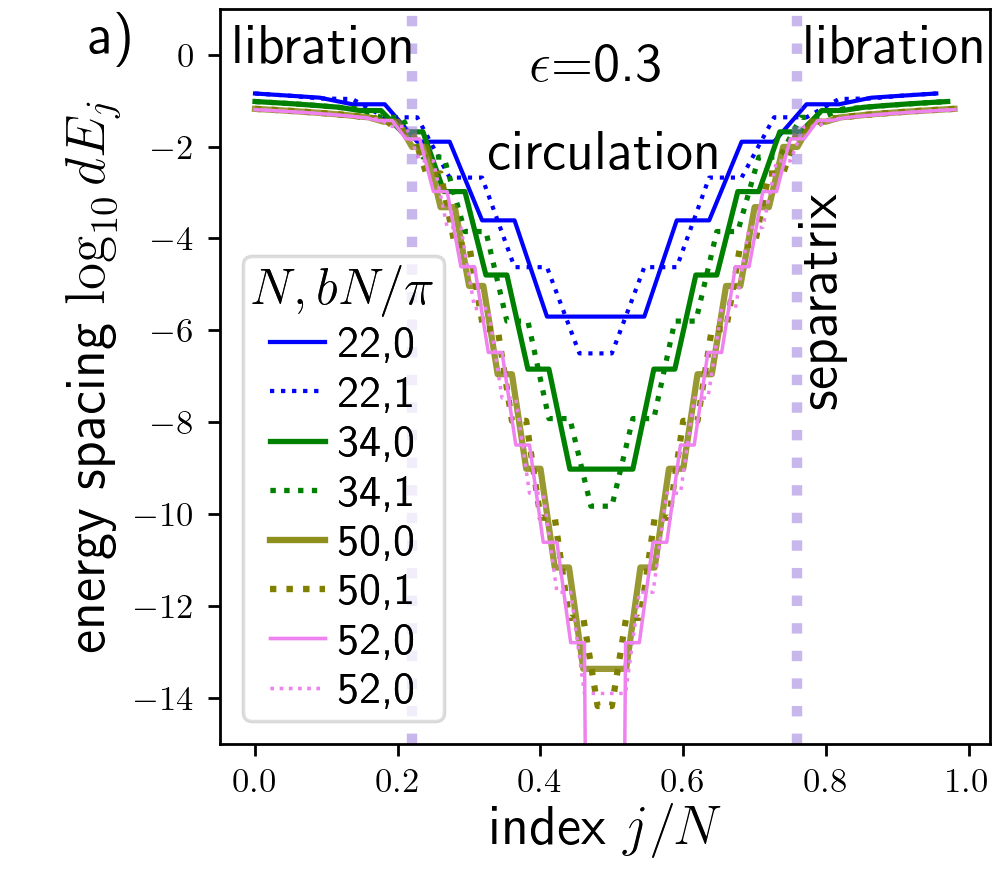}
\includegraphics[width=1.66truein,trim = 14 0 1 0, clip]{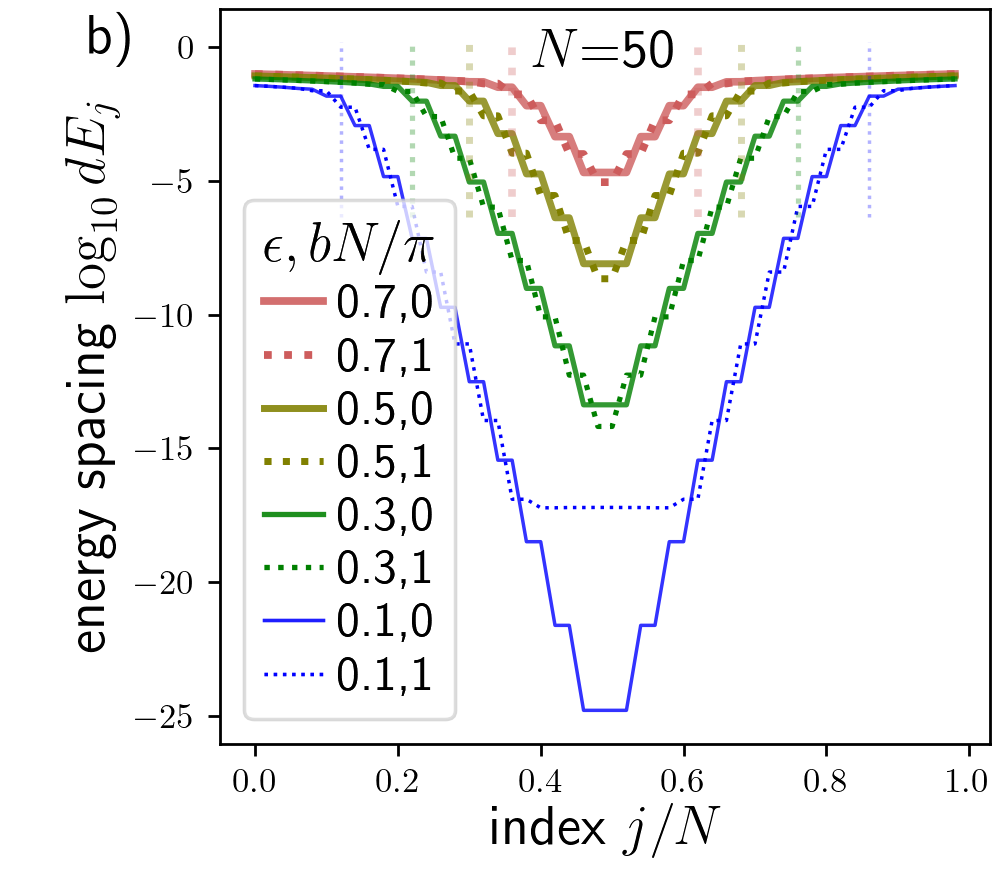}
\includegraphics[width=1.66truein,trim = 14 0 1 0, clip]{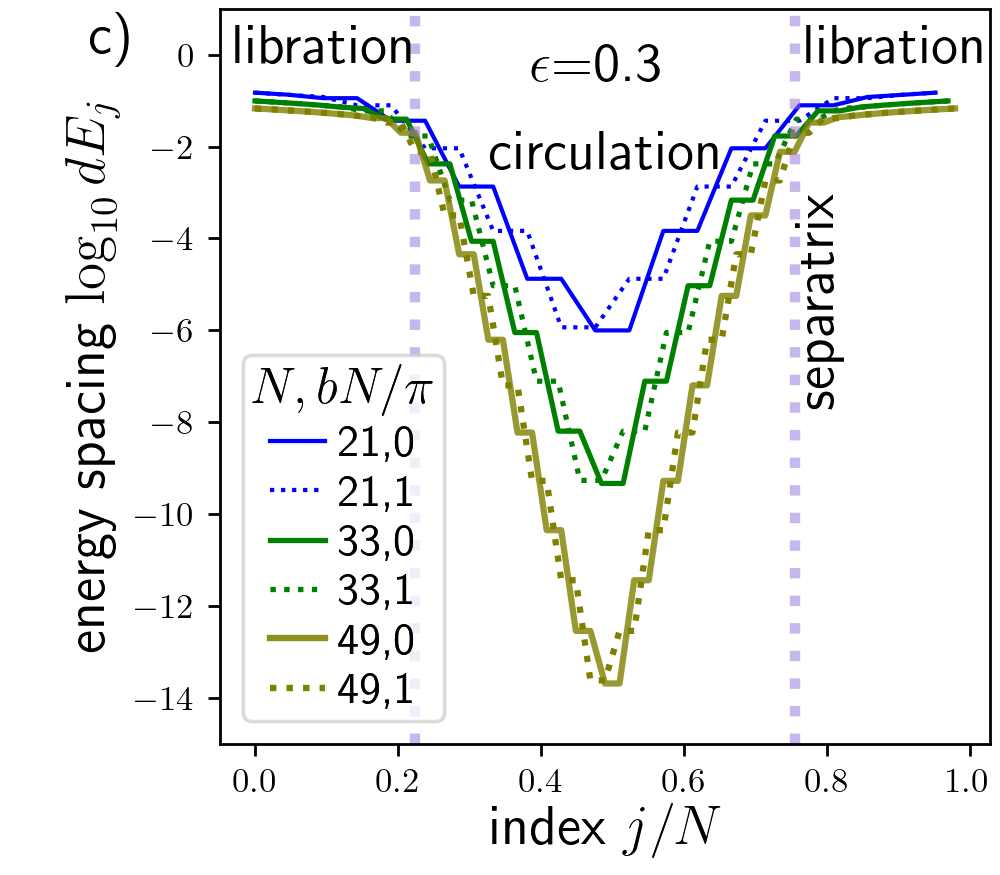}
\includegraphics[width=1.66truein,trim = 14 0 1 0, clip]{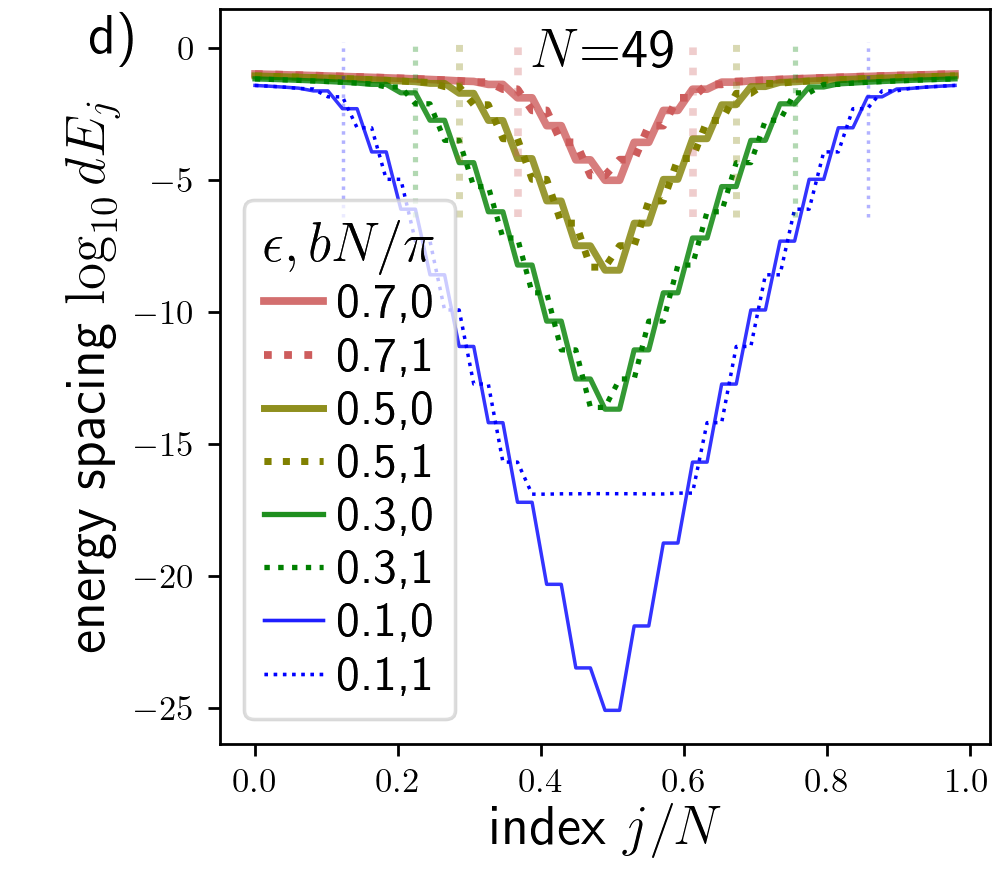}
\caption{a)  The log$_{10}$ of spacings $dE_j$ between neighboring energy levels 
 is shown as a function 
of index $j$ of ordered energy levels for  different values of dimension $N$, for $N$ even, and at 
resonance strength $\epsilon=0.3$ for the operator $\hat h(a,b,\epsilon)$ with $a=1$. 
The solid and dotted lines show 
spacings $dE_j$ computed for $b=0$ and $b=\pi/N$, respectively.  The indices 
that have energy near that of the separatrices
are shown with vertical dotted lines. 
b) Similar to a) except we show spacings for even $N=50$ for 4 different values of $\epsilon$. 
c) Similar to a) except for three different values of $N$ odd. 
d) Similar to b) except we show spacing for odd $N=49$. 
These figures illustrate that the spacing between pair of eigenvalues drops  
within the circulating region approximately as a power law. 
\label{fig:slice}
}
\end{figure}

The minimum distance between any pair of 
eigenvalues is reached at the center of the circulating region.  
Figure \ref{fig:mins} shows the minimum distance between any pair eigenvalues as a function 
of $N$ and $\epsilon$ for $b=0$ and $b=\pi/N$ and for $N$ even but not a multiple 
of 4 and $N$ odd. 
We find that the minimum distance is approximately given by 
\begin{align}
dE_{min} \sim \begin{cases} 
 \epsilon^\frac{(N-1)}{2}\frac{3}{N}  & \text{if } N \text{ odd }\\
 \epsilon^\frac{(N-2)}{2}\frac{5}{N}  & \text{if } N \text{ mod } 4 = 2\\
 0 & \text{if } N \text{ mod } 4 = 0
 \end{cases} . \label{eqn:min_dist}
\end{align}
These functions are plotted on Figure \ref{fig:mins} which shows that they 
 approximately describe the minimum spacing computed numerically.  
 Low order perturbative expansion techniques, which we use in the next section,   
  are often used to estimate
 energy level spacings (e.g., \cite{Zakrzewski_2023}).  However, 
 the power of $\epsilon$ in equation  \ref{eqn:min_dist} 
  is high, so a low order perturbative expansion would 
fail to estimate the minimum spacing  between energy levels in the Harper operator. 
Using the characteristic polynomial of the operator $\hat h(1,0,\epsilon)$, we give  
a rough derivation of equation \ref{eqn:min_dist} in appendix \ref{ap:min_spacing}. 
 
\begin{figure}[htbp]\centering
\includegraphics[width=1.65truein]{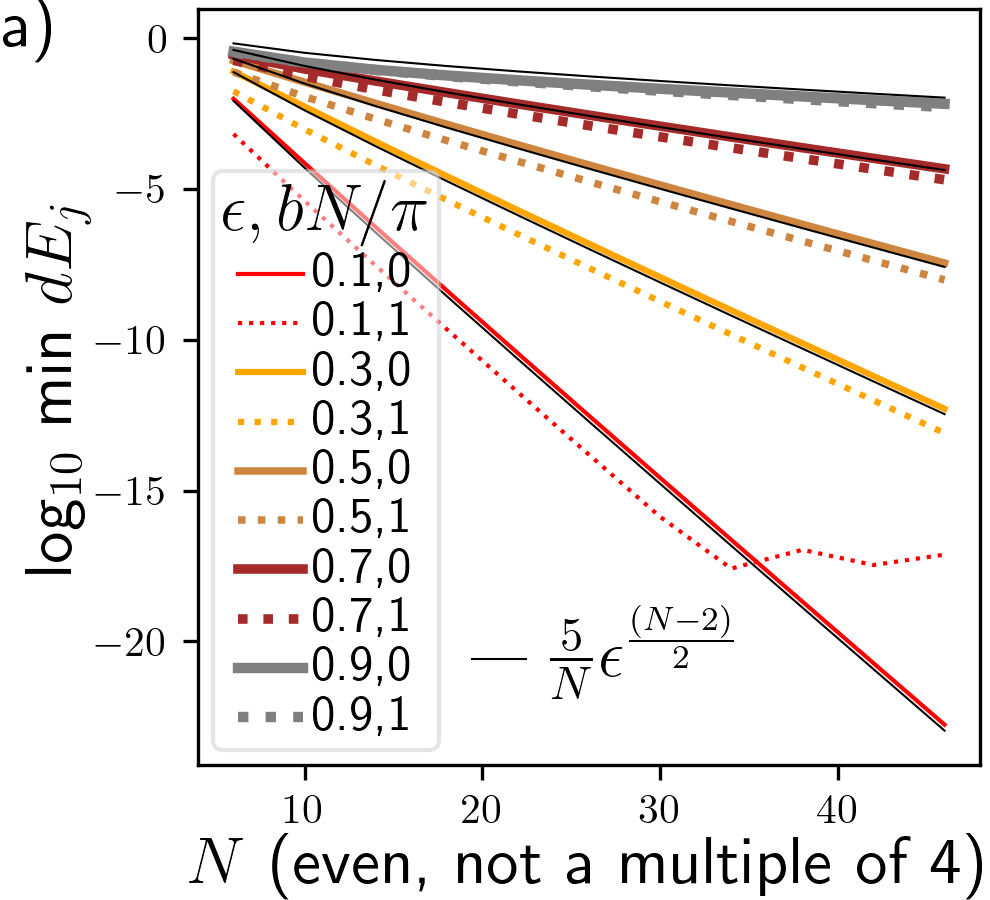}
\includegraphics[width=1.65truein]{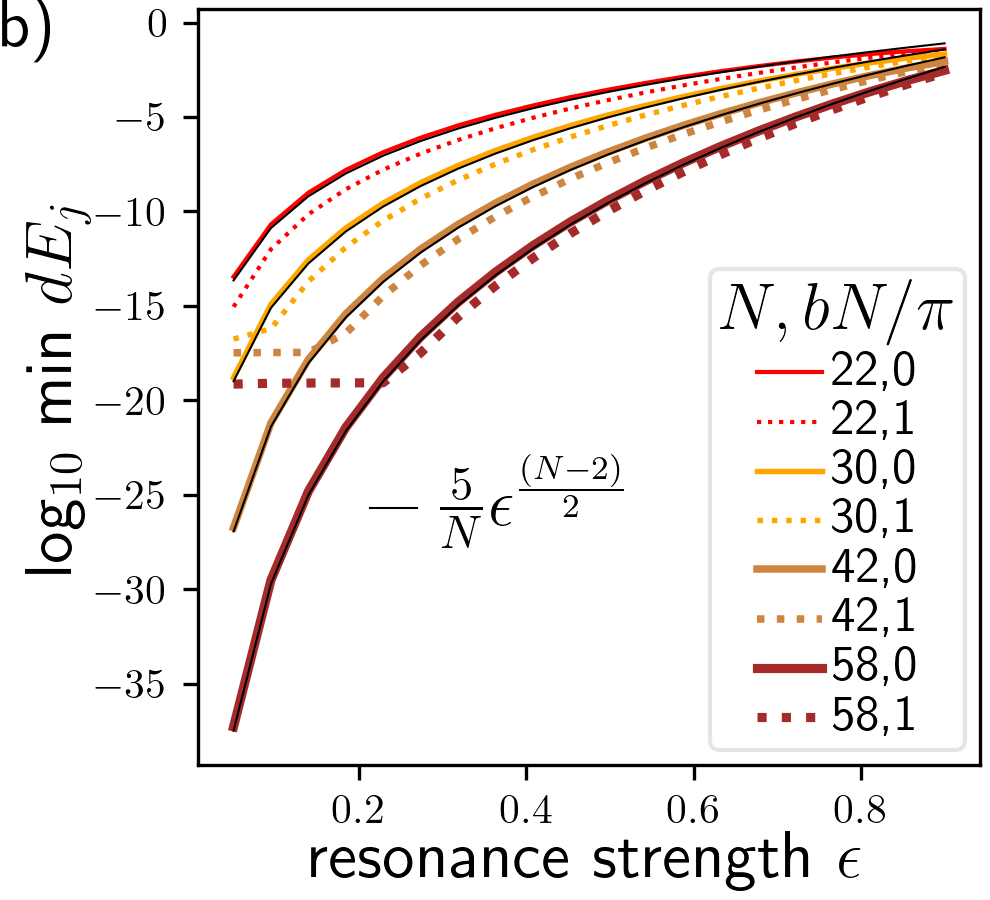}
\includegraphics[width=1.65truein]{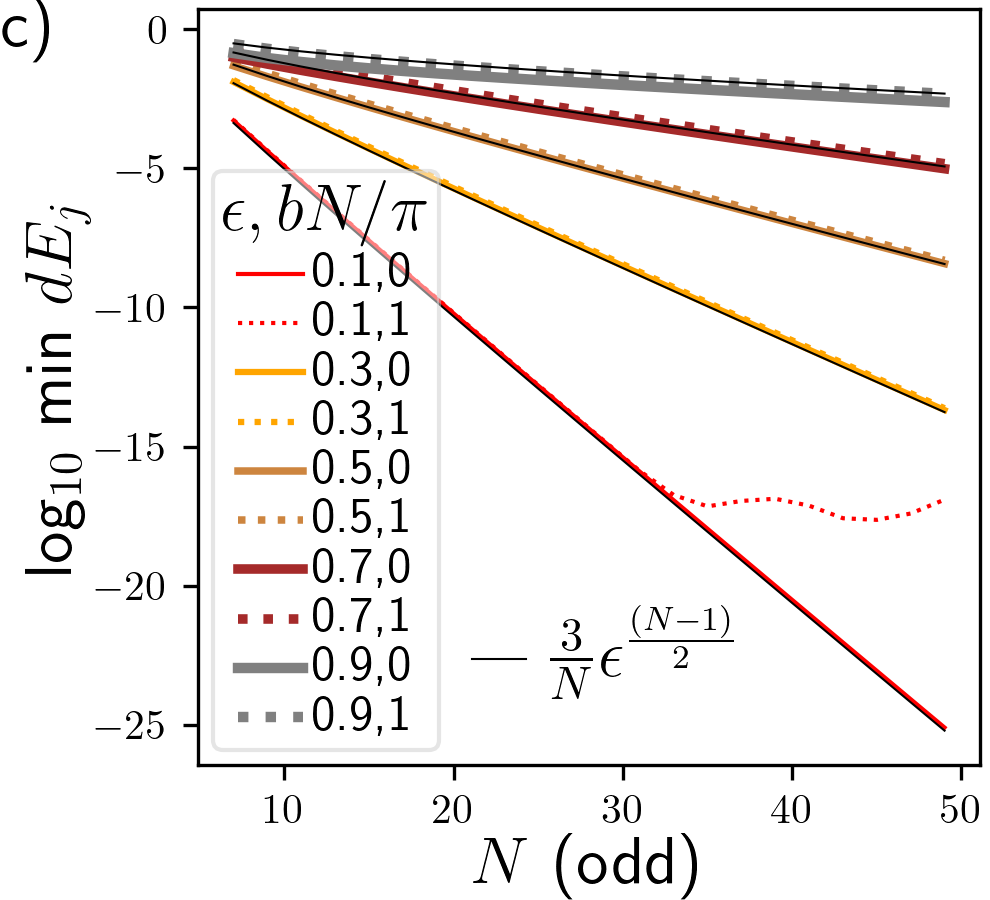}
\includegraphics[width=1.65truein]{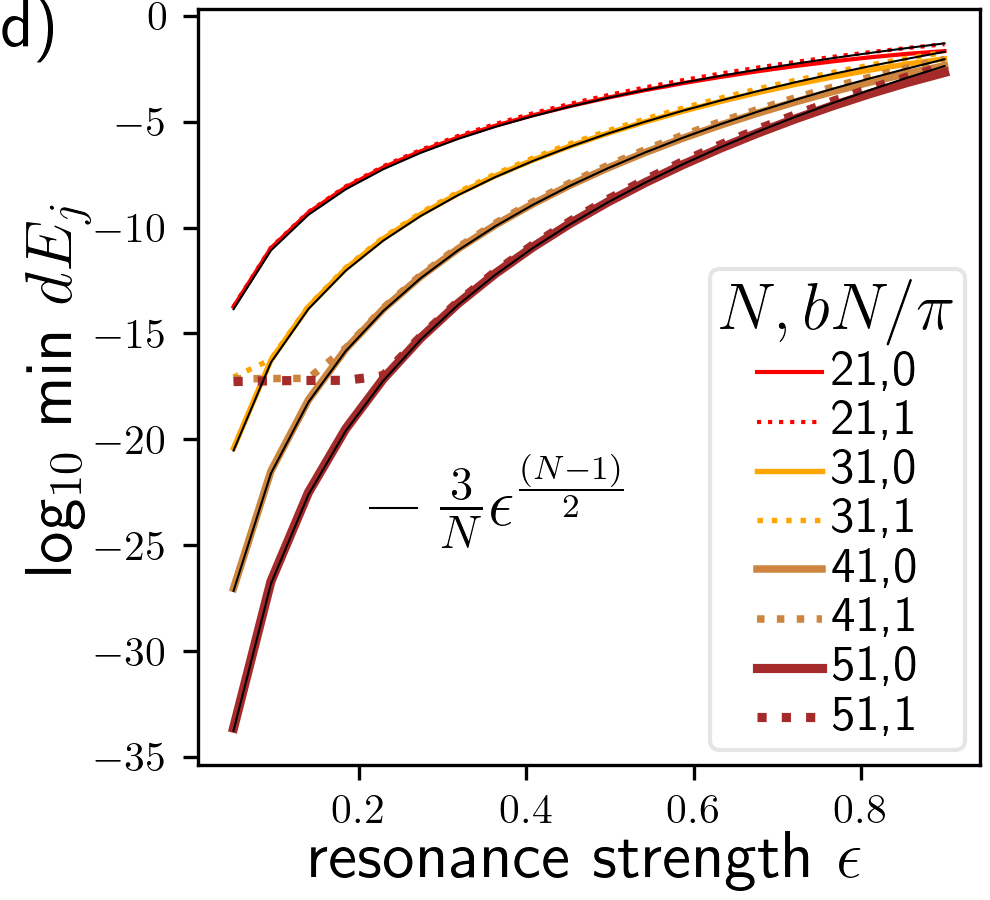}
\caption{a) The log$_{10}$ of the minimum spacing between eigenvalues (min $dE_j$) 
for the operator $\hat h(a,b,\epsilon)$, with $a=1$,  as a function of $N$ 
for $N$ even but not a multiple of 4.  
We show the minimum space  for different values of resonance strength $\epsilon$. 
Solid lines have $b=0$ and dotted lines  $b=\pi/N$. 
The thin black lines show the estimate for the minimum gap given by equation \ref{eqn:min_dist}.   
b) Similar to a) except as a function of resonance strength $\epsilon$.  
c) Similar to a) except for $N$ odd.    
d) Similar to b) except for $N$ odd.   
\label{fig:mins}}
\end{figure}

\subsection{Heuristic analogies for the sensitivity of the energy levels to the parameter $b$}
\label{sec:heu}


In our finite dimensional space and for small resonant strength $\epsilon$, an energy level of 0
is in the circulating region and is above the top of the cosine potential well. 
We approximate the Hamiltonian of equation \ref{eqn:h0} (with $a=1$) as 
\begin{align}
\hat h(1,b,\epsilon) \approx \cos (\hat p - b)\end{align} 
which has eigenvalues 
\begin{align}
E_k(b) = \cos \left( \frac{2 \pi k}{N} - b \right) 
 \end{align}
 for $k \in \{0, 1, \ldots, N-1 \}$.   
Equivalently we can take $k \in -(N/2-1), \dots,  N/2 $  if $N$ is even 
 or $k \in \{ -(N-1)/2, \ldots, 0, \ldots (N-1)/2 \}$ if $N$ is odd. 
For $b=0$ most energy levels have multiplicity 2 as $E_k(0) =E_{-k}(0)$. 
For $b \ne 0$, $k>0$, energy levels are split by 
\begin{align}
E_k(b) - E_{-k}(b) = 2 \sin  \left( \frac{2 \pi k}{N} \right ) \sin b.
\end{align}
For small $b$ the splitting between energy levels is first order in $b$. 
The energy levels diverge as $b$ increases. 

In contrast, at the bottom of the potential well in the Harper model the Hamiltonian can be approximated
by a harmonic oscillator with Hamiltonian 
$\hat h = \hat a^\dagger \hat a + {\rm constant}$ where $\hat a^\dagger, \hat a$ 
are raising and lowering operators.  The energy spectrum is a non degenerate ladder spectrum 
$E_n = n + {\rm constant}$ with $n$ a non negative integer. 
The eigenstates $\ket{n}$ obey $a^\dagger a \ket{n} = n\ket{n}$. 
A shift associated with parameter $b$ 
can be modeled with a perturbation 
$\hat V = b \hat p = \frac{b}{\sqrt{2}i}(\hat a - \hat a^\dagger) $. 
 First order perturbations vanish as $\bra{n} \hat V \ket{n} = 0$.  The perturbation gives 
 second order perturbations to the energy levels 
 \begin{align}
 E_n(b) &\approx n + \frac{b^2}{2} \Bigg( \frac{ | \bra{n+1} (a - a^\dagger)\ket{n} |^2}{-1}  \nonumber \\
& \qquad +\frac{ | \bra{n-1} (a - a^\dagger)\ket{n} |^2}{1} \Bigg) \nonumber \\
 & \approx n - \frac{b^2}{2} + {\cal O }(b^3).
 \end{align}
 The result is a shift in the spectrum that to second order in $b$ does not affect the 
 energy spacing between eigenstates. 
 
These examples heuristically illustrate why the energy levels do not strongly vary in the librating regions 
as $b$ varies but are quite sensitive to $b$ in the circulating regions. 

\begin{figure*}[htbp]\centering 
\includegraphics[width=7truein]{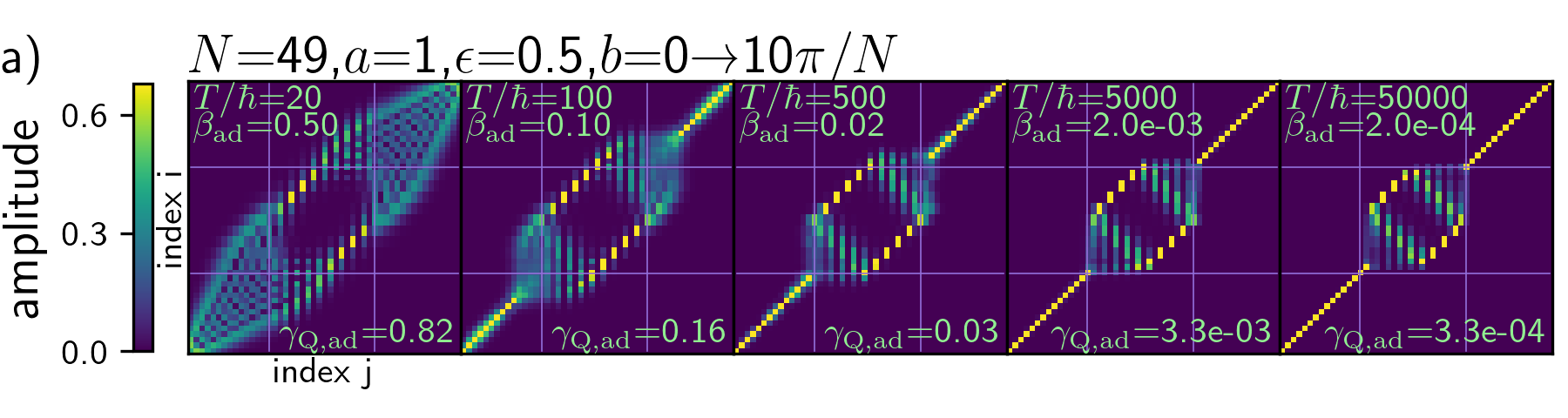}
\includegraphics[width=2truein]{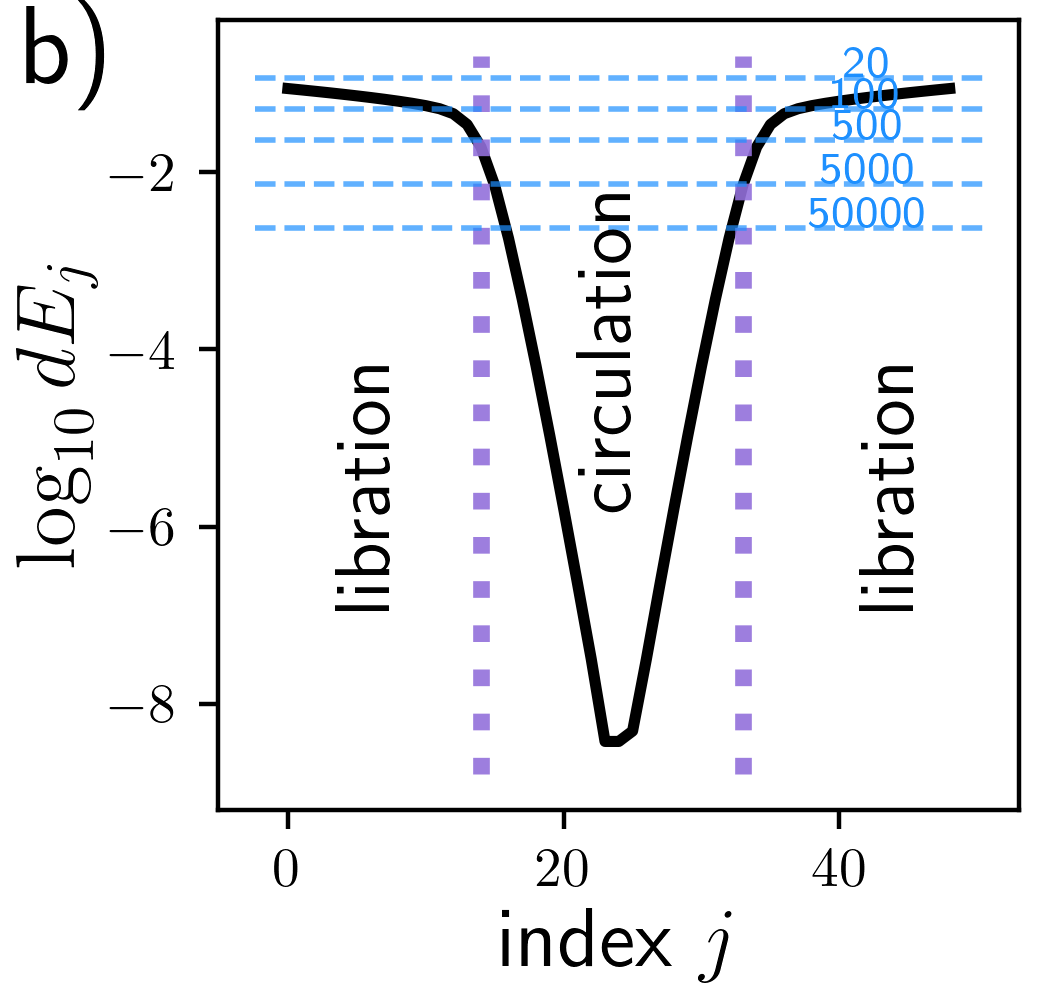}
\caption{
a) We show the amplitudes of the transition matrix $A_{ij} = |\bra{w_j} U(0, T) \ket{v_i}|$  as an image. 
The amplitudes are computed for different durations $T$ using time dependent Hamiltonian operator $\hat h(a,b(t),\epsilon)$ (in the form of equation \ref{eqn:h0}) with $a, \epsilon$ fixed but with $b$ drifting linearly in time, initially at $b= 0$ 
and reaching a value of $10\pi/N$. The  
dimension $N = 49$.   Indices in the transition matrix are ordered so that eigenvalues are increasing to the 
right and upward. 
The set of eigenstates $\ket{v_i}$ are those of the operator $\hat h$ 
initially and the set of eigenstates $\ket{w_j}$ are those of the operator at the final time. 
Each panel shows the transition matrix amplitude 
computed for different durations of drift but for the same total change in system parameters. 
The transition matrix contains a 1 if the system remains in an eigenstate.  If 
the amplitude is between 0 and 1 then transitions can occur.  
In the leftmost panel, the drift rate is sufficiently fast that transitions take place in the librating region 
but transitions are diabatic within the circulating region near the center of the image. 
Horizontal and vertical violet lines show the indices of the separatrix.  Indices that correspond 
to transitions between circulating eigenstates are in the center, whereas those that 
correspond to transitions between librating states are on the lower left and upper right. 
b) The minimum distance between energy levels $dE_j$ (defined in equation \ref{eqn:dEj}) for 
the operator $\hat h$ 
 is shown on the right. The indices with eigenvalues nearest the separatrices are shown with dotted purple vertical lines.  For each drift rate shown in panel a), horizontal blue dashed lines show the spacing
 that would give a probability
of about 1/2 for a diabatic transition estimated using the Landau-Zener model (using equation 
\ref{eqn:de_min_half}). 
This figure illustrates that the system undergoes transitions for all the drift rates we used, and 
would not be completely adiabatic (and transition-less) unless 
the drift duration was 5 orders of magnitude longer than our longest and slowest integration.  
\label{fig:drift}
}
\end{figure*}

\section{Drifting the quantized Harper model}

The near degeneracy of pairs of eigenstates for the Harper operator affects the adiabatic behavior
of the drifting quantized system. 
A slowly varying quantum system with Hamiltonian operator $\hat h(t)$ a function of time $t$ 
 is described with the unitary transformation, called a propagator, 
\begin{align}
\hat U(t, t_0) = {\cal T} e^{- \frac{i}{\hbar} \int_{t_0}^t \hat h(t) dt } \label{eqn:UT}
\end{align}
where $\cal T$ denotes time ordering for each portion of the integral and $\hbar$ is Planck's constant. 
A system initially at time $t_0$ in quantum state $\ket{\psi_0}$ would be in state
$\hat U(t, t_0) \ket{\psi_0}$ at a later time $t$. 

We consider  the time dependent Hamiltonian operator
$\hat h(t) = \hat h(a,b(t), \epsilon(t))$ with operator $\hat h(a,b,\epsilon)$ equal
to the Harper operator (equation \ref{eqn:h0}). 
We allow parameters $b$ and $\epsilon$ to drift linearly in time, and fix $a =1$.  

If the system drifts adiabatically, then a system begun 
in an eigenstate remains in an eigenstate. We denote  
 $\ket{v_i}$ to be the eigenstates of the initial Hamiltonian $\hat h(t_0)$ 
and $\ket{w_j}$ to be eigenstates of the final Hamiltonian $\hat h(t_0+T)$ where $T$ is the duration 
of the drift.  With the two sets of eigenstates,  we compute a matrix  of transition amplitudes 
\begin{align}
A_{ij} = |\bra{w_j} \hat U(t_0,t_0+T) \ket{v_i}|. \label{eqn:trans}
\end{align}
If the system drifts sufficiently adiabatically,  the transition matrix $A_{ij}$ would only contain 1s and zeros
as a system begun in an eigenstate would remain in an eigenstate. 

\subsection{Varying the resonance center}

Figure \ref{fig:drift}a shows transition matrices computed for a drifting system with $\hat h(a,b(t),\epsilon)$ 
(defined in equation \ref{eqn:h0})
with parameters $a,\epsilon$ fixed and $b$ varying linearly with time ($\frac{d b}{dt}$ is constant) 
and initial $b(t=0) = 0$. 
We compute the propagator $U(0,T)$ (equation \ref{eqn:UT}) 
with drift duration $T$ and for different drift durations.    
The total change in the $b$ parameter is $\Delta b = 10 \frac{\pi}{N}$  and it  
is the same for each computed propagator but the duration $T$ of the drift differs in each panel in Figure \ref{fig:drift}a. 
To construct the transition matrix, eigenstates are sorted in order of increasing energy. 
The transition matrix is shown as a color image with the color of each pixel indexed by $i,j$,
corresponding to the specific value of $A_{ij}$.   The vertical axis gives the index $i$ of the initial
eigenstates and the horizontal axis is the index of $j$ the final eigenstates. 
The indices of the states nearest the separatrix energies are shown with thin violet horizontal and 
vertical lines. 

Figure \ref{fig:drift}b shows the minimum spacing between pairs of eigenvalues 
reached as a function of index during the entire duration of the drift.   Pairs of eigenvalues 
undergo different close approaches at $b$ an even multiple of $\pi/N$ and at $b$ an odd 
multiple of $\pi/N$, so the smallest distance shown is the minimum taking into account both types of
avoided crossings. 
The indices of energies nearest the separatrix energy are shown on the plot with light-purple dotted lines.  As expected, 
 the minimum distance between eigenvalues is smallest in the center of the circulating region. 
Horizontal lines on this plot show an estimate (derived in 
appendix \ref{ap:LZ}) based on the Landau-Zener model
 for the energy spacing between two states that would 
have a diabatic transition probability of 1/2 (computed with equation \ref{eqn:de_min_half}).  
A horizontal line is shown for each of the drift durations $T$ shown in Figure \ref{fig:drift}a
and the duration $T$ values are labelled on top of each line.  

Figure \ref{fig:drift}a shows that at rapid drift rates (with short duration for the drift), transitions 
take place between eigenstates in the librating region, whereas in the inner part of the
 circulating region (in the central region of the image panels) there are diabatic transitions. 
At intermediate drift rates ($T/\hbar = 500, 5000$), transitions are suppressed in the librating regions where the dynamical behavior is 
adiabatic and some transitions can take place in the outer parts of the circulating region.  
With decreasing drift rate, the diabatic transition region in the center of the circulating region shrinks.  However, Figure \ref{fig:drift}b
shows that the drift rate would need to be about 5 orders of magnitude lower than our
longest integration to ensure that
no transitions take place and that the system drifts adiabatically for any initial state. 

We compare the drift rates for $b$ shown in Figure \ref{fig:drift} to the adiabatic limit for resonance capture estimated for the classical pendulum (following \citet{Quillen_2006}). 
The KNH theorem holds and the classical system is said to be drifting adiabatically if $\dot b \ll \omega_0^2 $ \citep{Quillen_2006} where $\omega_0$ is the frequency of libration and also the timescale of exponential divergence from the hyperbolic fixed point contained in the separatrix.    
This limit is consistent with the requirement 
that the time to cross the resonance width exceeds the libration period within resonance. 
For the classical pendulum in the form of equation  \ref{eqn:pend2}, the frequency of libration 
\begin{align}
\omega_0  = \sqrt{a\epsilon}  \label{eqn:omega_0} 
\end{align}
which is the same as that of the Harper classical Hamiltonian at the bottom of its potential well. 
We define a dimensionless ratio  
\begin{align}
\beta_{\rm ad} \equiv \frac{\dot b }{\omega_0^2} =\frac{\Delta b}{T} \frac{1}{a \epsilon} . \label{eqn:beta}
\end{align}
where $\Delta b$ is the change in $b$ during a time $T$.   
In the classical setting, 
drift would be considered adiabatic if 
\begin{align}
\beta_{\rm ad} \ll 1.
\end{align}
 
Each panel in Figure \ref{fig:drift}a, shows the $\beta_{\rm ad}$ parameter computed 
 for the parameters of the integration.   These are computed 
 using $\hbar = 2 \pi/N$ resulting from quantization \citep{Quillen_2025}. 
For the shortest duration drift (the leftmost panel in Figure \ref{fig:drift}) with $T/\hbar =20$, the ratio  
$\beta_{\rm ad} = 0.5$ and as this is near 1, the drift is near the adiabatic limit.  
However for $T/\hbar=500$, the ratio $\beta_{\rm ad} = 0.02$ and is below the adiabatic limit for classical resonance capture. 
Figure \ref{fig:drift} illustrates that at drift rates that are well below the classically estimated adiabatic limit 
for resonance capture, both adiabatic and diabatic transitions are likely. 

\begin{figure*}[htbp]\centering 
\includegraphics[width=7truein]{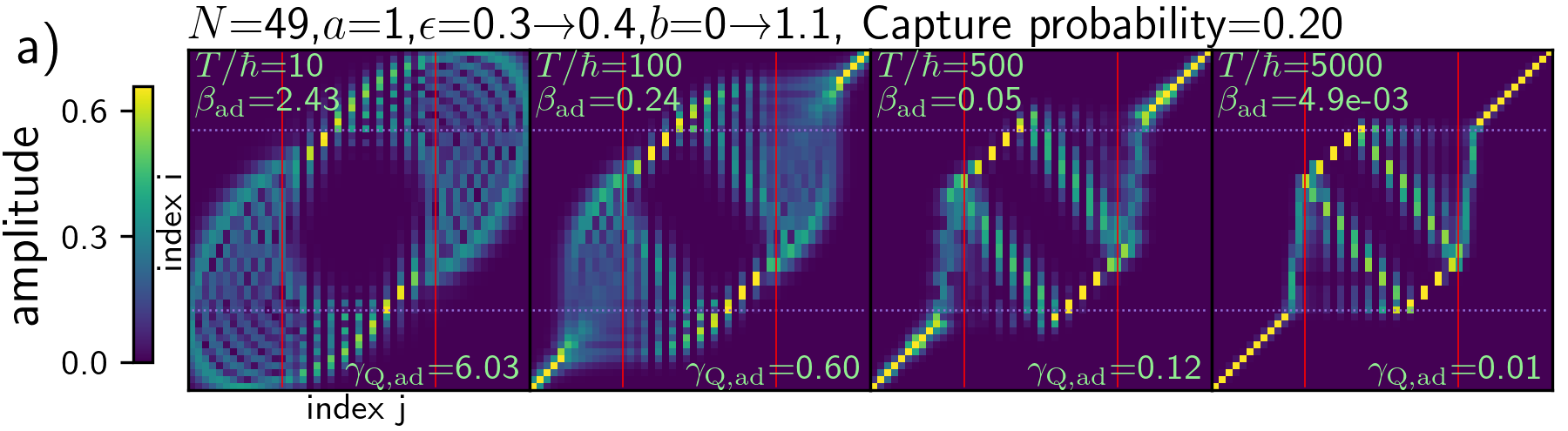}
\includegraphics[width=7truein]{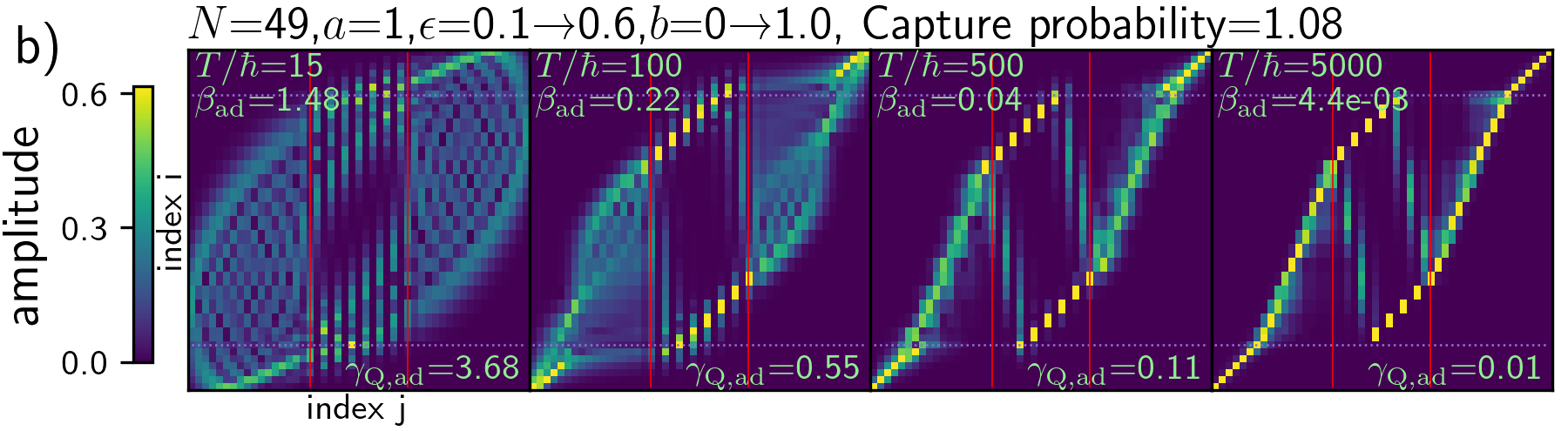}
\caption{
a) Similar to Figure \ref{fig:drift}a except we show transition matrices for the operator 
$\hat h(a, b(t), \epsilon(t)$ 
 with  both $b$ and $\epsilon$ varying.  The probability of capture 
into the librating region is computed with equation \ref{eqn:pcap} (via the KNH theorem) and is printed
on the top of the image.  For this series of 
integrations, the probability of capture is low.   The dimensionless ratio $\beta_{\rm ad}$  
which 
determines whether resonance capture is adiabatic for the associated classical system
(computed via equation \ref{eqn:beta}) is printed on each panel. 
The dimensionless ratio $\gamma_{Q,ad}$ (equation \ref{eqn:gamma_ad})  
describes whether drift is likely to cause Landau-Zener diabatic 
transitions in the librating regions is also printed on each panel. 
The red thin solid vertical lines are at  the indices of the separatrices at the end of the simulation.
The dotted thin horizontal purple lines are at the indices of initial location of the separatrices.  
b)  Similar to a) except $\dot \epsilon$ and $\dot b$ are chosen 
so that the probability of capture into resonance is high. 
\label{fig:drift2}
}
\end{figure*}

\begin{figure}[htbp]\centering 
\includegraphics[width=3.5truein]{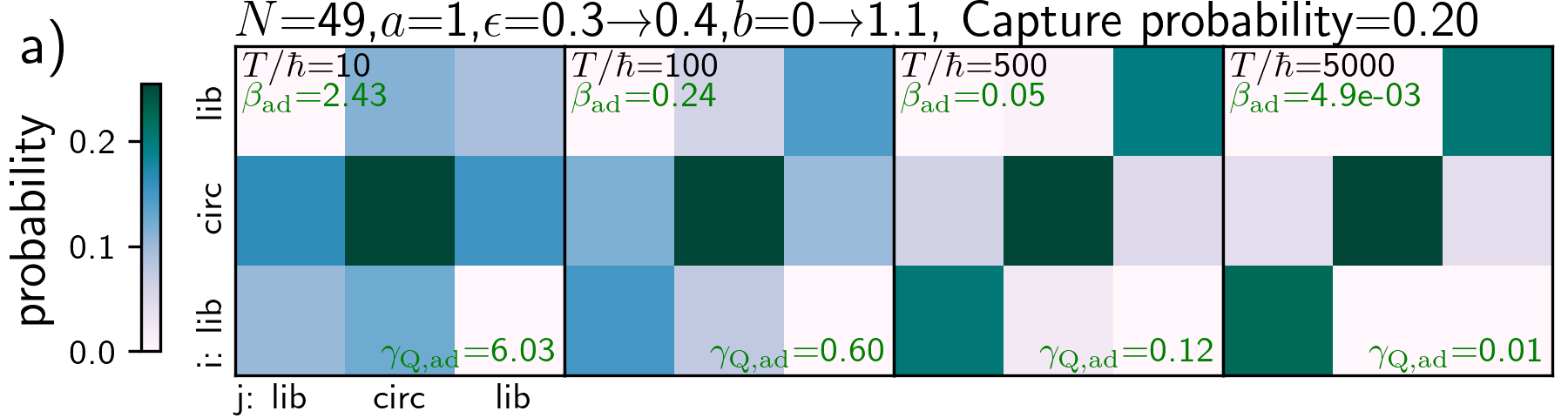}
\includegraphics[width=3.5truein]{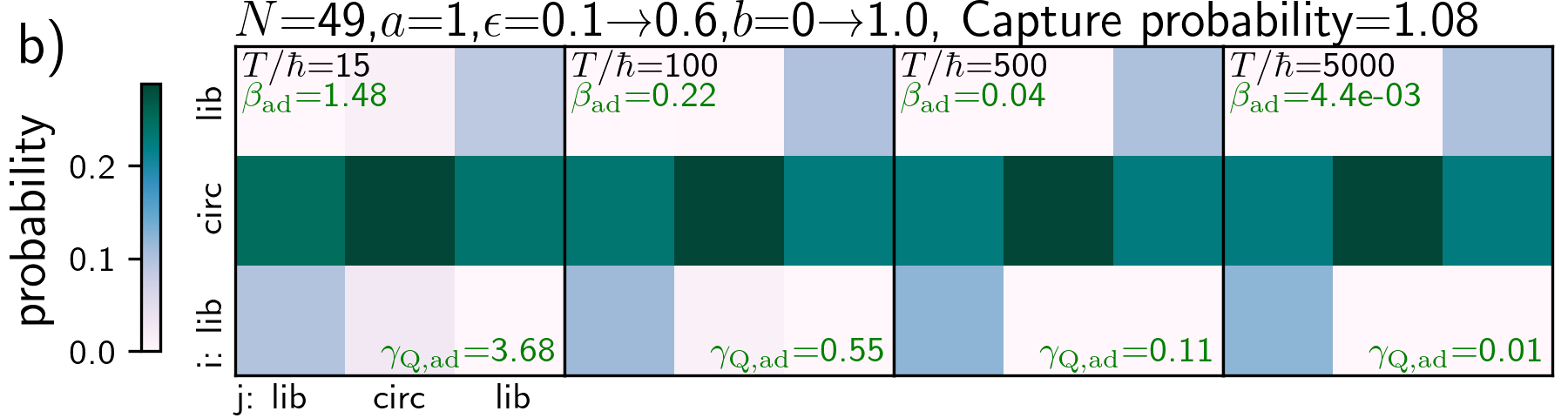}
\caption{a) For the same integrations shown in Figure \ref{fig:drift2}a, we show the probability 
of transition between an initial libration or circulation region and a final libration or circulation region.  
At slow drift rates (on the right) states begun in a librating region remain there and states
begun in a circulating region remain there because the classically estimated probability 
of capture into resonance is low.   At faster and non-adiabatic drift rates (on the left),
 transitions take place between states in circulating and librating regions. 
b) Similar to a) except for the integrations shown in  Figure \ref{fig:drift2}b. 
Because the resonance grows in width and the classically estimated probability of resonance capture is high, there are transitions between states initially in a circulating region 
and those in a librating region.   
At the two slowest drift rates (the rightmost two panels in subfigures a and b) the probabilities of transition between different regions are similar, indicating that drift is both classically adiabatic, $\beta_{\rm ad} <1$ 
and the quantum version (derived by \citet{Stabel_2022}) of the KNH theorem holds. 
\label{fig:probs}
}
\end{figure}

\subsection{Varying both resonance strength and location}

In this section we vary both resonance strength, set by $\epsilon$ and the $b$ parameter that sets 
the location of the resonance.  For different drift rates, 
we compare the transition matrix (equation \ref{eqn:trans}) computed from the propagator  (equation \ref{eqn:UT}) formed from  the time dependent Hamiltonian operator $\hat h(a,b(t), \epsilon(t))$.  
In the classical system, a particle that starts in the circulating region and is later within a 
librating region is said to have been captured into resonance. 
In this section, we compare the transition matrices for parameters that would give a low probability 
resonance capture in the associated classical system to one that would give a high 
probability of capture. 

As shown by \citet{Henrard_1982,Yoder_1979}, Liouville's theorem implies that the probability of capture 
into resonance (provided drift is sufficiently adiabatic; $\beta_{\rm ad} \ll 1$) 
depends on a ratio of the rate that two volumes in phase space vary.  
For a pendulum with Hamiltonian 
$H(p,\phi) = a \frac{p^2}{2} - \epsilon \cos \phi$, 
 the phase space area within the resonance can be computed by integrating momentum as a function of angle 
 inside the separatrix contour.   At the separatrix energy 
 $E_{sep} = a\frac{p^2}{2} - \epsilon \cos \phi  =  \epsilon $ giving 
 $p(\phi) = \sqrt{\frac{2 \epsilon}{a} } \sqrt{1 + \cos \phi}$.   The area within the separatrix orbit 
 \begin{align}
V_{\rm res} = 2\sqrt{\frac{2\epsilon}{a}} \int_0^{2 \pi} \sqrt{1 + \cos \phi }\ d \phi  =  16 \sqrt{\frac{\epsilon}{a}}.
\end{align}
We use the pendulum to estimate the area inside the separatrix 
as it is much easier to integrate than  the Harper model.  For $a$ fixed, 
the rate that the area of phase space within the resonance varies is
 \begin{align}
\dot V_{\rm res} = \frac{8 \dot \epsilon}{\sqrt{\epsilon a }} =  \frac{8 \dot \epsilon}{\omega_0}. \end{align}  
The rate that the upper separatrix sweeps up volume in phase space,  
is $\dot V_+ = 2 \pi \dot b$. 
The probability of capture into resonance (or capture into a libration region from a circulating region) 
based on the KNH theorem is 
\begin{align}
P_c \approx \frac{ \dot V_{\rm res}}{\dot V_+}  
\approx  \frac{\dot \epsilon}{\dot b} \frac{ 4 }{\pi \sqrt{\epsilon a} }  . \label{eqn:pcap}
\end{align}

In Figure \ref{fig:drift2} we show transition matrices computed with 
both $\epsilon$ and $b$ varying linearly in time.  The total distance in $b$ and $\epsilon$ that vary during 
the integration are shown on the top of each subfigure. 
The probability of capture, $P_c$,  
into the librating region is computed via equation \ref{eqn:pcap} using the value of $\epsilon$ 
midway through the integration and is printed on the top of each plot. 
For Figure \ref{fig:drift2}a, the ratio $\dot b/\dot \epsilon$ is sufficiently high 
that the probability of capture into the libration region is low. 
The opposite is true in Figure \ref{fig:drift2}b where the probability of 
 capture into the librating region is high.  
 The dimensionless number $\beta_{\rm ad}$   (equation \ref{eqn:beta}), 
 characterizing whether the drift rate of the associated classical system would be 
 considered adiabatic, is written on each panel.  
 
For Figure \ref{fig:drift2}a the total drift in $b$ is about 1/3 of the distance across the $2 \pi$ range for momentum $\hat p$.  Because the probability of capture is low, most states that are initially within the circulation region 
remain in that region, with a high probability, unless the drift rate is high enough that numerous transitions 
take place.     At low drift rates  (on the right in Figure Figure \ref{fig:drift2}a) states initially with 
the circulating region transition to other states in the same region, whereas states initially within the libration 
region remain in that region. 
In Figure \ref{fig:drift2}b, because the libration region grows in size, states initially in the circulating region 
are likely to transition into the librating region.

The dimensionless ratio $\Gamma$ of the Landau-Zener model is conceptually similar  to 
the dimensionless parameter $\beta_{\rm ad}$ of equation \ref{eqn:beta}. 
The dimensionless quantity $\beta_{\rm ad}$ is approximately a ratio consisting of the time it takes the resonance to drift the distance of a resonance width divided by  the libration period.  The parameter $\Gamma$ 
is approximately  the time it takes to drift the eigenvalues a distance equal to the energy difference of the avoided crossing divided by the period of phase oscillations in the avoided crossing (which  depends on the minimum energy difference and
$\hbar$).  Dimensionless ratios $\beta_{\rm ad}$ and $\Gamma$ are both likely to be important
 in a resonant quantum system, as 
 $\beta_{\rm ad} \ll 1 $ is required for transition probability integrated over 
different regions in phase space to match those predicted by the KNH theorem, 
whereas $\Gamma \ll 1$ is required for suppression of diabatic transitions between states. 

We construct a dimensionless parameter in the spirit of the Landau-Zener model but dependent upon
the energy level spacing in the bottom of the cosine potential well.  For the Harper operator, 
the spacing between energy levels at the bottom of the potential well 
is $\Delta E = \hbar \omega_0 $ with oscillation frequency  
$\omega_0 = \sqrt{a \epsilon}$.  
The $\Gamma$ factor of the Landau-Zener model is the ratio of the square of an energy difference to $\hbar$ times an energy drift rate.   It can be described as the ratio of the time to cross the energy difference 
at a specified energy drift rate 
and the period of phase oscillations caused by this energy difference.    
We construct a similar dimensionless quantity   
 $\Gamma_{\rm ad} = \frac{(\hbar \omega_0)^2} {\hbar \alpha}$ where $\alpha$ 
 is an energy drift rate.  If $b$ drifts by $2 \pi$ then energy varies between its lowest and highest
 values, which for $|\epsilon/a| <1$ is a range of about 2a, giving  $\alpha \sim {a\dot b}/{\pi}$
 (equations \ref{eqn:DotE} and \ref{eqn:alpha} in appendix \ref{ap:LZ}). 
So that we have a parameter that is small when the system drifts adiabatically, we 
define a dimensionless ratio that is the inverse of $\Gamma_{\rm ad} $; 
 \begin{align}
 \gamma_{\rm Q,ad} & = \Gamma_{\rm ad}^{-1}   
 \sim  \frac{\hbar a\dot b}{\pi ( \hbar\omega_0)^2 } 
 \sim \frac{1}{\pi \hbar \epsilon }\frac{\Delta b} {T}  
  \sim \frac{N}{2\pi^2 \epsilon }\frac{\Delta b} {T} ,  \label{eqn:gamma_ad}
 \end{align}
 where we have used $\hbar =  \frac{2 \pi}{N}$ as the 
 effective value of $\hbar$ for the quantized Harper operator \citep{Quillen_2025}. 
 The values of the dimensionless ratio $\gamma_{\rm Q, ad}$ are printed on each panel 
 in Figure \ref{fig:drift2} a, b.  
 
 In Figure \ref{fig:drift2} we confirm that when $\gamma_{\rm Q,ad} \ll 1$ transitions between 
 states in the librating region are unlikely. 
 While $\gamma_{\rm Q, ad}$ is estimated using the energy spacing in the librating region, 
 the difference in energy levels near the separatrix is only a few times smaller (see Figure 
 \ref{fig:drift}b) than that in at the bottom of the potential well in the libration region. 
 Thus $\gamma_{\rm Q, ad}$ can be used to estimate whether transitions are adiabatic 
 or diabatic near the separatrices.   As the probability of capture into resonance depends upon 
 the nature of transitions in the vicinity of the separatrix, $\gamma_{\rm Q, ad}$ could give a
quantum based estimate for whether transitions (or tunneling) between states would cause a deviation 
from the prediction of the KNH theorem. 
 
Due to transitions into superposition states (sometimes described as quantum tunneling), the classically estimated probability of capture 
into resonance (estimated via the KNH theorem) can underestimate the probability of capture
 into resonance. 
\citet{Stabel_2022} found that by summing over probabilities of Landau-Zener transitions, 
the KNH theorem could be modified to take into account both diabatic and adiabatic transition probabilities.  
In Figure \ref{fig:probs} we compute the probabilities of transition between initial libration and circulation 
and final libration and circulation regions for the same integrations shown in Figure \ref{fig:drift2}. 
The transition probabilities are estimated by summing the square of the transition matrix elements 
in regions of the transition matrix that are bounded by the indices of states with energies closest
to those of the separatrix orbits.  
This Figure illustrates that a sum of transition probabilities between regions is sensitive to the drift 
rate, though as both dimensionless parameters $\beta_{\rm ad}$ and $\gamma_{\rm Q, ad}$ decrease as the drift rate decreases, 
it is difficult to separate between classical and quantum non-adiabatic behavior.  
We infer that the KNH theorem holds in the adiabatic limit of both $\beta_{\rm ad} \ll 1$ and 
$\gamma_{\rm Q, ad} \ll 1$ as the probabilities of transitions between 
regions of phase space approach values that are independent 
of drift rate in this limit. 
This inference is consistent with the picture described by \citet{Stabel_2022} generalizing the 
KNH theorem in the quantum setting.  Note that even at the low drift rates on the right hand side
of Figure \ref{fig:drift3} transitions between eigenstates take place within the circulation region 
where the energy levels are nearly degenerate.  Thus the KNH theorem is likely obeyed in the quantum 
system even when transitions into superposition states take place. 
 
 \begin{figure*}[htbp]\centering 
\includegraphics[width=7truein]{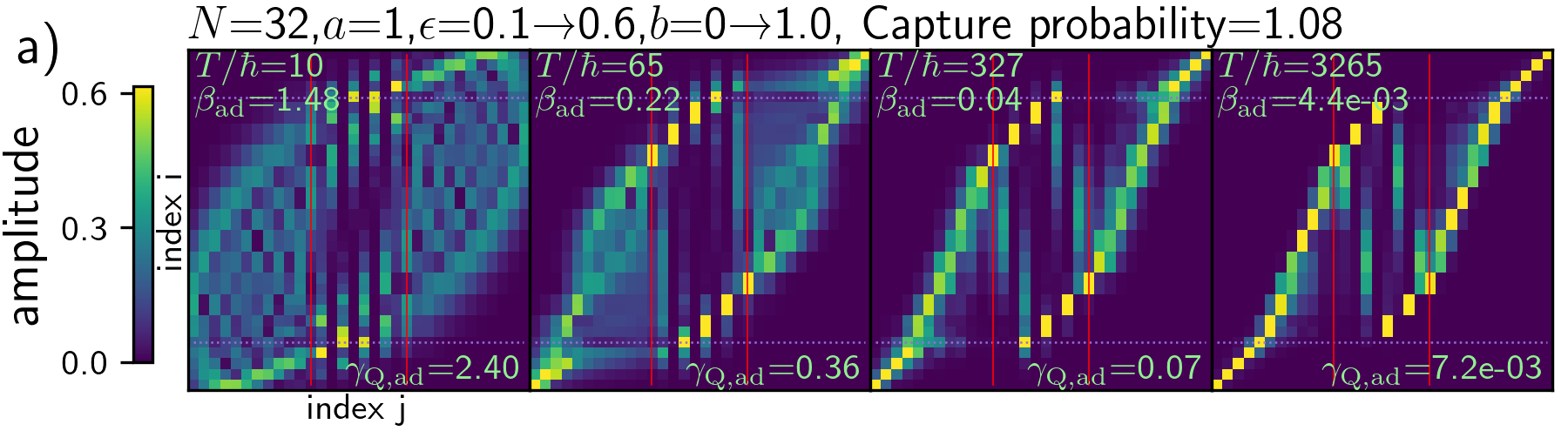}
\includegraphics[width=7truein]{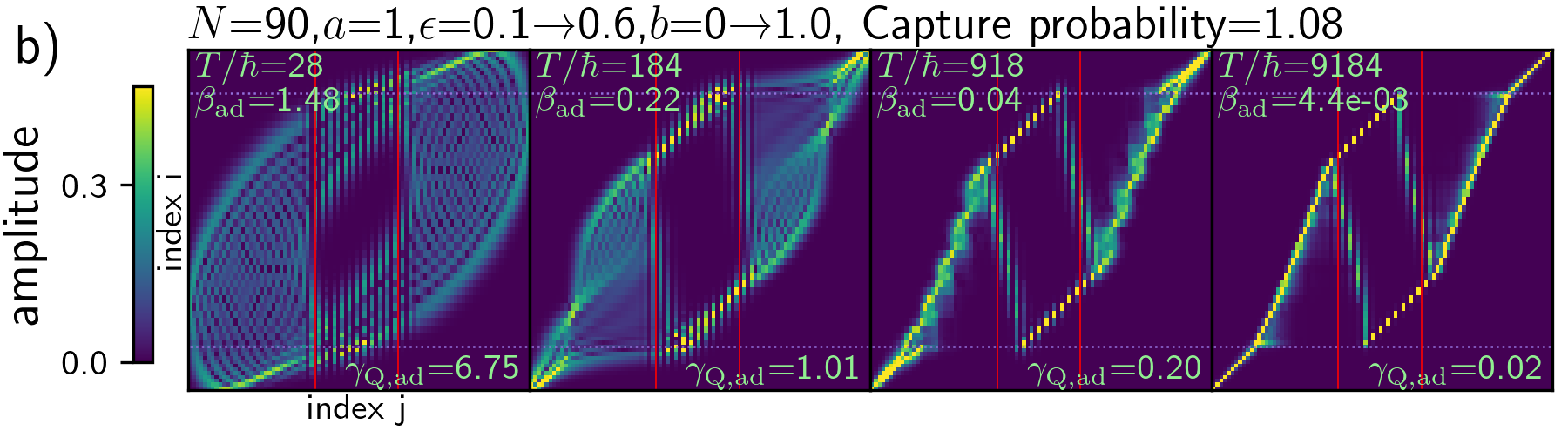}
\caption{ a) Similar to Figure \ref{fig:drift2}b except the dimension of the quantum system $N=32$ differs and is lower than in Figure \ref{fig:drift2}b. The durations $T$ of the drifts are the same as in Figure \ref{fig:drift2}b. 
b) Similar to a) except the dimension $N =90$ is higher than that of a).  Classical quantities $\beta_{\rm ad}$,  characterizing the adiabatic limit, and capture probability $P_{cap}$ are the same in a) and b) but because the effective value of $\hbar$
differs between the two operators, the parameter $\gamma_{\rm Q, ad}$, describing 
the likelihood of Landau-Zener transitions in the librating regions, differs between a) and b). 
The regions of phase space that undergo transitions are similar in the two models and governed
by the classical adiabatic parameter $\beta_{\rm ad}$.   
\label{fig:drift3}
}
\end{figure*}

\begin{figure}[htbp]\centering 
\includegraphics[width=3.5truein]{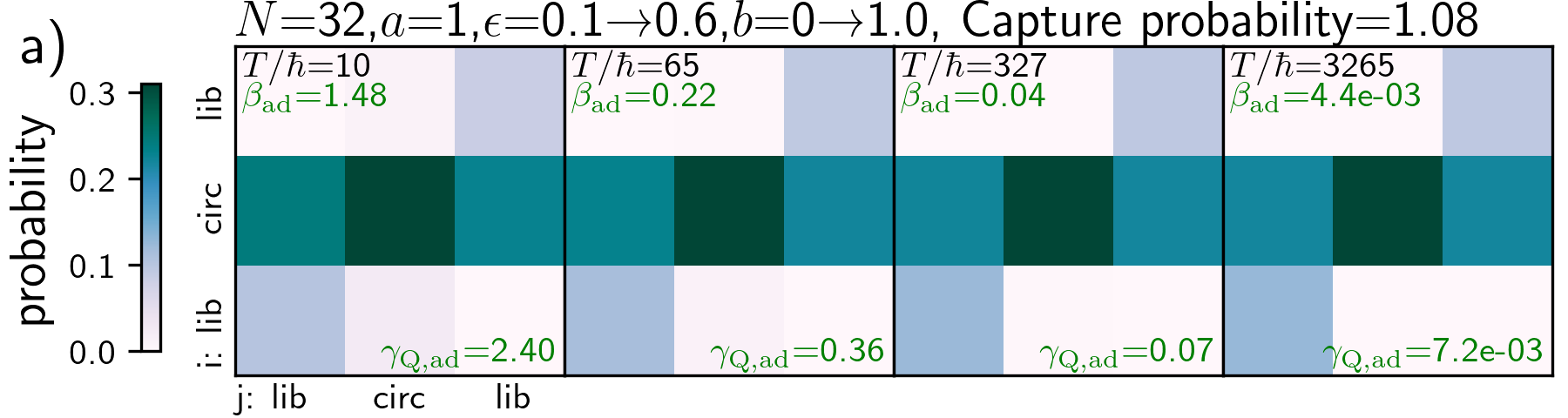}
\includegraphics[width=3.5truein]{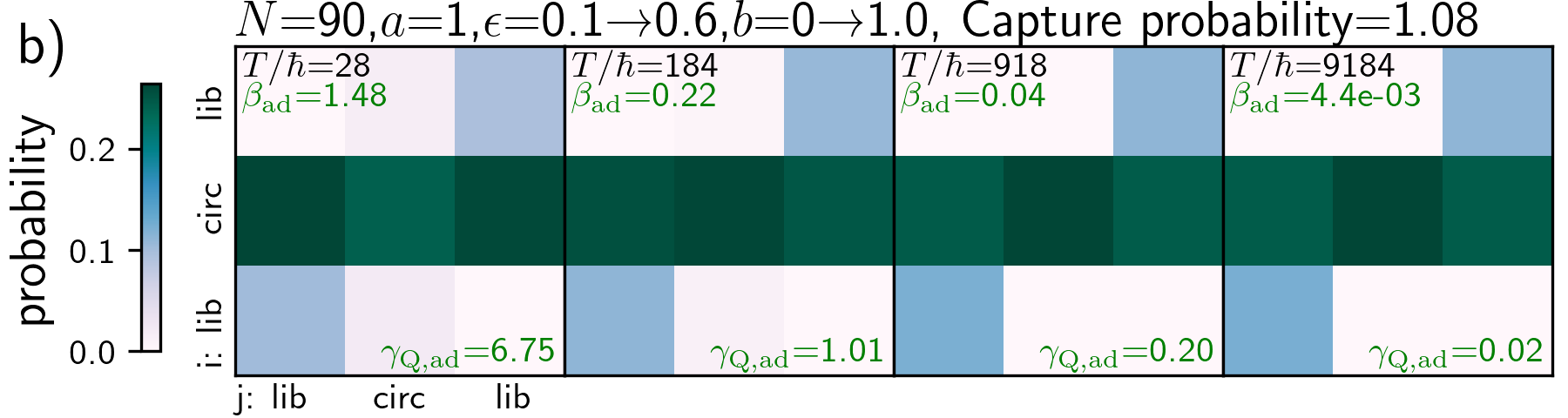}
\caption{a) For the same integrations shown in Figure \ref{fig:drift3}a, we show the probability 
of transition between an initial libration or circulation region and a final libration or circulation region.  
b) Similar to a) except for the integrations shown in  Figure \ref{fig:drift3}b. 
\label{fig:probs3}
}
\end{figure}

Whereas dimensionless parameters $\beta_{\rm ad}$ (characterizing the adiabatic limit) 
and $P_{cap}$ (the probability of capture) are quantities describing the drifting classical model, 
$\gamma_{\rm Q, ad}$  depends on  Planck's constant $\hbar$.   In Figure \ref{fig:drift3} we compare 
transition matrices for the same set of drift durations as shown in Figure \ref{fig:drift2}b that has a high probability of capture into resonance, but with different dimensions $N$ and associated effective values of $\hbar$. 
Figure \ref{fig:drift3} shows that the probabilities of transitions between different regions of phase 
space is insensitive to drift rate as long as both $\gamma_{\rm Q, ad} \ll 1 $ and $\beta_{\rm ad} \ll 1$.  
The transition matrices shown in Figure \ref{fig:drift3} have similar morphology 
at the two different $N$ values, so we attribute to the small differences in transition probabilities
 in Figure \ref{fig:drift3}a, b to  differences in the fraction of states in the different regions at low $N$. 
 
The morphology of the transition matrices in Figure \ref{fig:drift3}a, b are remarkably similar, so we suspect 
that  the classical adiabatic parameter $\beta_{\rm ad}$  governs the 
probability of transitions between regions of phase space, though the Landau-Zener model would 
determine whether diabatic or adiabatic transitions take place between nearby (in energy) quantum states. 

If the volume in phase space remains fixed while the number
of quantum states $N$ is increased,
$\hbar \to 0$ corresponding to the semi-classical limit. 
With large $N$, the parameters of the associated 
classical model remain fixed but the number of bound states within the cosine potential well of the quantized model increases.  For a drifting model, the capture probability and 
 classical adiabatic parameter $\beta_{\rm ad}$ are independent of the number of quantum states $N$, 
 but the quantum adiabatic parameter $\gamma_{\rm Q,ad}$ depends on $\hbar^{-1} \propto N$. 
In the semi-classical limit 
the parameter $\gamma_{\rm  Q, ad} \to 0$.   This implies that transitions between states would be 
increasingly likely, even at low drift rates, but as the energy difference between states decreases, 
transitions would mostly be between states that have similar energy 
 and the dynamics would increasingly resemble that of the associated classical system.  

In contrast, the momentum scale of the torus can be increased as $N$ increases, while maintaining 
the number of bound quantum states within the cosine potential well.  
In appendix  \ref{ap:mathieu} we compare the spectrum of the Mathieu equation to 
that of the Harper model and find that the Harper model spectrum is a good approximation 
for energy levels that are below the center of the circulating region.  
In this limit $\epsilon$ is decreased as $N$ increases (fixing $q$ in equation \ref{eqn:scale_q}, 
as discussed  in appendix  \ref{ap:mathieu}). 
 
\section{Summary and Discussion}

The finite dimensional Harper operator is a simple quantum system on the phase space of a torus that exhibits 
complex behavior when time dependent.  Even though it is remarkably simple 
when written in terms of shift and clock operators (see appendix \ref{ap:defs}), and is related to an integrable 
(non-chaotic) classical system (see Figure \ref{fig:Harp}), the notion of adiabatic drift is non-trivial in the quantum system because of the presence of separatrix orbits in the associated classical system 
that divide phase space into librating and circulating regions. 
Because the operator is finite dimensional, it is straight forward to numerically evaluate its spectrum 
and because it is simply written in terms of operators, a number of symmetries can be exploited 
to aid in studying its spectrum (see appendix \ref{ap:spectrum}). 
%

We find that the nature of spectrum of the quantized system depends upon 
 whether eigenstate energies are in the circulating or librating region of the associated classical system.  
Within the librating region, energy levels are 
well separated, however, in the circulating region, pairs of eigenvalues are nearly 
degenerate.  
When the center of the resonance is drifted, via increasing or decreasing the $b$ parameter
which sets the center of resonance, the spectrum in the circulating region exhibits a lattice of 
avoided energy level crossings (Figures \ref{fig:cc2}, and \ref{fig:cc3}).  In the librating region, 
the energy levels are shifted and the spacing varies to a lesser extent. 
The difference in behavior is explained heuristically in section 
\ref{sec:heu} with a perturbed harmonic oscillator (for the librating region) and rotor (for the circulating region).    
Symmetries of the Harper operator imply that near degeneracies in the spectrum occur at $b$ multiples of $\pi/N$  (see appendixes \ref{ap:bshift}, \ref{ap:derivs}, \ref{ap:degen} and Figure \ref{fig:cc2}). 

Despite the simplicity of the Harper operator, it exhibits an extremely wide range in 
its energy level spacings which affects evolution of the quantum system
when the Hamiltonian operator is time dependent. 
For parameter $b$ (which sets the resonance center) slowly varying, the energy levels of eigenstates in the librating region are not much affected,  
however energy levels within the circulating region 
 undergo a series of close approaches with minimum distance between pairs of eigenvalues 
 that span many orders of magnitude (e.g., see Figure \ref{fig:slice}).  
 The minimum spacing is in the middle of circulating region and depends on a high power 
 of the resonance strength parameter (equation \ref{eqn:min_dist}, appendix \ref{ap:min_spacing}).
 This high power would be difficult to predict with conventional perturbation theory.  
 We computed the propagator of the time dependent Hamiltonian operator for systems that vary the same total amount in $b$ but over different durations. 
For a quantum state initialized in an eigenstate,  transitions take place over a wide range in  drift rates
(see Figure \ref{fig:drift}). 
Only at drift rates many orders of magnitude below the classical adiabatic limit for resonance capture would a system initially begun in any eigenstate remain in one. 

In quantum systems, one notion of adiabatic drift is that drift 
is sufficiently slow that a system initialized  
in an eigenstate state of a time dependent Hamiltonian operator would remain in an eigenstate
of the Hamiltonian.   This notion is consistent with the two level Landau-Zener model. 
In this sense, only at negligible drift rates would the time dependent Harper operator 
behave adiabatically for all possible initial conditions.     This is not inconsistent with the fact pointed out
by \citet{Berry_1984}; adiabatic and semiclassical limits 
give opposite results when two levels in a slowly-changing system pass a near-degeneracy. 

In a classical system containing a separatrix orbit, an alternate notion of adiabatic behavior is 
whether the probability 
for resonance capture is well described by a probability computed via the KNH theorem which is based on 
conservation of volume in phase space (Liouville's theorem). 
At sufficiently slow drift rates, by taking account all transition probabilities,  \citet{Stabel_2022}
showed using a semi-classical limit that a quantum system 
approaches the same resonance capture probability as predicted from the KNH theorem. 
We support, confirm and expand on the work by \citet{Stabel_2022} 
with a complimentary and finite dimensional Hamiltonian model. 
A classical drifting resonance can be described via two dimensionless parameters, 
a capture probability (computed via the KNH theorem) and a dimensionless parameter that describes 
whether the drift is sufficiently slowly that it is considered adiabatic \citep{Quillen_2006} (here the $\beta_{\rm ad}$ parameter  of equation \ref{eqn:beta}). 
For a resonant quantum system we estimate an additional dimensionless parameter $\gamma_{\rm Q, ad}$ (equation \ref{eqn:gamma_ad}), similar to
the $\Gamma$ parameter of the Landau-Zener model, to describe whether states in 
the resonance libration region are sufficiently separated (in energy) that the drift rate would not cause diabatic transitions near the separatrix energy.   
We find that probabilities of transition between different regions of 
phase space seems primarily dependent
upon the classical parameter, with probability approaching a constant 
value for $\beta_{\rm ad} \ll 1 $.  However, transitions in 
the vicinity of the separatrix cease for $\gamma_{\rm Q, ad} \ll 1$.  Because of
the wide range in energy level spacings that are present in a quantized system associated
with a classical one that contains a separatrix orbit, 
transitions between eigenstates take place in the circulation region even if the drift rate
 is sufficiently slow that the KNH theorem should hold. 
In other words, drift can be sufficiently adiabatic that 
 the KNH theorem would be obeyed in a quantum system, taking into 
account multiple transitions near the separatrix region. Yet a system begun in a eigenstate would 
not necessarily remain in one particularly if there is a large range in separations between neighboring eigenvalues.  
The notion of adiabatic behavior in which both quantized and classical systems obey the KNH theorem 
for resonance capture  is consistent with but differs from the notion of adiabatic behavior in which a system begun in an eigenstate remains in an eigenstate of a time dependent Hamiltonian operator.   
In the vicinity of a separatrix in a classical system, the dynamics is not strictly adiabatic. 
The associated quantum manifestation of non-adiabatic behavior near the separatrix energy 
could be the existence of diabatic transitions. 

Expansion of a non-local or long range 
perturbation, such as the gravitational force from a planet, gives a series
 of cosine terms, each associated with an orbital resonance (e.q., \cite{Quillen_2011}). 
Hence cosine potential terms can be ubiquitous in complex classical Hamiltonian systems. 
Many of the symmetries obeyed by the quantized Harper model would also be obeyed by 
finite dimensional quantum 
operators with additional Fourier terms (see appendix \ref{ap:other}). 
There is similarity between 
the spectrum of the finite dimensional Harper operator and the Mathieu equation (appendix \ref{ap:mathieu}). 
This implies that the finite dimensional Harper model and its variants could be used to approximate 
infinite dimensional quantum systems with cosine or more complex non-local potentials.   
Future studies could attempt to determine if there is a type of universality associated 
with resonant quantum systems that is illustrated via the deceptively simple looking Harper operator. 

The quantum pendulum has nearly degenerate eigenvalues with spacing that decreases as index 
$m \to \infty$.  Consequently for any drift rate, no matter how slow,  there would be 
some energy above which diabatic transitions would occur (see appendix \ref{ap:mathieu}). 
Despite the fact that the distance between pairs of energy levels increases with energy ($\propto m^2$ with eigenstate index $m$), 
the energy difference between the states in each pair shrinks rapidly.  The difference drops so rapidly  (via a power law) 
that near degeneracy can be reached at a moderate energy.  
In this sense,  it is amusing to think of the 
the quantized pendulum as an example of a quantum system that has no formal adiabatic limit.  
An implication is that if ionization (or transitions to the circulating region)  can occur  (e.g., 
in transmons \cite{Dumas_2024}) then 
non-adiabatic phenomena could be present over a wide range of possible drift or driving frequencies. 

While there are a number of techniques for placing limits on eigenvalues of an operator using perturbation theory or  with 
the Cauchy interlacing theorem or Weyl's inequalities, (e.g., \cite{Parlett_1998,Bhatia_2007}), it is more challenging 
to place limits on the spacing between eigenvalues (though see \cite{Mignotte_1982,Haviv_1984,Movassagh_2017,Zakrzewski_2023}). 
This motivates studying specific systems, such as the Harper model,
 which could represent a class of resonant models, and for 
 developing additional tools in linear algebra to aid in estimating distances between 
eigenvalues. 

Adiabatic and diabatic phenomena are relevant for design of counter-diabadic protocols or transitionless quantum driving  \citep{Demirplak_2008,Chen_2010,Sels_2017}, and adiabatic computational algorithms.  
Design of varying Hamiltonians that purposely contain 
a range of gap sizes could be used for the opposite effect, giving enhanced numbers of transitions  
for fast effective thermalization.   The system we have studied here is based on an integrable classical system, and 
more complex phenomena is likely to be present in drifting systems that are associated with chaotic classical 
systems (e.g., \cite{Quillen_2025}).

\vskip 1 truein

{\bf Acknowledgements}

We thank Sreedev Manikoth  
%
for helpful  discussions. 

Figures in this manuscript are generated with python notebooks available 
at \url{https://github.com/aquillen/H0drift}.  

\bibliographystyle{elsarticle-harv}
\bibliography{Qchaos}

\appendix 

\section{Properties of the spectrum of the Harper operator}
\label{ap:spectrum}

\subsection{The discrete and shifted Harper/finite almost Mathieu operator}
\label{ap:defs}

Our quantum space is 
 an $N$ dimensional complex vector space, typically with $N>3$.  
In an orthonormal basis denoted $\{ \ket{j}\}: j \in \{0, 1, \ldots ,N-1 \} $, 
the discrete Fourier transform  is 
\begin{align}
 \hat Q_{FT} &= \frac{1}{\sqrt{N}} \sum_{j,k=0}^{N-1} \omega^{jk}\ket{j}\bra{k},  
\end{align}
with complex root of unity
\begin{align}
\omega &\equiv e^{2 \pi i/N}. 
\end{align}
Using the discrete Fourier transform, we construct an orthonormal basis consisting of states 
\begin{align}
\ket{m}_F  = \hat Q_{FT} \ket{m} 
= \frac{1}{\sqrt{N}} \sum_{j=0}^{N-1} \omega^{mj} \ket{j}  \label{eqn:fbasis}
\end{align}
with $m \in \{0, 1, \ldots, N-1\}$. 

Angle and momentum operators are defined using the conventional and Fourier bases (defined in 
equation \ref{eqn:fbasis}), 
\begin{align}
\hat \phi  &= \sum_{j=0}^{N-1} \frac{2 \pi j}{N} \ket{j}\bra{j}\nonumber \\
\hat p & = \sum_{k=0}^{N-1} \frac{2 \pi k}{N} \ket{k}_F\bra{k}_F . \label{eqn:phi_p}
\end{align}
Clock and shift operators  \citep{Schwinger_1960} are defined as 
\begin{align}
\hat Z &= \sum_{j=0}^{N-1} \omega^j \ket{j}\bra{j} = \sum_{k=0}^{N-1} \ket{(k+1)\text{ mod } N}_F \bra{k}_F \nonumber \\
\hat X &= \sum_{k=0}^{N-1} \omega^{-k} \ket{k}_F \bra{k}_F = \sum_{j=0}^{N-1} \ket{(j+1) \text{ mod } N} \bra{j} . \label{eqn:clock_shift}
\end{align}
The clock and shift operators are also called generalized Pauli matrices or Weyl-Heisenberg matrices \citep{Appleby_2005}. 
They obey
 \begin{align}
\hat Z \hat X &= \omega \hat X \hat Z, \qquad   
\hat Z \hat X^\dagger  = \omega^{-1} \hat X^\dagger \hat Z   \label{eqn:ZflipX}
\end{align} and 
\begin{align}
\hat Z& = \hat Q_{FT} \hat X \hat Q_{FT}^\dagger, 
\qquad \hat Z^\dagger = \hat Q_{FT}^\dagger \hat X \hat Q_{FT} \nonumber \\
\hat Z^N &= \hat X^N = \hat I   \label{eqn:QFTstuff}
\end{align}
where $\hat I$ is the identity operator. 

The parity operator 
\begin{align}
\hat P = \sum_{n=0}^{N-1} \ket{n }\bra{ - n \ {\rm mod} \ N} 
= \sum_{k=0}^{N-1} \ket{k }_F \bra{ - k \ {\rm mod} \ N}_F.  \label{eqn:parity}
\end{align}
Trigonometric functions are particularly simple in terms of the clock and shift operators, 
\begin{align}
 \cos \hat  \phi & = \frac{1}{2} (\hat Z + \hat Z^\dagger) \nonumber \\
 \sin \hat \phi & =  \frac{1}{2i} (\hat Z -  \hat Z^\dagger) \nonumber \\
 \cos \hat p & = \frac{1}{2} (\hat X + \hat X^\dagger) \nonumber  \\
 \sin \hat p & = \frac{1}{2i} (- \hat X + \hat X^\dagger).
\end{align}

We are interested in the spectrum and the dynamics of the $N$-dimensional Hermitian operator 
of equation \ref{eqn:h0} which we restate here, 
\begin{align}
\hat h\left(a,b,\epsilon\right) &=  a\cos (\hat p - b) + \epsilon \cos \hat \phi , \label{eqn:h0abe}
\end{align}
and with $\hat p, \hat \phi$ operators defined in equation \ref{eqn:phi_p}. 
The coefficients $a, \epsilon, b$ are real numbers, and typically $a>0$ and $| \epsilon| \le 1$. 
With $b=0$, this operator is equivalent to 
 the finite almost Mathieu operator studied by \citet{Strohmer_2021}.
The parameter $b$ allows us to shift the momentum or kinetic operator term as illustrated
for the classical system in Figure \ref{fig:Harp}. 
Except for its sign and a constant term proportional to the identity operator, the operator $\hat h$ is 
 similar to the operator we studied previously \citep{Quillen_2025} (called $\hat h_0$) 
 that was derived from 
 the classical Harper Hamiltonian \citep{Harper_1955}. 
The operator is traceless; $\tr \hat h (a,b,\epsilon) = 0$.  

In terms of clock and shift operators, the Hermitian operator of equation \ref{eqn:h0abe}
\begin{align}
\hat h \left(a,b,\epsilon\right)
& =  \frac{a}{2} \left( \hat X + \hat X^\dagger\right) \cos b + \frac{a}{2i} \left(- \hat X + \hat X^\dagger\right)\sin b \nonumber  \\
& \ \ \ \ + \frac{\epsilon}{2}\left(\hat Z + \hat Z^\dagger\right)  \nonumber  \\
&=  \frac{a}{2} \hat X e^{ib} + \frac{a}{2} \hat X^\dagger e^{-ib} 
+ \frac{\epsilon}{2} \left(\hat Z + \hat Z^\dagger\right) .  \label{eqn:h0_short}
\end{align}
With $b=0$, $a=\epsilon$, the operator commutes with the discrete Fourier transform operator, 
\begin{align}
 \left[\hat Q_{FT},\hat h\left(a,0,a\right) \right] = 0
\end{align} 
 so the finite almost Mathieu operator helps classify 
eigenstates of the discrete Fourier transform \citep{Dickinson_1982}.

The Hermitian operator $\hat h$ is triagonal with the addition of two additional terms on the top right and lower
left corners and in the form  
\begin{align}
\begin{pmatrix}
\alpha_0 & \beta^*       & 0              &  \cdots & 0 & \beta \\
\beta    & \alpha_1 &  \beta^*       &  \ddots & \vdots & 0 \\
0             & \beta     & \alpha_2  & \ddots & \vdots  & \vdots \\
\vdots     &  \ddots     & \ddots  &    \ddots & \ddots & \vdots  \\
0             &  \cdots     &    \ddots   & \beta & \alpha_{N\!-\!2}  &  \beta^* \\
\beta^*     &  0 &   \cdots  & 0 &  \beta  &  \alpha_{N\!-\!1}
\end{pmatrix}. \label{eqn:tridiag_ab}
\end{align}
In the conventional basis and for $j \in \{0, 1, \ldots, N-1\}$,  the operator 
$\hat h(a,b,\epsilon)$ is a Hermitian  matrix in the form of equation \ref{eqn:tridiag_ab} with coefficients 
\begin{align}
\alpha_j & = \epsilon \cos \left( \frac{2 \pi j}{N} \right) \nonumber \\
\beta & = \frac{a}{2} e^{ib} .  \label{eqn:cofs_j}
\end{align}
In the Fourier basis (defined in equation \ref{eqn:fbasis}) and for $k \in \{0, 1, \ldots, N-1\}$, the operator 
$\hat h(a,b,\epsilon)$  is a  symmetric real matrix in the form of \ref{eqn:tridiag_ab} with coefficients 
\begin{align}
\alpha_k & = a \cos\left( \frac{2 \pi k - b}{N} \right) \nonumber \\
\beta & = \frac{\epsilon}{2} .\label{eqn:cofs_k}
\end{align}

\subsection{How shifts in parameter $b$ affect the spectrum} 
\label{ap:bshift}

From shift, clock and parity operators, we can try to construct an invertible operator $\hat V$ that obeys 
  $\hat V \hat h(a, b, \epsilon) \hat V^{-1} = \hat h(a',b,'\epsilon')$ where $a', b', \epsilon'$ are not
 necessarily the same parameters as $a,b,\epsilon$. 
 If we can find such an operator $\hat V$, 
 then the spectrum of $h(a, b, \epsilon) $ is the same as the spectrum of $\hat h(a',b',\epsilon')$.

\begin{theorem} \label{th:A1}
For the operator $\hat h(a,b,\epsilon)$  in equation \ref{eqn:h0abe}
and $k \in {\mathbb Z}$, the spectrum or set of eigenvalues obeys 
\begin{align}
{\rm spectrum}\left[ \hat h\left(a,   \dfrac{2 \pi k}{N},\epsilon \right)  \right] &= 
{\rm spectrum}\left[ \hat h\left(a, 0,\epsilon \right)  \right] 
.
\end{align} 
\end{theorem}

\begin{proof}
We use short-hand 
 $e^{\frac{2 \pi i k}{N}} = \omega^k,$  and equation \ref{eqn:h0_short} to find 
\begin{align}
\hat h(a, \frac{2\pi k}{N}, \epsilon) & =  \frac{a}{2} \hat X \omega^k + \frac{a}{2} \hat X^\dagger \omega^{-k} 
+ \frac{\epsilon}{2} \left(\hat Z + \hat Z^\dagger\right).
\end{align}
For $k \in {\mathbb Z} $, 
using the relations in equations \ref{eqn:ZflipX} 
\begin{align}
\hat Z^k \hat X \hat Z^{-k} & = \omega^k \hat X \nonumber \\
\hat Z^k \hat X^\dagger \hat  Z^{-k} & = \omega^{-k} \hat X^\dagger.  \label{eqn:ZXZk}
\end{align}
These relations help us compute   
\begin{align}
\hat Z^{-k}  \hat h\left(a,\frac{2 \pi k}{N},\epsilon\right) \hat  Z^{k}  = \hat h \left(a, 0, \epsilon\right) .
 \label{eqn:ZZ}
\end{align}
The operator $\hat h$ is Hermitian, and so it is diagonalizable. 
Since $\hat Z^{k}$ is an invertible operator (that is also unitary),  
equation \ref{eqn:ZZ} implies that the two 
operators in equation \ref{eqn:ZZ} have the same set of eigenvalues or spectrum. 
\end{proof}

\begin{corollary}
If $\ket{n}$ is an eigenstate of $\hat h(a, 0,\epsilon)$ with eigenvalue $\lambda_n$,  
then for $k \in {\mathbb Z}$, 
$Z^k \ket{n}$ is an eigenstate of $\hat h(a, \frac{2\pi k}{N}, \epsilon )$ with the same eigenvalue. 
\end{corollary}
\begin{proof}
Using equation \ref{eqn:ZZ}
\begin{align}
\hat Z^{-k}  \hat h\left(a,\frac{2 \pi k}{N},\epsilon\right) \hat  Z^{k} \ket{n} &=  \hat h \left(a, 0, \epsilon\right) \ket{n} 
= \lambda_n \ket{n} \nonumber \\
 \hat h\left(a,\dfrac{2 \pi k}{N},\epsilon\right) \hat  Z^{k} \ket{n} &= \lambda_n \hat Z^{k} \ket{n}. 
\end{align}
This implies that  $\lambda_n$ is an eigenvalue of $\hat h(a,\frac{2 \pi k}{N},\epsilon)$  
with eigenstate $\hat Z^{k} \ket{n}$.  
\end{proof}

\begin{theorem} \label{th:h0r}
For the operator $\hat h(a,b,\epsilon)$ defined in equation \ref{eqn:h0abe}, 
  $k \in {\mathbb Z} $, and $0\le r < \frac{2 \pi}{N}$, 
\begin{align}
{\rm spectrum}\left[ \hat h\left(a,  r+  \dfrac{2 \pi k}{N},\epsilon \right)  \right] &= 
{\rm spectrum}\left[ \hat h\left(a, r,\epsilon \right)  \right].
\end{align} 
\end{theorem}

\begin{proof}
Equation \ref{eqn:h0_short} and equations \ref{eqn:ZXZk} give 
\begin{align} 
\hat Z^{-k}  \hat h\left(a, \dfrac{2 \pi k}{N}  + r ,\epsilon\right) \hat  Z^{k} 
& =  \frac{a}{2} \left( \hat X \omega^{r} + \hat X^\dagger \omega^{-r}\right)  + \epsilon \cos \hat \phi
\nonumber \\
& =  \hat h\left(a,r,\epsilon\right). \label{eqn:ZZ2}
\end{align}
As $Z^k$ is a unitary operator which is invertible, this implies that both operators in equation \ref{eqn:ZZ2} have the same spectrum. 
\end{proof}

\begin{corollary}
With $0 \le r < \frac{2 \pi}{N}$, 
if $\ket{n}$ is an eigenstate of $\hat h(a, r,\epsilon)$ with 
with eigenvalue $\lambda_n$,  then 
$Z^k \ket{n}$ is an eigenstate of $\hat h(a, r + \frac{2\pi k}{N}, \epsilon )$ with the same eigenvalue. 
\end{corollary}

This corollary follows directly from equation \ref{eqn:ZZ2}. 

\begin{theorem} \label{th:bmb}
For $\hat h$ defined  in equation \ref{eqn:h0abe}
\begin{align}
{\rm spectrum}\left[ \hat h(a,  b,\epsilon )  \right] &= 
{\rm spectrum}\left[ \hat h(a, -b,\epsilon )  \right]  .
\end{align}
\end{theorem}

\begin{proof}
With parity operator $\hat P$ defined in equation \ref{eqn:parity} we find that 
\begin{align}
\hat P \hat Z \hat P &= \hat Z^\dagger \nonumber \\
\hat P \hat X \hat P &= \hat X^\dagger .  \label{eqn:Par_com}
\end{align}
This implies that 
\begin{align}
\hat P \cos \hat  p\  \hat P &= \cos \hat p \nonumber  \\
\hat P \cos \hat  \phi\   \hat P&= \cos \hat  \phi \nonumber \\ 
\hat P \sin \hat  p\  \hat P &= - \sin \hat p \nonumber  \\
\hat P \sin \hat  \phi\   \hat P&=  - \sin \hat  \phi  \label{eqn:PsP}
\end{align}
As the parity operator commutes with both $\cos \hat p$ and $\cos \hat \phi$, 
when $b=0$, 
the Hamiltonian operator $\hat h$ commutes with the parity operator 
\begin{align}
\left[\hat P ,\hat h(a,0,\epsilon)\right] =  0 .  \label{eqn:Ph0_com}
\end{align}

For $b \ne 0 $ 
\begin{align}
\hat P \hat h(a,b,\epsilon)\hat P &=  a \cos \hat p \cos b - a \sin \hat p \sin b + \epsilon \hat \cos \phi
\nonumber \\
& = \hat h (a, -b, \epsilon) . \label{eqn:Pars}
\end{align}
Because the parity operator $\hat P$ is its own inverse, equation \ref{eqn:Pars} implies 
that the two operators in equation \ref{eqn:Pars} have the same spectrum. 
\end{proof}

\begin{corollary} \label{th:A6}
For $k \in {\mathbb Z}$, and $0\le r < \frac{2 \pi}{N}$ 
\begin{align}
{\rm spectrum}\left[ \hat h\left(a,  r + \frac{2 \pi}{N} ,\epsilon \right)  \right] &= 
{\rm spectrum}\left[ \hat h(a, \pm r,\epsilon )  \right]. \label{eqn:combo1}
\end{align}
\end{corollary}

\begin{proof}
A consecutive use of 
theorems \ref{th:bmb} and \ref{th:h0r} together imply that the two operators in equation \ref{eqn:combo1}
have the same spectrum. 
\end{proof}

\begin{corollary}
For $\ket{n}$ an eigenstate of $\hat h(a,b,\epsilon)$ (defined in equation \ref{eqn:h0abe}), 
the state $\hat P \ket{n}$ (where $\hat P$ is the parity operator) 
is an eigenstate of $\hat h(a,-b,\epsilon)$ with the same eigenvalue. 
\end{corollary}

\begin{proof}
We take $\ket{n}$ to be an eigenstate of $\hat h(a,b,\epsilon)$ with eigenvalue $\lambda_n$. 
Using equation \ref{eqn:Pars}
\begin{align*}
\hat P \hat h(a, -b, \epsilon) \hat P \ket{n} &= \hat h(a,b,\epsilon)\ket{n} = \lambda_n \ket{n} \\
\hat h(a, -b ,\epsilon) \hat P \ket{n} & = \lambda_n \hat P \ket{n}.
\end{align*}
This implies that $\hat P \ket{n}$ is an eigenstate of $\hat h(a, -b, \epsilon)$
with the same eigenvalue $\lambda_n$. 
\end{proof}

\begin{theorem} \label{th:hhalf}
For $k \in {\mathbb Z}$, remainder $0 \le \delta < \frac{\pi}{N}$ and operator $\hat h$ defined in equation \ref{eqn:h0abe}
\begin{align}
{\rm spectrum}\left[  \hat h\left(a,  \frac{(2k+1)\pi}{N} + \delta ,\epsilon\right) \right]  = \qquad \qquad \ \  \nonumber \\
\qquad  {\rm spectrum}\left[  \hat h\left(a, \frac{\pi}{N} - \delta ,\epsilon\right) \right].
\label{eqn:half}
\end{align}
\end{theorem}

\begin{proof}
With a calculation similar to equation \ref{eqn:ZZ2}, we find that 
\begin{align}
\hat Z^{\dagger}  \hat h\left(a, \frac{\pi}{N} + \delta,\epsilon\right) \hat Z &= \hat h \left(a, - \frac{\pi}{N} + \delta,\epsilon\right).
\end{align}
Applying the parity operator 
 \begin{align}
\hat P \hat Z^\dagger  \hat h\left(a,\frac{\pi}{N} + \delta,\epsilon\right) \hat Z \hat P  &= \hat h \left(a,  \frac{\pi}{N} - \delta,\epsilon\right). \label{eqn:PZhh}
\end{align}
The product $\hat Z \hat P$  is invertible with inverse $\hat P\hat Z^{\dagger}$ as $\hat Z$ is unitary and $\hat P$ is its own inverse, so equation \ref{eqn:PZhh} 
 implies that the two operators have the same spectrum.   
We then apply theorem \ref{th:h0r} to show that the spectra of the two operators in equation \ref{eqn:half} 
 are equivalent.
\end{proof}

\begin{corollary} \label{th:PZ_com}
\begin{align}
\left[\hat P \hat Z^\dagger, \hat h(a,\frac{\pi}{N} ,\epsilon)\right]  = 0. \label{eqn:pzk}
\end{align} 
\end{corollary}
\begin{proof}
Equation \ref{eqn:pzk} follows directly from equation \ref{eqn:PZhh} using $\delta=0$. 

\end{proof}

\begin{corollary}
If $\ket{n}$ is an eigenstate of  $\hat h(a,  \frac{\pi}{N} + \delta ,\epsilon)$ 
with eigenvalue $\lambda_n$, then  $ \hat P \hat Z^\dagger \ket{n}$ 
is an eigenstate of  $\hat h(a,  \frac{\pi}{N} - \delta ,\epsilon)$ 
with the same eigenvalue. 
\end{corollary}

\begin{proof}
Using equation \ref{eqn:PZhh} 
\begin{align}
\hat Z \hat P    \hat h\left(a,\frac{\pi}{N} - \delta,\epsilon\right)\hat P  \hat Z^\dagger   &= \hat h \left(a,  \frac{\pi}{N} + \delta,\epsilon\right) \nonumber \\
\hat Z\hat P    \hat h\left(a,\frac{\pi}{N} - \delta,\epsilon\right) \hat P\hat Z^\dagger  \ket{n} 
& = \lambda_n \ket{n} \nonumber \\
 \hat h\left(a,\frac{\pi}{N} - \delta,\epsilon\right) \hat P\hat Z^\dagger \ket{n} & = \lambda_n 
 \hat P\hat Z^\dagger \ket{n}.
\end{align}
This  shows that $\hat P\hat Z^\dagger \ket{n}$ is an eigenstate of $ \hat h\left(a,\frac{\pi}{N} - \delta,\epsilon\right)$
with eigenvalue $\lambda_n$. 
\end{proof}

\begin{theorem} \label{th:sym1}  
for $k \in \mathbb{Z}$, 
\begin{align}
\left[ \hat P\hat Z^{-2k}, \hat h\left(a, \frac{2 \pi k}{N}, \epsilon\right) \right] = 0 . \label{eqn:Pz2k}
\end{align}
\end{theorem}
\begin{proof}
Equation \ref{eqn:ZZ} gives 
\begin{align}
\hat Z^{-k} h\left( a,\frac{2 \pi k}{N}, \epsilon\right) \hat Z^k = \hat h(a,0,\epsilon) .
\end{align}
We use the fact that parity $\hat P$ commutes with $\hat h(a,0,\epsilon)$ (equation \ref{eqn:Ph0_com}) 
giving
\begin{align}
\hat P \hat Z^{-k} h\left( a,\frac{2 \pi k}{N}, \epsilon\right) \hat Z^k \hat P& = \hat h(a,0,\epsilon) \nonumber \\
& =  \hat Z^{-k} \hat h\left( a,\frac{2 \pi k}{N}, \epsilon\right)  \hat Z^k \end{align}
and \begin{align}
\hat Z^{k} \hat P Z^{-k} h\left( a,\frac{2 \pi k}{N}, \epsilon\right) \hat Z^k \hat P \hat Z^{-k} &=
\hat h\left( a,\frac{2 \pi k}{N}, \epsilon\right) .
\end{align}
We use the fact that $\hat P \hat Z = \hat Z^\dagger\hat P  $ (from equation \ref{eqn:Par_com}) to 
give 
\begin{align}
\hat P  Z^{-2k} h\left( a,\frac{2 \pi k}{N}, \epsilon\right) \hat Z^{2k} \hat P  &=
\hat h\left( a,\frac{2 \pi k}{N}, \epsilon\right) 
\end{align}
which gives the commutator of equation \ref{eqn:Pz2k}.
\end{proof}

This symmetry is related to the close approaches between eigenvalues at $b$ a multiple of $2 \pi/N$. 

\begin{theorem} \label{th:sym2}  
For $k \in \mathbb{Z}$, 
\begin{align}
\left[ \hat P\hat Z^{-(2k+1)}, \hat h\left(a, \frac{ \pi (2k+1)}{N}, \epsilon\right) \right] = 0 . \label{eqn:Pz1k}
\end{align}
\end{theorem}

\begin{proof}
Using equation \ref{eqn:ZZ2} with $r=\frac{\pi}{N}$  and with $k \in \mathbb{Z}$ 
\begin{align}
\hat Z^{-k} \hat h\left( a,\frac{ \pi(2 k+1)}{N}, \epsilon\right) \hat Z^{k} & = \hat h\left(a, \frac{\pi}{N}, \epsilon\right) .
\end{align}
Using Equation \ref{eqn:pzk}
\begin{align}
\hat P \hat Z^{-(k+1)} \hat  h \Big( a, & \frac{ \pi(2 k+1)}{N}, \epsilon\Big) \hat Z^{k+1} \hat P   = 
\hat h\left(a, \frac{\pi}{N}, \epsilon \right) \nonumber \\
& = \hat Z^{-k} \hat  h\left( a,\frac{ \pi(2 k+1)}{N}, \epsilon \right) \hat Z^{k} . 
\end{align}
We multiply both sizes by factors of $\hat Z$, 
\begin{align}
Z^k \hat P \hat Z^{-(k\!+\!1)}  \hat  h\left( a,\!\frac{ \pi(2 k\!+\!1)}{N}, \epsilon \right) & \hat Z^{k\!+\!1} \hat P \hat Z^{-k}  \nonumber \\
&\!=  \hat  h\left( a,\!\frac{ \pi(2 k\!+\!1)}{N}, \epsilon \right)\nonumber .
\end{align}
We commute the parity operator 
\begin{align}
\hat P \hat Z^{-(2k+1)}  \hat  h\left( a,\!\frac{ \pi(2 k\!+\!1)}{N}, \epsilon\right)  \hat Z^{2k\!+\!1} \hat P& = 
\hat  h\left( a,\frac{ \pi(2 k\!+\!1)}{N}, \epsilon\right)
\end{align} 
and this gives the commutator in equation \ref{eqn:Pz1k}.
\end{proof}

This symmetry is related to the close approaches between eigenvalues at $b$ at odd multiples of $\pi/N$. 

\begin{theorem} \label{th:Q}
The spectrum of $\hat h(a,0,\epsilon)$ is the same as the spectrum of 
$\hat h(\epsilon,0,a)$.  
\end{theorem}
\begin{proof}
We apply the discrete Fourier transform (equation \ref{eqn:QFTstuff}), giving
\begin{align}
\hat Q_{FT} \hat h(a,0,\epsilon)\hat Q_{FT}^{-1} = \hat h(\epsilon,0,a) . 
\end{align}
This implies that the spectrum of $\hat h(a,0,\epsilon)$ is the same as that of 
 $\hat h(\epsilon,0,a)$. 
 \end{proof}

\subsection{The derivatives of the eigenvalues with respect to drift parameter $b$}
\label{ap:derivs}

We consider how the eigenvalues of $\hat h(a,b,\epsilon)$ (defined in equation \ref{eqn:h0abe})
 vary when $b$  varies.   
We take $\ket{n(b)}$ to be an eigenstate that satisfies 
$\hat h(a,b,\epsilon)\ket{n(b)} = \lambda_n(b) \ket{n(b)}$ 
with eigenvalue $\lambda_n(b)$. 

\begin{theorem}
For $k \in {\mathbb Z}$, each 
eigenvalue $\lambda_n(b) $ of $\hat h(a,b,\epsilon)$ (defined in equation \ref{eqn:h0abe}) that remains distinct from other  eigenvalues obeys 
\begin{align}
\frac{\partial}{\partial b} \lambda_n(b) \Bigg|_{b = \frac{k \pi}{N}} = 0 . 
\end{align}
\end{theorem}

\begin{proof}

Using equation \ref{eqn:h0_short} 
\begin{align}
\frac{\partial}{ \partial b} \hat h(a,b,\epsilon) & = a \sin (\hat p - b) \\
\frac{\partial}{ \partial b} h(a,b,\epsilon)  \Bigg|_{b=0} & =   a \sin \hat p . \label{eqn:db0}
\end{align}
We apply the Hellman-Feynman theorem (e.g., \cite{Squillante_2023}) 
which gives the expression 
\begin{align}
 \bra{n(b)} \frac{ \partial}{\partial b}  \hat h(a,b,\epsilon) \ket{n(b)} &= \frac{\partial}{\partial b} \lambda_n(b).
 \label{eqn:HF} 
\end{align}
At $b=0$ 
\begin{align}
\bra{n(0)} \frac{ \partial}{\partial b}  \hat h(a,b,\epsilon) \Bigg|_{b=0} \ket{n(0)} &=
\frac{\partial}{\partial b} \lambda_n(b) \Bigg|_{b=0}.
\end{align}
Using equation \ref{eqn:db0} 
\begin{align} 
\bra{n(0) } a \sin \hat p \ket{n(0)} = \frac{\partial}{\partial b} \lambda_n(b) \Bigg|_{b=0} . \label{eqn:vertb}
\end{align}
We previously showed (equation \ref{eqn:Ph0_com}) that for $b=0$, $\hat h$ commutes with the 
parity operator $\hat P$.   
Hence  an eigenstate with a distinct eigenvalue must also be an eigenstate of the parity operator
which has eigenvalues of $\pm 1$.  That implies that 
\begin{align}
\bra{n(0) } \hat P a \sin \hat p  \hat P \ket{n(0)}   = \bra{n(0) }  a \sin \hat p   \ket{n(0)} .
\end{align} 
However we also showed previously 
that $\hat P \sin \hat p \hat P = - \sin \hat p$ (equation \ref{eqn:PsP}). 
This implies that 
\begin{align}
\bra{n(0) } \hat P a \sin \hat p  \hat P \ket{n(0)}  = - \bra{n(0) }  a \sin \hat p   \ket{n(0)}
\end{align}
Together these imply that 
\begin{align}
\bra{n(0) }  a \sin \hat p   \ket{n(0)}  = - \bra{n(0) }  a \sin \hat p   \ket{n(0)} = 0 . 
\end{align}
Hence equation \ref{eqn:vertb} gives 
\begin{align}
\frac{\partial}{\partial b} \lambda_n(b) \Bigg|_{b=0} =0. 
\end{align}
We extend this to $b$ a multiple of $2 \pi/N$ using theorem \ref{th:h0r}. 

It is convenient to compute 
\begin{align}
\hat Z \sin (\hat p-b) \hat Z^\dagger & = \sin ( \hat p + \frac{2\pi}{N} - b)  \label{eqn:Zsin}.
\end{align}

We find that $\hat P \hat Z^\dagger $ commutes with $\hat h(a,\frac{\pi}{N}, \epsilon)$
which means $\hat P \hat Z^\dagger $ and $\hat h(a,\frac{\pi}{N}, \epsilon)$ are simultaneously diagonalizable. 
Because $\hat P \hat Z^\dagger  $ is a unitary operator, its eigenvalues are complex
numbers with magnitude 1 so they cancel in the following expression 
\begin{align}
\bra{n} \hat Z \hat P \sin (\hat p + \frac{\pi}{N}) \hat P \hat Z^\dagger  \ket{n}& = \bra{n}\sin (\hat p + \frac{\pi}{N})  \ket{n}. \label{eqn:bbb1}
\end{align}
We apply the parity operator to the sine and equation \ref{eqn:Zsin} on the right hand side,  
\begin{align}
 -  \bra{n} \hat Z \sin (\hat p - \frac{\pi}{N}) \hat Z^\dagger   \ket{n} 
& = - \bra{n} \sin (\hat p + \frac{\pi}{N}) \ket{n} .  \label{eqn:bbb2}
\end{align}
Together  equations \ref{eqn:bbb1} and  \ref{eqn:bbb2} imply that 
\begin{align}
\bra{n} \sin (\hat p + \frac{\pi}{N}) \ket{n} = 0 . 
\end{align}
Hence the Feynman Hellman theorem (equation \ref{eqn:HF}) gives 
\begin{align}
\frac{\partial}{\partial b} \lambda_n(b) \Bigg|_{b= \frac{\pi}{N}} =0. 
\end{align}
We extend this relation to $b$ equal to odd multiples of $ \pi/N$ using theorem \ref{th:h0r}. 
\end{proof}

\subsection{Even dimension $N$} \label{ap:even}

\begin{theorem} \label{th:even}
If dimension $N$ is even, then for each eigenvalue $\lambda_n$ of $\hat h(a,b,\epsilon)$ (defined in equation \ref{eqn:h0abe}), 
$- \lambda_n$ is also an eigenvalue of $\hat h(a,b,\epsilon)$.  
\end{theorem}

\begin{proof}
If $N$ is even then $k=N/2$ is an integer.  The factor $\omega^\frac{N}{2} = -1$. 
This and equation \ref{eqn:ZXZk} give 
\begin{align}
\hat X^\frac{N}{2} \hat Z^\frac{N}{2}  \hat h(a,b,\epsilon) \hat X^{-\frac{N}{2} } \hat Z^{-\frac{N}{2}} = - \hat h(a,b,\epsilon)
\end{align}
This implies that  $\hat h $ has the same spectrum as $- \hat h$ and consequently 
that if $\lambda_n$ is an eigenvalue of $\hat h$ then so is $- \lambda_n$. 
\end{proof}

\begin{theorem} \label{th:even2}
If dimension $N$ is even, then the spectrum of $\hat h(a,b,\epsilon)$ is equal to 
the spectrum of $\hat h(a,b,-\epsilon)$.   
\end{theorem}

\begin{proof}
If $N$ is even then $k=N/2$ is an integer and the factor $\omega^\frac{N}{2} = -1$. 
We find that 
\begin{align}
\hat X^{-\frac{N}{2}} \left(\hat Z + \hat Z^\dagger\right) \hat X^{\frac{N}{2}} = - \left(\hat Z + \hat Z^\dagger\right).
\end{align}
Hence 
\begin{align}
\hat X^{-\frac{N}{2}} \hat h(a,b,\epsilon) \hat X^\frac{N}{2} = \hat h(a,b, -\epsilon). \label{eqn:xhx}
\end{align}
The operator $\hat X^{\frac{N}{2}}$ is invertible, so equation \ref{eqn:xhx} 
 implies that the two operators $\hat h(a,b,\epsilon)$ and $\hat h(a,b,-\epsilon)$ have the same spectrum. 
\end{proof}

\subsection{Odd dimension $N$}

\begin{theorem}
If dimension $N$ is odd then for each eigenvalue $\lambda_n$ of $\hat h(a,b,\epsilon)$ (defined in equation \ref{eqn:h0abe}), 
$- \lambda_n$ is an eigenvalue of $\hat h(a,b + \frac{\pi}{N},-\epsilon)$. 
\end{theorem}

\begin{proof}
If $N$ is odd we take $\tilde k=(N-1)/2$ which is an integer. 
Equation  \ref{eqn:ZXZk} gives 
\begin{align}
\hat Z^{\tilde k} \hat X \hat Z^{- \tilde k} &=  \hat X \omega^{\tilde k } = \hat X
e^{2 \pi i \frac{ (N-1)}{2 N} } = - \hat X  \omega^{- \frac{1}{2}} \\
\hat Z^{\tilde k} \hat X^\dagger \hat Z^{- \tilde k} &=  \hat X^\dagger \omega^{-\tilde k } = 
-  \hat X^\dagger   \omega^\frac{1}{2}. 
\end{align}
We find that 
\begin{align}
\hat Z^{\tilde k} \cos \hat p \hat Z^{-\tilde k}  = - \cos\left( \hat p + \frac{\pi}{N}\right) \\
\hat Z^{\tilde k} \sin \hat p \hat Z^{-\tilde k}  = - \sin\left( \hat p + \frac{\pi}{N}\right) .
\end{align}
This lets us compute 
\begin{align}
\hat Z^{\tilde k} \cos (\hat p-b) \hat Z^{-\tilde k}   & = 
\hat Z^{\tilde k} \cos \hat p \cos b \hat Z^{-\tilde k}  + \hat Z^{\tilde k} \sin \hat p \sin b \hat Z^{-\tilde k} \nonumber \\
& = - \cos (\hat  p -b + \frac{\pi}{N}) .
\end{align}
Consequently 
\begin{align}
\hat Z^{\tilde k} \hat h (a,b,\epsilon) \hat Z^{-\tilde k} = -\hat h\left(a, b+ \frac{\pi}{N} , -\epsilon\right)
\end{align}
for $\tilde k=(N-1)/2$. 
The spectrum of  $h (a,b,\epsilon)$ is the same as that of $-\hat h\left(a, b+ \frac{\pi}{N} , -\epsilon\right)$. 

\end{proof}

\subsection{Degeneracy of eigenvalues} \label{ap:degen} 

Studying the $\epsilon = a$ case, \citet{Dickinson_1982} conjectured that 
energy levels of $\hat h(a,0,a)$ (with $\hat h$ defined in equation \ref{eqn:h0abe})
are not degenerate except if $N$ is a multiple of 4 and in that case only for a pair of 
eigenstates with zero energy.   Via numerical calculations, we have confirmed this conjecture and extend it 
 to the case of $|\epsilon| \ne a$ with $a,\epsilon \ne 0$. 
The existence of a pair of zero eigenvalues in the case of $N$ a multiple of 4 is shown in appendix  
\ref{ap:det} below (lemma \ref{th:zeros}). 

When $\epsilon = 0, b=0, a \ne 0$, there are multiple pairs of degenerate eigenstates 
as the eigenvalues are $a \cos \frac{2 \pi k}{N}$ for $k \in \{0, 1, \ldots, N-1\}$. 
If $N$ is even and $\epsilon=0,b=0,a\ne 0$,  there are $N/2 - 1$ pairs of degenerate eigenvalues and 2 non-degenerate ones. 
If  $N$ is odd and $\epsilon=0,b=0,a\ne 0$, there are $(N-1)/2$  pairs of degenerate eigenvalues and 1 non-degenerate one. 

The Cauchy interlacing theorem was used  to show that the eigenvalues of the Harper operator $\hat h(a,0,a)$  have multiplicity at most 2 \citep{Dickinson_1982}.   We review this argument and extend it to the more
general case 
$\hat h(a,b,\epsilon)$  for $a \epsilon \ne 0$. 

A case of the 
the Poincar\'e separation theorem, also known as the Cauchy interlacing theorem,  (e.g, \cite{Bhatia_2007})
 is the following:
Let $A$ be an $n\times n$ Hermitian operator and let $B$ be a principal submatrix
of $A$.  
The eigenvalues of $A$ in decreasing order are  $\lambda_1 \ge \lambda_2 \ge \hdots \ge \lambda_n$
 and the  eigenvalues of $B$ are $\mu_1 \ge \mu_2 \ge \hdots \ge \mu_{n-1} $. 
Then for $j = 1,2,..., n-1$, 
\begin{align}
\lambda_j \ge \mu_j \ge \lambda_{j+1}.
\end{align}

To show that the eigenvalues of the Harper operator 
$\hat h(a,b,\epsilon)$ have at most a multiplicity of 2 (following \cite{Dickinson_1982}) we consider 
the principal submatrix (lacking the top row and left column) of $\hat h$. 
In the conventional basis, and if $\epsilon \ne 0$, 
this principal submatrix is a tridiagonal matrix with non-zero off-diagonal elements.
A real symmetric tridiagonal matrix with non-zero off-diagonal elements has distinct eigenvalues
\citep{Parlett_1998}.  Hence the principal submatrix of $\hat h(a,b,\epsilon)$ has distinct eigenvalues. 
Application of
the Cauchy interlacing theorem then implies that the eigenvalues of $\hat h$ have at most
a multiplicity of 2.  Furthermore if there is a pair of eigenvalues of multiplicity 2, then the 
principal submatrix also has (or inherits) this same eigenvalue.  

\begin{theorem} \label{th:distinct}
If $b$ is not a multiple of $ \pi/N$  and $a \epsilon \ne 0$, 
 then the eigenvalues of $\hat h (a,b,\epsilon)$ 
(defined in equation \ref{eqn:h0abe}) are distinct. 
\end{theorem}

\begin{proof}
The periodic almost tridiagonal $n\times n$ matrix with two additional components, on the top right and lower left 
\begin{align}
\hat M = 
\begin{pmatrix}
a_1 & b_1    &            &   c_0       \\
c_1 & \ddots &  \ddots &           \\
       &  \ddots &  \ddots &  b_{n-1} \\
b_n    &              &  c_{n-1}  & a_n   \\
 \end{pmatrix} \label{eqn:tridiag_p}
\end{align}
has determinant that can be written in terms of a product of 2$\times$2 matrices; 
\begin{align}
\det \hat M&  =  \tr \left[  
	\begin{pmatrix}
	a_n & \!\! -b_{n\!-\!1} c_{n\!-\!1} \\
	1 & 0 
	\end{pmatrix}
	\hdots 
	\begin{pmatrix}
	a_2 & \!-b_{1} c_{1} \\
	1 & 0 
	\end{pmatrix}
	\begin{pmatrix}
	a_1 &\!-b_n c_0 \\
	1 & 0 
	\end{pmatrix}
  \right] \nonumber  \\
&  + (-1)^{n+1}\left( \prod_{j=1}^n b_j + \prod_{j=0}^{n-1} c_j \right) \label{eqn:detm}
\end{align}
\citep{Molinari_2008}.

We consider the matrix  
\begin{align}
\hat M = x \hat I - 2 \hat h(1,b,\epsilon). \label{eqn:hatM}
\end{align} 
Here $x$ is a number and $\hat I$ is the identity matrix. 
The matrix $\hat M$ (of equation \ref{eqn:hatM}) is a nearly tridiagonal matrix in 
the form of equation \ref{eqn:tridiag_p} with matrix components in the conventional basis 
(see equation \ref{eqn:cofs_j})
\begin{align}
a_j & =  x - 2\epsilon \cos ( \frac{2 \pi (j-1)}{N} )  \nonumber \\
b_j & = -e^{-ib} \nonumber   \\
c_j & = -e^{ib} 
\end{align}
for $ j  \in \{1, 2, .... , N\}$.  
With equation \ref{eqn:detm},  we compute the determinant for $\hat M$ of equation \ref{eqn:hatM}
\begin{align}
\det \left[x\hat I - 2 \hat h \right]& = 
\tr \Bigg[  
	\begin{pmatrix}
	 x \! -\! 2\epsilon \cos ( \frac{2 \pi (N\!-\!1)}{N} ) & \!\! -1 \\ 1 & 0 
	\end{pmatrix} \times \nonumber \\
	& \ \ \ \ \ \ \ \ \ \begin{pmatrix}
	 x \!-\! 2\epsilon \cos ( \frac{2 \pi (N\!-\!2)}{N} ) & \!\! -1 \\1 & 0 
	\end{pmatrix} \times \nonumber \\
	& \ \ \ \ \  \hdots 
	\begin{pmatrix}
	x- 2 \epsilon \cos ( \frac{2 \pi }{N} ) & -1 \\1 & 0 
	\end{pmatrix}  \times \nonumber \\
	& \ \ \ \ \ \ \ \ \ \ 
	\begin{pmatrix}
	x- 2 \epsilon &-1 \\1 & 0 
	\end{pmatrix}
  \Bigg] \nonumber  \\
 &\ \ \ \   - 2 \cos (Nb) .  \label{eqn:det1}
\end{align}

The traced term in equation \ref{eqn:det1} is a monic polynomial of degree $N$ that 
does not depend on $b$  (but does depend upon $\epsilon$) that we call $f(x)$ so that 
\begin{align}
\det \left[x \hat I - 2 \hat h(1,b,\epsilon) \right]& =  f(x) - 2 \cos (Nb).
\end{align}
The right hand side is the characteristic polynomial of the matrix $2 \hat h(1,b,\epsilon)$. 

Because the operator $\hat h$ is Hermitian, it must have $N$ real eigenvalues and this is true 
for any $b\in \mathbb{R}$.  

Eigenvalues of $2 \hat h(1,b,\epsilon)$ are values of $x$ that are roots 
of equation \ref{eqn:det1};  they satisfy $f(x) - 2 \cos (Nb) = 0 $. 

Following the argument presented by \citet{Molinari_lec6} in the context of a different matrix, we  
take $x$ real and consider the curve in the real $xy$ plane described by $y = f(x)$. 
At each eigenvalue $\lambda_k$ of $\hat h(1,b,\epsilon)$, the value of $x = 2 \lambda_k$ gives an 
intersection between the curve $y = f(x) $ and  the $y= 2 \cos (Nb)$ horizontal line. 
The quantity $2 \cos (Nb)$ is confined to the interval $[-2,2]$. 
A maximum or minimum  in the curve $f(x)$ can only exist at 
 $|y| \ge 2$.  Otherwise there would be a value of $b$ giving two fewer intersections of the curve 
$y=f(x)$ with the horizontal line $y=2 \cos (Nb)$ and this would give two fewer than $N$ roots of the characteristic polynomial. 
Because $f(x)$ cannot contain maxima or minima at $|f(x)| <2$, we find the following: 
As  long as $b$ is not a multiple 
of $ \pi/N$ (giving $\cos (Nb) = \pm 1$) then there must be $N$ intersections of $f(x)$ with the 
horizontal line $y=2 \cos (Nb)$.   This implies that the roots  of $f(x) - 2 \cos (Nb)$ 
must be distinct.  This in turn implies that the roots of 
$\hat h(1,b,\epsilon)$  are distinct for $b$ not equal to a multiple of $\pi/N$.   We can multiply 
by any $a \ne 0$ to find that the same is true for $\hat h(a,b,\epsilon)$. 

We check that we have no contradiction in the case that $\epsilon=0$ but $a\ne 0$. 
If $\epsilon = 0$ then the eigenvalues are $a\cos( \frac{2 \pi k}{N} - b) $ for $k \in \{0, 1, .... N-1\}$. 
There is no degeneracy as long as $b$ is not a multiple of $\pi/N$. 

\end{proof}

Note that multiplicity 2 eigenvalues of $\hat h$ could exist for $b$ a multiple of $ \pi/N$.    Numerically we have
found that a multiplicity 2 pair eigenvalues of $\hat h$ only exist for $N$ a multiple of 4, $b=0$, and when both eigenvalues are zero. 

We consider the spectrum of 
$\hat h(a,b,\epsilon)$ with $a\ne 0,\epsilon\ne 0 $ fixed and as a function of $b$.  Because they 
are distinct, 
energy levels do not cross in the intervals $\frac{\pi k}{N}   <b<\frac{\pi (k+1)}{N}$  for all integers $k$.  
Note that $b$ a multiple of $\pi/N$  are precisely the $b$ values where  $\frac{\partial \lambda_j(b) }{\partial b} =0$ as shown in appendix \ref{ap:derivs}. 
We suspect that close approaches (avoided crossings) between 
pairs of eigenvalues can only occur where $b$ is an multiple of $\pi/N$. 
This is likely related to the symmetries described by the commutators in equations \ref{eqn:Pz2k} 
 and \ref{eqn:Pz1k} 
which hold for $b$ equal to even and odd multiples of $\pi/N$. 
Numerically we find that eigenvalues are distinct for $b$ a multiple of $\pi/N$ and $\epsilon \ne 0$ 
except in the special case of $N$ a multiple of 4 and for two zero eigenvalues. However, a 
proof that the eigenvalues are distinct for $b$ multiple of $\pi/N$ (and excluding the special case) 
has eluded us.  It would be nice to show that avoided crossings (closest approaches between eigenvalues) 
only occur for $b$ equal to a multiple of $\pi/N$. 

\subsection{Other potentials} \label{ap:other}

Many of the relations given in appendices \ref{ap:bshift} -- \ref{ap:degen}  are not sensitive to the form of the potential term in the operator $\hat h$ of equation \ref{eqn:h0abe}.   
We summarize the relations that hold for potential functions that
differ from $\cos \hat \phi$. 

Consider a Hermitian operator in the form 
\begin{align}
\hat k(a,b,\epsilon) = a \cos (\hat p - b) + \epsilon V(\hat Z, \hat Z^\dagger) \label{eqn:khat}
\end{align}
where $V()$ is a polynomial in $\hat Z, \hat Z^\dagger$.  If $V()$ 
is a polynomial of the Hermitian operators $\cos \hat \phi = \frac{1}{2}(\hat Z + \hat Z^\dagger)$ 
and $\sin \hat \phi = \frac{1}{2i} ( \hat Z - \hat Z^\dagger)$, then it  would be periodic in $\hat \phi$ 
and Hermitian. 

Theorems \ref{th:A1},  \ref{th:h0r} and \ref{th:distinct} hold for 
a Hermitian operator $\hat k$ in the form of equation \ref{eqn:khat}. 

If the potential $V(\hat Z, \hat Z^\dagger)$ polynomial commutes with the parity operator $\hat P$ 
(equation \ref{eqn:parity}), then 
Theorems \ref{th:bmb}, \ref{th:hhalf},  \ref{th:sym1}, \ref{th:sym2}
 and corollaries \ref{th:A6}, and \ref{th:PZ_com} hold for a Hermitian 
 operator in the form of $\hat k$ of equation \ref{eqn:khat}. 

\subsection{Determinants} \label{ap:det}

We calculate 
the determinant $\det \hat h(1,0,\epsilon)$  (with operator defined in equation \ref{eqn:h0abe}).

The operator 
 $2 \hat h(1,0,0)$ in the conventional basis is a matrix that has a zero diagonal
and has 1s on the two off diagonals.   Equation \ref{eqn:detm} gives for the determinant 
\begin{align}
\det  [2 \hat h(1,0,0) ]= \tr \left[  \begin{pmatrix} 0 & -1 \\ 1 & 0 \end{pmatrix}^N \right]
+ 2 (-1)^{N+1}. \label{eqn:2hsimple}
\end{align}
The matrix \begin{align}
\hat A = \begin{pmatrix} 0 & -1 \\ 1 & 0 \end{pmatrix}  \label{eqn:hatA} 
\end{align}
obeys $\hat A^2 = - \hat I$ and is proportional to the Pauli Y matrix.  
If $N$ is even then $\tr (\hat A^N) = 2(-1)^{N/2}$.   If $N$ is a multiple of 4 then $\tr (\hat A^N) = 2$.
If $N$ is even but not a multiple of 4 then $\tr (\hat A^N) = -2$. 
If $N$ is odd then $\tr (\hat A^N) = 0$ because $\tr \hat A=0$. 
Putting these together gives 
\begin{align}
\det  [ 2 \hat h(1,0,0) ] = \begin{cases} 
0  & \text{ if } N  \text{ mod } 4 = 0  \\
2 & \text{ if } N  \text{ is odd} \\
-4 & \text{ if } N \text{ mod } 4 = 2  \\
\end{cases}.
\end{align}
We remove the factor of 2 from inside the determinant to find 
\begin{align}
\det [ \hat h(1,0,0)]  = 
\begin{cases} 
 0  & \text{ if } N  \text{ mod } 4 = 0  \\
 2^{1-N} & \text{ if } N  \text{ is odd} \\
 -2^{2-N} & \text{ if } N \text{ mod } 4 = 2  \\  
\end{cases}. \label{eqn:deth000}
\end{align}
Since the determinant of $\hat h(1,0,0)$ is equivalent to the product of its eigenvalues,
 the product of cosines it is handy to also write 
\begin{align}
\det [ \hat h(1,0,0)]   & = \!
\prod_{j=0}^{N-1}\!\! \cos \left(\frac{2 \pi j}{N} \right) 
  =\! \frac{1}{2^N}\! \prod_{j=0}^{N-1} \!\! \left( \omega^j + \omega^{-j} \right) 
 \label{eqn:prod_cosines}
\end{align}
where $\omega = e^{\frac{2 \pi i}{N}}$.   

We now compute the determinant for the Harper operator with $\epsilon \ne 0$. 
Again using equation \ref{eqn:detm} to compute the determinant  (and in the conventional basis) 
\begin{align}
\det \left[  2\hat h(1,0,\epsilon) \right] =  & 
\tr \Bigg[ 
\begin{pmatrix}2 \epsilon \cos \frac{2 \pi (N-1)}{N} & -1 \\ 1 & 0 
 \end{pmatrix} \times \nonumber \\
 & \ \ \begin{pmatrix}2 \epsilon \cos \frac{2 \pi (N-2)}{N} & -1 \\ 1 & 0 
 \end{pmatrix} \times \ldots \nonumber \\
 & \ \ \begin{pmatrix}2 \epsilon \cos \frac{2 \pi}{N} & -1 \\ 1 & 0 
 \end{pmatrix} \times \nonumber \\
 &\ \  \begin{pmatrix}2 \epsilon  & -1 \\ 1 & 0 
 \end{pmatrix} \Bigg] \nonumber \\
&+ (-1)^{N-1} 2 . \label{eqn:mess}
\end{align}
We define some 2$\times$2 matrices that lie within the trace term 
in equation \ref{eqn:mess} ;
\begin{align}
\det \left[ 2\hat h(1,0,\epsilon) \right] & =   
         \tr  \Bigg[ \prod_{\substack{\text{ordered} \\j = N-1 \\ \text{to } 0}} \! \! \hat S_j \Bigg] + (-1)^{N-1} 2 \nonumber \\
\hat S_j & = 
     \begin{pmatrix} 2\epsilon \cos \frac{2 \pi j}{N} & -1 \\ 1 & 0  
         \end{pmatrix}. \label{eqn:S_j}
\end{align}
\begin{align}
\frac{d \hat S_j}{d\epsilon}  & = 
     \begin{pmatrix} 2 \cos \frac{2 \pi j}{N} & 0 \\ 0 & 0 
     \end{pmatrix}  \equiv \hat B_j \label{eqn:hatBj} \\
 \frac{d^2\hat  S_j}{d\epsilon^2}  & =  0 \end{align} 
 \begin{align}
 \hat S_j \Bigg|_{\epsilon=0} & =  \begin{pmatrix} 0 & -1 \\ 1 & 0 
     \end{pmatrix} \equiv \hat A.  \label{eqn:hatA2}
 \end{align}
 
The determinant in equation \ref{eqn:mess} is a polynomial in $\epsilon$, 
 \begin{align}
 \det  [2\hat h(1,0,\epsilon)] = \sum_{j=0}^N g_j \epsilon^j.  \label{eqn:gg}
 \end{align}
 We aim to compute the coefficients $g_j$.  
 
 In the Fourier basis using 
 equation \ref{eqn:detm} the determinant  
\begin{align}
\det  \left[ 2  \hat h(1,0,\epsilon)\right] =  & 
\tr \Bigg[ 
\begin{pmatrix} 2   \cos \frac{2 \pi (N-1)}{N} & - \epsilon^2 \\ 1 & 0 
 \end{pmatrix} \times \nonumber \\
 & \ \ \begin{pmatrix}2  \cos \frac{2 \pi (N-2)}{N} & -\epsilon^2\\ 1 & 0
 \end{pmatrix} \times \ldots \nonumber \\
 & \ \ \begin{pmatrix}2  \cos \frac{2 \pi}{N} & -\epsilon^2\\ 1 & 0 
 \end{pmatrix} \times \nonumber \\
 &\ \  \begin{pmatrix}2  & -\epsilon^2 \\ 1 & 0
 \end{pmatrix} \Bigg] \nonumber \\
&+ (-1)^{N-1} 2  \epsilon^N. \label{eqn:mess2}
\end{align}
This shows that the coefficients $g_j$  of the polynomial in $\epsilon$  (for $j<N$)
 are only nonzero if $j$ is even as 
there are only terms proportional to $\epsilon^2$ within the trace operator in equation \ref{eqn:mess2}.  
 Hence 
 \begin{align}
 \text{for } j \text{ odd and } j < N, \ \   g_j = 0. 
 \end{align}
 The coefficient 
 \begin{align}
 g_0 = \det [2\hat h(1,0,0)]\end{align}
  is equal to the expression given 
 in equation \ref{eqn:deth000} times a power of 2.  
 
 In the Fourier basis using 
 equation \ref{eqn:detm} the determinant  
\begin{align}
\det  \left[ 2 \epsilon^{-1} \hat h(1,0,\epsilon)\right] =  & 
\tr \Bigg[ 
\begin{pmatrix} 2  \epsilon^{-1} \cos \frac{2 \pi (N-1)}{N} & - 1 \\ 1 & 0 
 \end{pmatrix} \times \nonumber \\
 & \ \ \begin{pmatrix}2 \epsilon^{-1} \cos \frac{2 \pi (N-2)}{N} & -1 \\ 1 & 0 
 \end{pmatrix} \times \ldots \nonumber \\
 & \ \ \begin{pmatrix}2\epsilon^{-1}  \cos \frac{2 \pi}{N} & -1 \\ 1 & 0 
 \end{pmatrix} \times \nonumber \\
 &\ \  \begin{pmatrix}2 \epsilon^{-1}  & -1 \\ 1 & 0 
 \end{pmatrix} \Bigg] \nonumber \\
&+ (-1)^{N-1} 2 . \label{eqn:mess3}
\end{align}
 In the large $\epsilon$ limit equation \ref{eqn:mess3} resembles \ref{eqn:2hsimple} for  
 the operator with $\epsilon=0$ 
 in the conventional basis giving 
 \begin{align} g_N = g_0.\end{align} 
 
Remaining to compute are coefficients $g_j$ for $j$ even and $0 < j < N$. 
By applying derivatives with respect to $\epsilon$ to equation \ref{eqn:mess},
 the coefficients of the polynomial 
in equation \ref{eqn:gg} depend upon 
an ordered product of $2 \times 2$ matrices 
 \begin{align}
 g_j & = \frac{d^j}{d \epsilon^j} 
  \tr  \Bigg[ \prod_{\substack{\text{ordered} \\j = N-1 \\ \text{to } 0}} \hat S_j \Bigg] \Bigg|_{\epsilon=0}.
  \label{eqn:g_j}
 \end{align}
 Since $\frac{d \hat S_j}{d \epsilon^2}  = 0$ when taking the derivatives, the result 
 is a sum of terms that contain a product of matrices $\hat B_j $ defined in equation 
 \ref{eqn:hatBj} and $\hat A$ (defined in equations \ref{eqn:hatA} or \ref{eqn:hatA2}).
 
 The $g_2$ coefficient contains 
 a sum of terms.  Each 
 term contains two matrices in
 the form of the matrix $\hat B_j$ defined in equation \ref{eqn:hatBj} 
 interspersed with $N-2$ matrices that are equal to $\hat A$.  
 It is convenient to compute  for $ j, j' \in \{ 0, 1, \ldots , N-1 \}$ and  $k\in \mathbb{Z}$.
 \begin{align}
 \hat A^{2k}& = (-1)^k \hat I \nonumber \\
 \hat A^{2k-1} &= (-1)^k \hat A\nonumber  \\
 \hat B_j \hat A \hat B_{j'} &= 0 \nonumber  \\
 \hat B_j \hat A^{2k+1} \hat B_{j'} & = 0 \nonumber\\
 \hat B_j \hat A^{2k}  \hat B_{j'}  &= (-1)^k \hat B_j \hat B_{j'}   .  \label{eqn:ABs}
 \end{align}
The trace of the product of a matrix with a single coefficient on the diagonal and $\hat A$ vanishes; 
 \begin{align}
 \tr \left[ \begin{pmatrix}1 & 0 \\ 0 & 0 \end{pmatrix} \hat A \right] = 0. 
 \end{align}
 This and equations \ref{eqn:ABs} imply that 
 if there is a odd number of 
 $\hat A$ matrices inside the trace term in equation \ref{eqn:g_j} then the resulting trace vanishes.   
 Hence if $N$ is odd then $g_2 = 0$.  
 The $g_k$ coefficient contains $k$  matrices 
 in the form of $\hat B_j$ and $N-k$ matrices that are equal to $\hat A$. 
 As $k$ must be even, this implies that if $N$ is odd, there are an odd number of $\hat A$
 matrices, again giving a trace of zero. Hence 
 \begin{align}
 g_k = 0 \text{ for } N \text{ odd, } 0<k<N. 
 \end{align}

 If $N$ is even then we find that 
 \begin{align}
 g_2 & = \sum_{\substack {j>j' \\ j'- j \text{ odd} }}
 \tr \left[ \hat B_j \hat B_{j'} (-1)^{ N/2 -1 } \right] . 
 \label{eqn:g_2}
 \end{align} 
 The condition that $j, j'$ differ by an odd number arises so that there are an even number 
 of $\hat A$ matrices between $\hat B_j, \hat B_{j'}$ within the product inside the trace in equation \ref{eqn:mess}.  The sum is over all possible combinations that satisfy this condition 
 with $j, j' \in \{0, 1, \ldots, N-1\}$.  Using equation \ref{eqn:hatBj} for the $\hat B$ matrices, 
 equation \ref{eqn:g_2} becomes 
 \begin{align}
 g_2 & =  \sum_{\substack {j>j' \\ j'- j \text{ odd} }} 4 \cos \left( \frac{2 \pi j}{N}\right) \cos \left( \frac{2 \pi j'}{N}\right) .
  \end{align} 
 Suppose we choose $j,  j'$. 
Because $j' - j$ is odd and $N$ is even,  we can always find another pair of indices, such as  
 $N/2+ j, j'$ or $j, N/2+j'$ that has the opposite sign in their product of cosines. 
Hence the sum over all possible $j, j'$ indices
vanishes.   An extension of this argument holds for $g_4$ and the other coefficients in the case of $N$ even. 
  Hence 
 \begin{align}
 g_k = 0 \text{ for } N \text{ even, } 0<k<N. 
 \end{align}
 
The arguments presented above imply that  the polynomial of equation \ref{eqn:gg} 
only contains two terms and 
\begin{align}
\det  \hat h(1,0,\epsilon) = \begin{cases} 
0  & \text{ if } N  \text{ mod } 4 = 0  \\
2^{1-N} (1 + \epsilon^N) & \text{ if } N  \text{ is odd} \\
-2^{2-N} (1 + \epsilon^N) & \text{ if } N \text{ mod } 4 = 2  \\
\end{cases}.
\label{eqn:deth_gen}
\end{align}

We have checked this relation numerically and it is probably exact.  

\begin{lemma} \label{th:zeros}
For $N$ a multiple of 4,  $a\ne 0$,  the operator $\hat h(a,0,\epsilon)$ has a zero eigenvalue of multiplicity 
 2. 
\end{lemma}
\begin{proof}
For $N$ a multiple of 4, via equation \ref{eqn:deth_gen} calculated above, 
the determinant $\det h(a,0,\epsilon) = 0$.  That implies that there is at least one
eigenvalue that is zero. 
For $N$ even, eigenvalues must come in pairs $\lambda, -\lambda$ (theorem \ref{th:even}).
Consequently  if there is a zero eigenvalue, it must have multiplicity of at least 2. 
As long as $a \epsilon \ne 0$, eigenvalues cannot have multiplicity greater 2, as discussed in appendix \ref{ap:degen} where we applied the Cauchy interlacing theorem.
This implies that if $N$ is a multiple of 4, there must be two zero eigenvalues. 
\end{proof}

\subsection{Estimate of the minimum eigenvalue spacing}
\label{ap:min_spacing}

We estimate the minimum distance between any the two eigenvalues of $\hat h(1,0,\epsilon)$    (equation \ref{eqn:h0abe}) 
(with $a=1, b=0$) that are closest to zero,  
and with dimension $N$ not equal to a multiple of 4.   We ignore the case $N$ a multiple of 4 because
in this case there are two zero eigenvalues (see lemma \ref{th:zeros}). 


We use $p_{\hat h}$ to represent the characteristic polynomial 
\begin{align}
p_{\hat h}(x) \equiv \det[ x \hat I -  \hat h(a,0,\epsilon)] . 
\end{align}
Using equation \ref{eqn:detm} and the Fourier basis  (equation \ref{eqn:cofs_k})
the characteristic polynomial of $\hat h(1,0,\epsilon)$   can be written in terms of
the trace of a product of 2$\times$2 matrices
\begin{align}
 p_{\hat h}(x) = \tr \Bigg[&
 \begin{pmatrix} x- \cos \frac{2 \pi (N\!-\!1)}{N}  & - \frac{\epsilon^2}{4} \\ 1 & 0   \end{pmatrix}\times\nonumber \\
 & \begin{pmatrix} x- \cos \frac{2 \pi (N\!-2\!)}{N}  & - \frac{\epsilon^2}{4} \\ 1 & 0   \end{pmatrix} \times \ldots \nonumber \\
&  \begin{pmatrix} x- \cos \frac{2 \pi }{N}  & - \frac{\epsilon^2}{4} \\ 1 & 0   \end{pmatrix} \times\nonumber \\
& \begin{pmatrix} x- 1 & - \frac{ \epsilon^2}{4} \\ 1 & 0   \end{pmatrix} \Bigg] \nonumber \\
& - 2 \left(\frac{\epsilon}{2} \right)^N
. \label{eqn:detM1}
\end{align}

We insert $x = \lambda_* + \delta$ into the characteristic polynomial of equation \ref{eqn:detM1} 
where $\lambda_*$ is the smallest near 0 degenerate eigenvalue of $\hat h(1,0,0)$.  
Here $\delta$, assumed small,
 will give us an estimate for the spacing between nearly degenerate eigenvalues as a function of $\epsilon$. 
The eigenvalue 
 \begin{align} 
\lambda_* = \cos \frac{2 \pi j_*}{N} = \cos \frac{2 \pi j_*'}{N}
\end{align}
 for $j_*$ near $ N/2$  and $j_*'$ near $-N/2$. 
 We assume that $\lambda_*$ is small (near 0) and that $|\delta| < |\lambda_*|$ is even smaller. 
 To second order in $\delta$ 
Terms containing factors of $\lambda_* - \cos \frac{2 \pi j_*}{N}$ and $\lambda_* - \cos \frac{2 \pi j_*'}{N}$  are zero,  so assuming that $\delta$ is small, the characteristic polynomial (equation \ref{eqn:detM1}) is approximately  
\begin{align}
p_{\hat h}(\lambda_* \!+ \!\delta) \sim  \tr & \Bigg[  
     \begin{pmatrix} \lambda_* - \cos \frac{2 \pi (N-1)}{N} & -\frac{\epsilon^2}{4} \\ 1 & 0 \end{pmatrix} \ldots  \nonumber \\
& \times \begin{pmatrix} \lambda_* - \cos \frac{2 \pi (j_*'+1)}{N} & -\frac{\epsilon^2}{4} \\ 1 & 0 \end{pmatrix} \nonumber \\
&  \times \begin{pmatrix}  \delta  & -\frac{\epsilon^2}{4} \\ 1 & 0 \end{pmatrix}   \nonumber  \\
& \times \begin{pmatrix}  \lambda_*- \cos \frac{2 \pi (j_*'-1)}{N} & -\frac{\epsilon^2}{4} \\ 1 & 0 \end{pmatrix} 
\ldots \nonumber \\
& \times \begin{pmatrix}  \lambda_*- \cos \frac{2 \pi (j_*+ 1)}{N} & -\frac{\epsilon^2}{4} \\ 1 & 0 \end{pmatrix}  \nonumber \\
& \times  \begin{pmatrix}  \delta  & -\frac{\epsilon^2}{4} \\ 1 & 0 \end{pmatrix} \nonumber \\
& \times \begin{pmatrix}  \lambda_*- \cos \frac{2 \pi (j_* - 1)}{N} & -\frac{\epsilon^2}{4} \\ 1 & 0 \end{pmatrix}  \ldots \nonumber \\
& \times  \begin{pmatrix}  \lambda_* - 1  & -\frac{\epsilon^2}{4} \\ 1 & 0 \end{pmatrix} \Bigg]
- 2^{1-N} \epsilon^N \label{eqn:plambda}.
\end{align}
Each 2$\times$2 matrix inside the trace of equation \ref{eqn:plambda} is associated with an integer index 
corresponding to a $N$-th complex
root of unity, as shown in Figure \ref{fig:circ}.   The 2$\times$2 matrices inside the trace can be cyclicly permuted so  we can think of each one as being located on the unit circle. 
If $\lambda_*$ is small compared to other roots of $\hat h(1,0,0)$ then in many cases we can approximate 
\begin{align}
\lambda_* - \cos \frac{2 \pi k}{N} \sim   - \cos \frac{2 \pi k}{N}  \text{ for } k \ne j_*,  j_*'. \label{eqn:approx}
\end{align} 

\begin{figure}[htbp]\centering
\includegraphics[width=3truein]{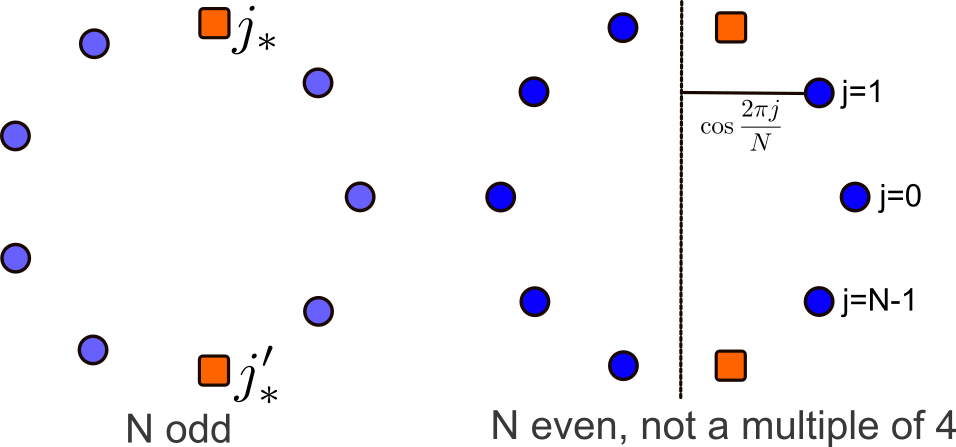}
\caption{Complex roots of unity displayed on the unit circle. On the left we show dimension $N$ odd (in this case $N=9$) 
and on the right $N$ even but not a multiple of 4 (in this case $N=10$).  
Two indices (denoted $j_*$ and $j_*'$ that have the same 
and lowest value of $\cos \frac{2\pi j_*}{N} = \cos \frac{2\pi j_*'}{N}$ are shown with orange squares. 
For $N$ odd, remaining indices form two disjoint sets,  one with an even number of the elements and the other with an odd number of elements.   For $N$ even (but not a multiple of 4) the remaining indices form two disjoint sets, both containing odd numbers of elements.   The 2$\times$2 matrices within the trace operator in equation \ref{eqn:plambda2} can be cyclicly permuted so 
we can think of each 2$\times$2 matrix as located on one of the nodes on the circle.   \label{fig:circ}}
\end{figure}

We define operators
\begin{align}
\hat C &=  \begin{pmatrix}  0  & -\frac{\epsilon^2}{4} \\ 1 & 0 \end{pmatrix}  \label{eqn:C_def} \\
\hat D_j &= \left( \lambda_*  -\cos \frac{2 \pi j}{N} \right) \ket{0}\bra{0} .  \label{eqn:D_j_def}
\end{align}
These operators obey relations similar to those in equation \ref{eqn:ABs}.  
For $ j, j' \in \{ 0, 1, \ldots , N-1 \}$ and  $k\in \mathbb{Z}$.
\begin{align}
 \hat C^{2k}& = (-1)^k \left(\frac{\epsilon}{2}\right)^{2k} \hat I \nonumber \\
 \hat C^{2k+1} &= (-1)^k \left(\frac{\epsilon}{2}\right)^{2k}  \hat D \nonumber  \\
 \hat D_j \hat C \hat D_{j'} &=  \ket{0}\bra{0} \hat C \ket{0}\bra{0}  = \hat D_j \hat C^{2k+1} \hat D_{j'} = 0 \nonumber  \\
 \hat D_j \hat C^{2k}  \hat D_{j'}  &= (-1)^k \left(\frac{\epsilon}{2}\right)^{2k}  \hat D_j \hat D_{j'}   .  \label{eqn:CDs}
 \end{align}
Using the $\hat C, \hat D_j$ operators equation \ref{eqn:plambda} becomes 
\begin{align}
p_{\hat h}(\lambda_* \!+\! \delta) \sim  \tr &\Bigg[ ( \hat D_{N-1} + \hat C)  ( \hat D_{N-2} + \hat C) \ldots  
 \nonumber \\
&\times ( \hat D_{j_*' + 1} + \hat C) ( \delta \ket{0}\bra{0} + \hat C)( \hat D_{j_*' -1} + \hat C) \ldots  \nonumber \\
& \times ( \hat D_{j_* + 1} + \hat C) ( \delta \ket{0}\bra{0} + \hat C)( \hat D_{j_* -1} + \hat C) \ldots  \nonumber \\
&\times  (\hat D_1 + \hat C)  (\hat D_0 + \hat C) \Bigg]
- 2^{1-N} \epsilon^N \label{eqn:plambda2}.
\end{align}

We consider two terms, one proportional to $\delta^2$ and the other independent of $\delta$, 
and write equation \ref{eqn:plambda2} as a polynomial in $\epsilon$ 
\begin{align}
p_{\hat h}(\lambda_* + \delta) \sim \sum_{k=0}^{N} ( c_k \epsilon^k + \delta^2 d_k \epsilon^k ) . 
\end{align}
We ignore a possible term proportional to $\delta$ as this would more likely to 
cause an $\epsilon$ dependent shift in an eigenvalue 
rather than split a degeneracy.  
For $k < N$, and $k$ odd $c_k = d_k = 0$.  This follows as the trace term in equation \ref{eqn:plambda2} 
only contains factors of $\epsilon^2$.  

We first consider the $d_k$ coefficients.  For these the product inside the trace 
must consist of a product of $N$ 2$\times$2 matrices which includes two operators 
$\delta \ket{0}\bra{0}$ at the location of indices $j_*$ and $j_*'$.  The remaining 
matrices are either $\hat D_j$ or $\hat C$ operators. 
 For each location on the unit circle that is not associated with $j_*$ or $j_*'$ 
(shown in Figure \ref{fig:circ} with blue dots) either $\hat D_j$ 
or $\hat C$ is chosen.       At the location of $j_*$ and $j_*'$, shown by
the orange scales in Figure \ref{fig:circ} the operator is $\delta \ket{0}\bra{0}$.  
To be non-zero, 
equation \ref{eqn:CDs} implies that there must be consecutive pairs of $\hat C$ operators. 
Each $\hat D_j$ operator gives a cosine whereas each pair of $\hat C$ operators gives
a factor of $-\epsilon^2/4$. 
For $k=0$ all indices that are not $j_*, j_*'$ give cosine factors; 
\begin{align}
d_0 &  \sim \prod_{k \ne j_*, j_*} \left( - \cos \frac{2 \pi k}{N} \right) \nonumber \\
& \sim (-1)^{N-2} \prod_{k=0}^{N-1} \frac{ \cos \frac{2 \pi k}{N}}{ \lambda_*^2}  \nonumber \\
& \sim \begin{cases} - 2^{1-N} \lambda_*^{-2} & N \text{ odd} \\
                              - 2^{2-N} \lambda_*^{-2} & N \text{ even, not a multiple of 4} 
                              \end{cases} \label{eqn:d0}
\end{align}
where we have used equation \ref{eqn:prod_cosines} for the product of cosines. 

Because there are two factors of $\delta$, and $k$ must be even, 
the maximum $k$ with $d_k$ non-zero
must be $d_{N-3}$ if $N$ is odd and $d_{N-2}$ if $N$ is even. 
In the case that $N$ is odd  all indices must give $\hat C$ operators except 1 
so that there are a even number. The lone $\hat D_j$ operator must be on the side of the 
circle containing an odd number of indices, otherwise  $\hat C$ operators
would not be in consecutive pairs;  
\begin{align}
d_{N-3} & \sim (-1)^{(N-3)/2} 2^{3-N} 
\sum_{\substack{k \ne j_*, j_*' \\ \text{one side} \\ k-j_* \text{odd} }} \left( - \cos \frac{2 \pi j}{N} \right) 
 \nonumber \\
&  \sim  (-1)^{(N-3)/2}2^{3-N}  \frac{N}{2\pi} 
\end{align}
where we have approximated the sum with an integral.  
The magnitude of the coefficient $|d_{N-3} |\epsilon^{N-3}$ is likely to be smaller in magnitude than $|d_0|$.   

In the case that $N$ is even (all indices must have $\hat C$ operators) 
\begin{align}
d_{N-2} & \sim 2^{(N-2)}(-1). 
\end{align}
Again we expect $|d_{N-2}|\epsilon^{N-2} < d_0$ 
for $\epsilon <1$.  
We guess that intermediate terms $|d_k| \epsilon^k$ 
with $0<k<N-1$ would also 
be smaller than $|d_0|$ if $\epsilon <1$.   Hence we assume that we can neglect coefficients $d_k$ for $k>0$. 

We now consider $c_k$ coefficients. 
The coefficient $c_0=0$ 
as $p_{\hat h}(\lambda_*)= 0$ if $\epsilon = 0$ (in this case $\lambda_*$ must be a root). 

For the $c_k$ coefficients (no factors of $\delta$), the trace in equation 
\ref{eqn:plambda2} contains a product of 2$\times $2 matrices that includes two separated $\hat C$ 
matrices at the location of $j_*$ and $j_*'$.   
Equation \ref{eqn:CDs} implies that there must be consecutive pairs of separated $\hat C$ operators, 
hence there must be at least 4 of them inside the trace for a non-zero result. Hence 
\begin{align} 
c_2 = 0. \end{align}  

\begin{figure}[htbp]\centering
\includegraphics[width=3truein]{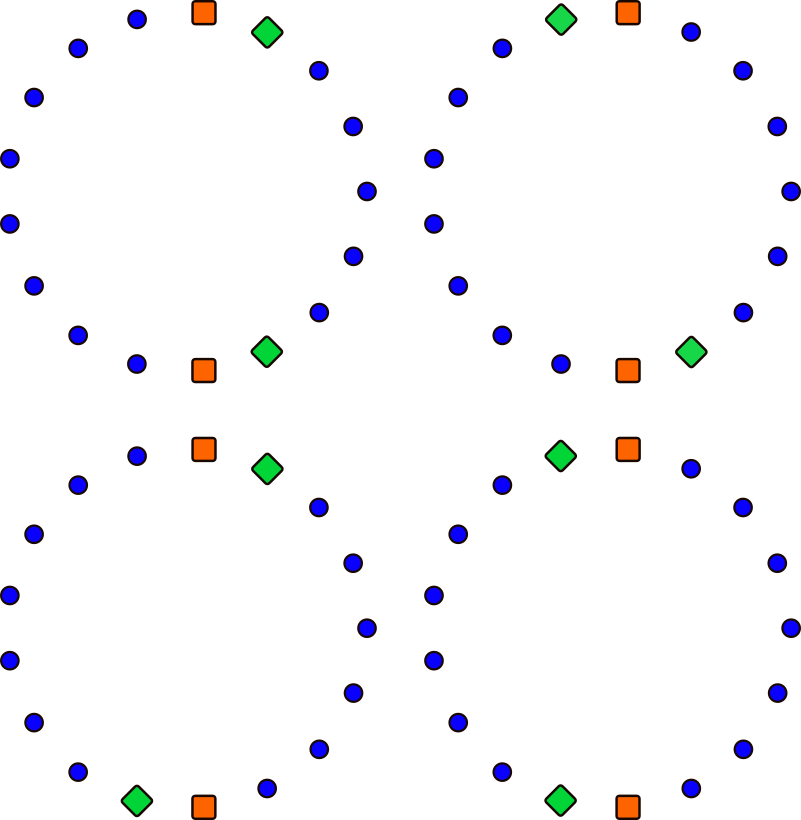}
\caption{The coefficient $c_4$ is a sum of 4 terms.  The terms 
inside the trace operator for each term have $\hat C$ operators at the 
locations of the orange squares and green diamonds and $\hat D_j$ operators at the locations 
of the blue circles.   The 4 terms approximately cancel, giving $c_4 \sim 0$. 
\label{fig:circ_c4}}
\end{figure}

We estimate the coefficient $c_4$ which requires two consecutive pairs of $\hat C$ operators, with location 
as shown in Figure \ref{fig:circ_c4}; 
\begin{align}
c_4 = & \tr \Bigg[ \hat D_{N-1} \hat D_{N-2} \ldots \hat D_{j_*'-2} [ D_{j_*' - 1} \hat C^2 + \hat C^2 D_{j_*'+1}] \times \nonumber \\
&  \ \ \ \ \hat D_{j_*' +2} \ldots \hat D_{j_* - 2} [ D_{j_* - 1} \hat C^2 + \hat C^2 D_{j_*+1} ] \times \nonumber \\
& \ \ \ \ \ \hat  D_{j* + 2}\ldots \hat  D_1 \hat  D_0 \Bigg]. \label{eqn:c4_11}
\end{align}
The remaining operators inside the trace are in the form of $\hat D_j$. 
Because $\hat C^2$ and $\hat D_j$ are diagonal, 
all the 2$\times$2 matrices in the product inside the trace in equation \ref{eqn:c4_11} are diagonal and commute, giving 
\begin{align}
c_4 = &  - \frac{\epsilon^2}{4}  \prod_{ 
\substack{ j \ne j_*', j_*'  \pm 1 \\  j\ne j_*, j_* \pm 1} } \left( \lambda_* - \cos \frac{2 \pi j}{N} \right) S
\end{align} 
with 
\begin{align}
S  =& \Bigg[ \left( \lambda_* - \cos \frac{2 \pi (j_*' + 1) }{N} \right)  + 
\left( \lambda_* - \cos \frac{2 \pi (j_*' - 1) }{N} \right)  \Bigg]  \times \nonumber \\ & 
\Bigg[  \left( \lambda_* - \cos \frac{2 \pi (j_* + 1) }{N} \right) +
\left( \lambda_* - \cos \frac{2 \pi (j_* - 1) }{N} \right)  \Bigg] . 
\end{align}
Each factor contains 
 terms that are approximately equidistant (on the $x$ axis) 
about $\lambda_* $ so they approximately cancel.  
As $ \lambda_* = \cos \frac{2 \pi j_*}{N} = \lambda_* = \cos \frac{2 \pi j_*'}{N}$,  
\begin{align}
S & \sim \left[ \sin \left( \frac{2 \pi j_*'}{N} \right) \left( \frac{2 \pi }{N} - \frac{2 \pi}{N}  \right) \right]^2
 \nonumber \\ & \sim 0 + {\cal O} (N^{-4}).
\end{align}
Consequently we can neglect the coefficient $c_4$.  
The coefficient $c_6$ can be written in terms of two set sets, one that is proportional to $c_4$
and so cancel, and another set.  The additional set can be written in terms of pairs 
that also cancel, so $c_6 \sim 0$. 

For $k=N$ we can use the determinant (equation \ref{eqn:deth_gen}) to find that 
\begin{align}
c_N  & =   \begin{cases}  -2^{2-N} &  N \text{ even, not a multiple of 4} \\
-2^{1-N} & N \text{ odd} \end{cases} .  \label{eqn:cN}
\end{align}
   
If $N$ is odd there cannot be $N$ (odd) factors of $\hat C$ inside the trace operator of equation \ref{eqn:plambda2} (otherwise the result is zero because of equation \ref{eqn:CDs}).  
We consider the coefficient $c_{N-1}$ which requires a single $\hat D_j$ operator within 
the trace of equation \ref{eqn:plambda2}.  Equation  \ref{eqn:plambda2}
has operators $\hat C$ at index $j_*$ and $j_*'$.    The trace can be cyclicly permuted so 
we can think of the indices as lying on a circle as shown in Figure \ref{fig:circ}. 
There is only a single $\hat D_j$ operator that can be located at any location other 
than $j_*$ or $j_*'$.  The remaining $N-1$ operators give a product 
of an even number of $ \hat C$ operators. 
This gives 
\begin{align}
c_{N-1} & \sim  (-1)^\frac{(N-1)}{2} 2^{1-N}\sum_{\substack{ k \ne j_*, j_*' }} 
 \left(\lambda_* - \cos \frac{2 \pi k}{N}  \right)  \nonumber \\
& \sim (-1)^\frac{(N-1)}{2} 2^{1-N}  (N-2) \lambda_*   \nonumber \\
&  \sim   (-1)^\frac{(N-1)}{2} 2^{1-N}  \pi, \label{eqn:cN1}
\end{align}
where we have used the fact that the sum of cosines is approximately zero. 
The factor of $\pi$ causes $|c_{N-1}|\epsilon^{N-1} > |c_N| \epsilon^N$ (as long as $\pi>\epsilon$) so 
the coefficient $c_{N-1}$ can be neglected. 

For $N$ even (not a multiple of 4), we compute $c_{N-2}$ giving 2 operators $\hat D_k, \hat D_{k'}$ operators
and the rest of the indices corresponding to $\hat C$ operators inside the trace in equation \ref{eqn:plambda2}.   The indices $j_*, j_*'$ must have $\hat C $ operators. 
 Due to equation \ref{eqn:CDs}, $\hat C$ operators must 
 be grouped in consecutive pairs.   That means that  $k-k'$ must be an 
 odd number so that even numbers of $\hat C$ can fit on both arcs of the unit circle; 
 \begin{align}
c_{N-2}  &= (-1)^{\frac{(N-2)}{2}} 2^{2-N}  \times \nonumber \\
 & \ \ \  \sum_{\substack{ k, k' \ne j_*, j_*'\\   k - k' \text{odd}}} 
\left(\lambda_* - \cos \frac{2 \pi k}{N}\right) \left(\lambda_* - \cos \frac{2 \pi k'}{N}\right).  \label{eqn:c2sum}
\end{align}
 Even though there are of order $N^2$ possible locations for $k,k'$ many of them can be 
 grouped in pairs 
that approximately cancel.  Given $k$, and $N$ even but not a multiple of 4, $N/2$ is odd and 
 the index $N/2 -k$ is separated from $k$ by an odd 
 number.   For  $k \ne N/2 -k' $ there is a term with indices $k, N/2 -k'$ that approximately 
has the opposite sign as the term with $k, k'$ (where we have used the approximation in equation \ref{eqn:approx}).
The terms that would not cancel would have 
 $k, k' = N/2-k$ on opposite sides of the unit circle. 
The largest terms that could be present could have $k, k'$ near $0, N/2$ 
which would give a factor of order 1 (as the product of the cosines is about $-1$).  
The coefficient should be somewhat  (a factor of order unity) larger than 
\begin{align}
c_{N-2} \sim2^{2-N}  \label{eqn:cN2}
\end{align}
as additional and smaller pairs of cosines would contribute to the sum. 
For $\epsilon <1$,  the coefficient $|c_{N-2}| \epsilon^{N-2}> |c_N|\epsilon^N$ 
 so the $c_N$ term can be neglected for $N$ even and not a multiple of 4. 

For $N$ even, not a multiple of 4, 
using equation \ref{eqn:cN2} 
for $c_{N-2}$ and equation \ref{eqn:d0} for $d_0$ and setting $p_{\hat h}(\lambda_* + \delta) = 0$ 
so that we find an eigenvalue, (and neglecting $c_k$ for $k<N-2$), we find a splitting 
of the eigenvalues near $\lambda_*$ 
\begin{align}
\delta & \sim \epsilon^{\frac{N-2}{2}} \lambda_*  \sim  \epsilon^{\frac{N-2}{2}} \frac{\pi}{N}  \nonumber \\
& \text{ for } N \text{ even, not a multiple of 4}.  \label{eqn:matche}
\end{align} 
which is approximately consistent with the minimum spacing we found numerically in section
\ref{sec:spacing} (equation  \ref{eqn:min_dist}). 
For $N$ odd using equation \ref{eqn:cN1} for $c_{N-1}$ and $d_0$ from equation 
\ref{eqn:d0} and setting $p_{\hat h}(\lambda_* + \delta) = 0$ (and neglecting $c_k$ for $k<N-1$), we find that 
\begin{align}
\delta & \sim \epsilon^{\frac{N-1}{2}} \frac{\pi}{N}   & \text{ for } N \text{ odd}  \label{eqn:matcho}
\end{align}
which is approximately consistent with the minimum spacing we found numerically in section
\ref{sec:spacing} (equation \ref{eqn:min_dist}). 

Note that 
we neglected to show  that $c_k$ can be neglected for $k < N-2$ (for $N$ even) or $k < N-1$ (for $N$ odd). 
Spacing constraints between
the cosine factors (due to the relations in equation \ref{eqn:CDs}) reduce the numbers of possible combinatorial factors and the cosines themselves are less than 1,  consequently many of these terms might  
 be relatively small.   Many terms cancel or approximately cancel because the cosines can be 
 both positive and negative and can be grouped in pairs that have similar size but opposite sign. 
 The approximation of equation \ref{eqn:approx} 
 can affect possible cancellations or near cancellations in these terms. 
 For example, when estimating $c_4$ (above), 
the terms nearly cancel each other, for either $N$ even (not a multiple of 4) or  $N$ odd, 
but to see this, we used $\hat D_j = (\lambda_* - \cos \frac{2 \pi j}{N}) \ket{0}\bra{0}$ 
and did not use the approximation of equation \ref{eqn:approx}. 
So far, we have failed to find a simple way (meaning one not involving combinatorics for each case) to show if $c_k$ for $k = 6$ to $N-3$ or $N-4$ (depending 
upon whether $N$ is even or odd) can be neglected, even though 
the correspondence between equations \ref{eqn:matche} and \ref{eqn:matcho} 
and our numerical estimates suggest that it might be possible. 
 A small change of coefficients of a polynomial may induce a dramatic change in its roots
(notoriously Wilkinson's polynomial; \cite{Wilkinson_1963}).  
Thus the rough estimates given here are not rigorous even though
 though they do approximately 
give the correct scaling we found numerically for the minimum distance between eigenvalues 
at the center of the circulating region.  

\section{Application of the Landau-Zener model} \label{ap:LZ}

For a drifting two state system with an avoided crossing in energy levels, the Landau-Zener model gives the probability of a diabatic transition 
\begin{align}
P_D = e^{-2 \pi \Gamma} \label{eqn:P_D}
\end{align}
with dimensionless quantity 
\begin{align}
\Gamma = \frac{A^2}{\hbar \alpha} \label{eqn:Gamma}
\end{align}
\citep{Zener_1932,Wittig_2005}
that depends on two numbers $\alpha$ and $A$. 
During the transition, the 
two eigenvalues of the Hamiltonian have a minimum separation $dE_{min} = 2A$ 
and the rate that the energy levels approach each 
other is $\dot E = \alpha$. 

If $\Gamma \gg 1$, then the drift is adiabatic.  When begun in an initial eigenstate, it remains 
 in it.  If $\Gamma \ll 1$, then the drift is diabatic.    The probability of a transition $P_D $ approaches 1. 
If initially begun in an eigenstate, the system transitions with high probability to the opposite eigenstate. 
For both $\Gamma \gg 1$ and $\Gamma \ll 1$, if the system begins in an eigenstate, 
then it is in an eigenstate after the avoided crossing.  A transition into 
 a superposition state is likely only if $\Gamma \sim 1$. 

If the probability of a diabatic transition is 1/2 then equations \ref{eqn:P_D} and \ref{eqn:Gamma} give 
the minimum value for the distance between the energy eigenvalues; 
\begin{align}
dE_{min,p1/2} &=2 \sqrt{  \frac{ \hbar\alpha \ln 2 }{2 \pi} } . \label{eqn:phalf}
\end{align}

For our drifting Harper model we draw upon Figure \ref{fig:cc3} to estimate the rate  $\alpha$ 
(in equation \ref{eqn:Gamma})
that energy levels approach each other if parameter $b$ (of $\hat h$) varies. 
In the circulating region $\alpha$ can be estimated from  the slopes in the spectral lattice shown in Figure \ref{fig:cc2}. 
The maximum and minimum energies in the spectrum 
are at $\sim \pm (|a| + |\epsilon|)$.  If the energy levels were equally spaced, they would be
separated by about $\frac{2(|a| + |\epsilon|)}{N}$.   Because the eigenstates are in nearly degenerate
pairs within  the separatrix, the distance between pairs of nearly degenerate 
eigenvalues is about twice this,  so we take $\Delta E \sim \frac{4(|a| + |\epsilon|)}{N}$. 
The parameter $b$ need only drift a distance of $2 \pi/N$ for an energy level 
to cover the distance between two levels initially at $b=0$ (see Figures \ref{fig:cc2} and \ref{fig:cc3}). 
This means that a single energy level drifts at a rate 
\begin{align}
\dot E & \sim  \frac{4(|a|+|\epsilon|) }{N} \frac{\dot b}{2 \pi/N}  \label{eqn:DotE}
\end{align} 
where $\dot b$ is the rate that $b$ varies.  
The relative drift rate between two crossing energy levels could be twice as large as $\dot E$ or 
\begin{align}
\alpha \approx \frac{4(|a| +| \epsilon|)}{\pi} \frac{\Delta b}{T},  \label{eqn:alpha}
\end{align} 
where $\Delta b$ is the total change of $b$ after a drift of duration  $T$ at a rate $\dot b$. 
We insert this into equation \ref{eqn:phalf} to estimate the distance 
between eigenvalues that gives a diabatic transition probability of 1/2 
\begin{align}
dE_{min,p1/2} 
& \approx 2 \sqrt{ \frac{ \ln 2}{ \pi^2}   (|a| +| \epsilon|)  \frac{\Delta b}{T/\hbar} }.
\label{eqn:de_min_half}
\end{align} 

We used equation \ref{eqn:de_min_half} to estimate the energy level spacing 
giving probability of 1/2 for diabatic transition at the different drift rates 
(of $b$) shown in Figure \ref{fig:drift}. 


\section{Comparing the eigenvalues of the Harper model to those of the Mathieu equation}
\label{ap:mathieu}

In this section we compare the eigenvalues of the Harper operator, with phase space  
 confined to a torus, to 
those of the quantized pendulum which is confined in angle $\phi$ but has momentum $p$ that
can go to infinity. The quantized pendulum is described by continuous wave functions $\psi(\phi)$ 
in an infinite dimensional Hilbert space, whereas state vectors for the Harper model are finite dimensional.  

Schr\"odinger's equation for the quantized pendulum \citep{Condon_1928}
\begin{align}
\left(\frac{\hat p^2}{2 I} - V_0 \cos \phi\right)\psi(\phi)  &= 
\left(-\frac{ \hbar^2}{2 I}  \frac{d^2 }{d\phi^2} - V_0 \cos \phi\right) \psi(\phi)\nonumber \\
& = E_\text{Sch}\psi(\phi) \label{eqn:Sch}
\end{align}
with moment of inertia $I$ and potential strength $V_0$,  
can be put in the form of the Mathieu equation with $z=\phi/2$. 
The Mathieu equation depends on a perturbation strength $q$, 
(usually taken to be real and non-negative) 
\begin{align}
\frac{d^2\psi(z)}{dz^2} + \left[\lambda - 2q \cos (2z)\right]\psi(z) = 0, 
\end{align}
where $\lambda$ is a real eigenvalue (also called a characteristic). 
Following \citet{Doncheski_2003},  
the eigenfunctions for the pendulum must exhibit $2 \pi$ periodicity in the variable
$\phi$, consequently we restrict eigenfunctions to those with  period $\pi$ in the variable $z$.  
The associated eigenvalues are those denoted $a_{2m}(q)$, for $m \ge 0$,  
with eigenfunction that is the even or cosine-like Mathieu function of the first kind 
$ce_{2m}(z,q)$ and eigenvalues denoted $b_{2m}(q)$, for $m >0$,  
with odd eigenfunctions that are odd or sine-like Mathieu functions of the first kind; $se_{2m}(z,q)$. 

The values of $q$ and $\lambda$  in Mathieu's equation that are 
consistent with Schr\"odinger's equation  for the quantized pendulum (equation \ref{eqn:Sch})  
\begin{align}
 q = \frac{4 I V_0}{\hbar^2} \qquad  \lambda = \frac{8 I E_\text{Sch}}{\hbar^2}
  \label{eqn:qdon}
 \end{align}
  (equation 13 by \citet{Doncheski_2003}).  

The number of quantum states that fits within a torus 
within phase space depends on Planck's constant. 
The Harper model was quantized with the effective value of Planck's constant 
$\hbar = \frac{2 \pi}{N}$ (following \citet{Quillen_2025} section II.A). 
Using this value of Planck's constant, and equation \ref{eqn:qdon} we relate 
the parameters $a,\epsilon$ of the Harper model 
 to the parameter $q$ in Mathieu's equation 
\begin{align}
q = \frac{\epsilon}{a} \left( \frac{N}{\pi} \right)^2. \label{eqn:scale_q} 
\end{align}
In the limit of low momentum the Harper equation has kinetic term $a p^2/2 $ 
matches a momentum of inertia $I=1/a$ in the associated quantized pendulum of equation \ref{eqn:Sch}. 
Again using $\hbar = \frac{2 \pi}{N}$ (following \citet{Quillen_2025} section II.A) for the quantization 
of the Harper model, and equation \ref{eqn:qdon}
we relate an eigenvalue $E$ of the Harper model to 
the eigenvalue $\lambda$ of the Mathieu equation 
\begin{align}
E &= \frac{\lambda}{2a} \left( \frac{ \pi }{N}\right)^2  + {\rm constant}. \label{eqn:scale_E}
\end{align}

In Figure \ref{fig:Mathieu} we compare the eigenvalues of the Harper operator
$\hat h(a,b,\epsilon) = a \cos (\hat p-b) + \epsilon \cos \phi $ 
with $a=1, b=0$,  $\epsilon = 0.4$ and dimension $N=50$ to the 
eigenvalues of the Mathieu equation that are computed using the scaling 
factors in equations
\ref{eqn:scale_q} and \ref{eqn:scale_E}.   
The characteristics $a_{2m}(q)$ (for $m \ge 0$) and  $b_{2m}(q)$ (for $m >0$) 
were computed with special functions available in python's \texttt{scipy} package. 
Both sets of eigenvalues were put in increasing order for this comparison, and the index $j$ 
on the $x$ axes in Figure \ref{fig:Mathieu} refer to the indices in this order, not $m$. 
After rescaling (with equation \ref{eqn:scale_E}), 
we shifted the spectrum of the Mathieu equation so that the lowest
eigenvalue matched that of the Harper model.  

Figure \ref{fig:Mathieu}a illustrates that within the potential well 
(at energies below that of the lower separatrix in the Harper model), 
the eigenvalues for both models are well separated and similar. 
Above the separatrix and in the region we called circulating, 
the eigenvalues in both systems are nearly degenerate.   The eigenvalues of the Mathieu equation 
diverge as the spectrum approaches that of a free rotor, whereas that of the Harper model has a 
symmetry which gives it an additional libration region at higher energies. 
The difference in behavior is because 
the kinetic energy for the pendulum is $\propto \hat p^2$ where as that in the Harper model
depends on $ \cos \hat p$. 

Figure \ref{fig:Mathieu}a shows that the spacing between eigenvalues at energy levels 
below 0 is similar.  While the Harper model has the smallest spacing in the middle of the circulating
region at an energy near 0, the spacing between pairs of nearly degenerate eigenvalues 
continues to drop rapidly as the index $j$ (or $m$) goes to infinity.  
Eigenvalues for the Mathieu equation are distinct unless $q=0$ (a result known as Ince's theorem) 
where as those of the Harper model are probably distinct unless $\epsilon=0$ or for the special 
case of the pair of zero eigenvalues when $N$ is a multiple of 4. 
Because the Harper model is finite dimensional, there is a minimum spacing between 
eigenvalues in the spectrum.  For a drifting system (for example with  $b$ slowly varying) 
there would be a drift rate that is slow enough that  a system begun in any 
eigenstate would remain in an eigenstate.   The quantized pendulum, modified so that 
it has kinetic energy $\propto (\hat p - b)^2$, can be similarly drifted by varying $b$.  
However the separation between pairs of 
eigenstates approaches zero as $m \to \infty$ for the quantized pendulum. 
No matter how slowly the quantized pendulum  is drifted, 
 there is always an energy above which 
avoided energy level crossings would be diabatic.   The drop in distances 
between energy levels is so rapid that for most laboratory settings 
the transition between adiabatic and diabatic behavior would not occur at 
a very high energy level.  

\begin{figure}[htbp]\centering
\includegraphics[width=1.67 truein]{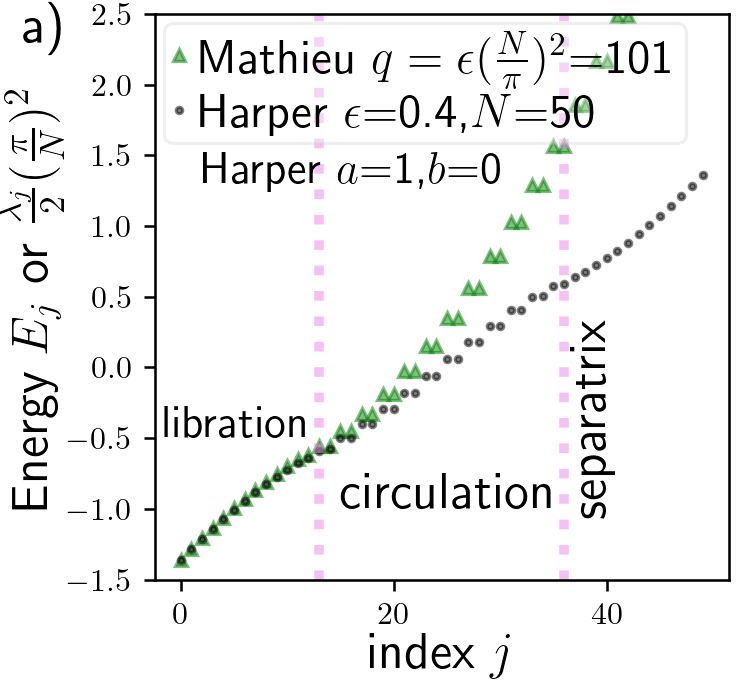}
\includegraphics[width=1.67 truein]{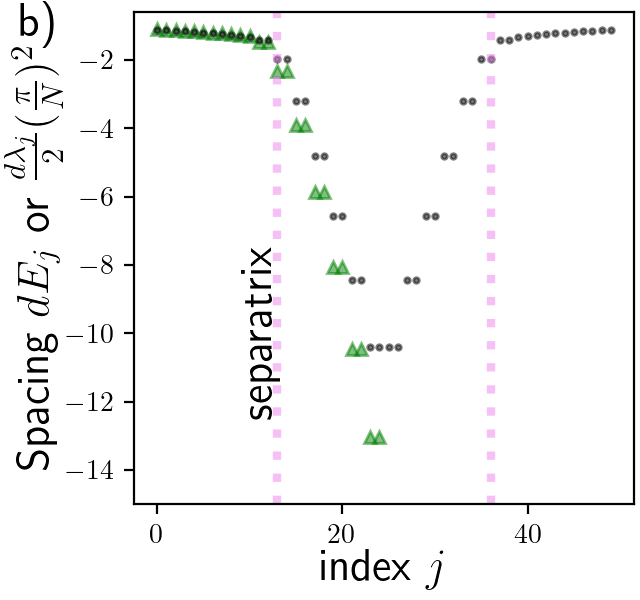}
\caption{a) We plot eigenvalues of a Harper operator (equation \ref{eqn:h0abe}) 
and those of a matching Mathieu equation. 
Scaling factors relating the two models are in equations \ref{eqn:scale_q} and \ref{eqn:scale_E}.   
The eigenvalues of the Mathieu equation match the quantized Harper model
within the libration region but diverge from the Harper model within the circulation region. 
Both systems have well separated eigenvalues 
in the libration region and nearly degenerate eigenvalues within the circulation region. 
b) Spacing between eigenvalues is plotted for both models.  For the Mathieu equation, 
as $j \to \infty$, the spacing between nearly degenerate pairs of eigenvalues 
continually decreases; $dE_j \to 0$.  
\label{fig:Mathieu}}
\end{figure}

Figure \ref{fig:Mathieu} suggests that a quantum system that resembles the pendulum 
can be approximated by a finite dimensional quantized Harper model. 
To compare the quantized Harper and pendulum operators, we compute the ratio of the energy level spacing 
at the bottom of the potential well to the width of the cosine potential for both systems.
For the Harper model this ratio  is 
\begin{align}
\frac{\Delta E}{2 \epsilon} &\sim\frac{ \hbar \omega_0}{2 \epsilon} \sim \sqrt{\frac{a}{\epsilon}} \frac{\pi}{N} 
\label{eqn:C6}
\end{align}
where we have used equation \ref{eqn:omega_0} for the libration frequency and $\hbar = 2\pi/N$
resulting from quantization.  
For the pendulum of equation \ref{eqn:Sch}, the dimensionless ratio is 
\begin{align}
\frac{\Delta E_\text{Sch}}{2 V_0} & \sim \hbar \frac{\sqrt{V_0/I} }{2 V_0} = 
\frac{\hbar}{\sqrt{IV_0}}.
\end{align}
The parameters of the Harper model can be chosen so that 
\begin{align} 
\frac{\Delta E}{2 \epsilon} &\approx   \frac{\Delta E_\text{Sch}}{2 V_0}  \to 
 \sqrt{\frac{a}{\epsilon}} \frac{\pi}{N}  = \frac{\hbar}{\sqrt{IV_0}}. \label{eqn:limit}
\end{align}  

To improve the approximation of the Harper model to the pendulum, $N$ can be increased, while 
maintaining the ratio of energy spacing to 
potential width so that it matches that of the pendulum.  In this limit,  the volume of the phase 
space torus of the Harper model increases with the number of states $N$, while the 
number of bound states within the potential well remains fixed. 
The resonant behavior remains at low momentum 
where the kinetic term which is proportional to $(1 - \cos \hat p) \sim \hat  p^2/2$. 
The larger the number of states, the better the kinetic energy of the Harper model 
approximates that of the pendulum in the lower energy phase space region near and 
containing librating states. 
To approximate the quantum pendulum with a Harper model with large $N$ 
equations \ref{eqn:C6} and \ref{eqn:limit} imply that different values of 
$N$ can be used  as long as $\epsilon N^2/a$ is chosen so that it  
 maintains the ratio of energy spacing to 
potential width in the potential well of the associated pendulum model,  via equation \ref{eqn:limit}.    
This condition is consistent with maintaining 
the value of $q$ (see equation \ref{eqn:scale_q}) for the associated Mathieu equation.  

\end{document}